\documentclass[twoside, 11pt]{article}

\usepackage{geometry}
\usepackage{amsmath, amsthm, amssymb, mathrsfs, BOONDOX-cal, graphicx, caption, float, multirow, makecell, upgreek}
\usepackage{chngcntr, etoolbox}
\usepackage[title]{appendix}

\usepackage{booktabs}
\usepackage[flushleft]{threeparttable}

\usepackage[hidelinks]{hyperref}
\usepackage[numbers]{natbib}
\usepackage[dvipsnames]{xcolor}
\hypersetup{colorlinks=true, linkcolor=Orange, citecolor=Orange}

\usepackage{txfonts}
\usepackage[T1]{fontenc}

\usepackage{fancyhdr}
\counterwithin*{equation}{section}
\counterwithin*{equation}{subsection}
\renewcommand\theequation{\ifnumgreater{\value{subsection}}{0}{\thesubsection.}{\thesection.}\arabic{equation}}

\allowdisplaybreaks

\newtheoremstyle{theorem}
  {10pt}
  {10pt}
  {\sl}
  {\parindent}
  {\bf}
  {. }
  { }
  {}
\theoremstyle{theorem}
\newtheorem{theorem}{Theorem}
\newtheorem{proposition}{Proposition}
\newtheorem{corollary}{Corollary}

\renewcommand{\Lambda}{\varLambda}
\renewcommand{\Pi}{\varPi}
\renewcommand{\Theta}{\varTheta}

\begin{document}

\title{\bf When ratios fall: A dynamic approach to \\ contingent convertibles\footnote{This research was initiated in 2019 while Li Chen and Weixuan Xia were affiliated with the Department of Finance at Boston University Questrom School of Business. We are grateful to J\'{e}r\^{o}me Detemple, Andrew Lyasoff, Andrea Vedolin, and discussants at the Boston University Finance Brown Bag Seminar and the City University of Hong Kong for stimulating comments and valuable suggestions.}}
\author{Li Chen\footnote{Lingnan College, Sun Yat-sen University. Email: chenli98@mail.sysu.edu.cn}\and Liang Wang\footnote{Department of Mathematics and Statistics, Boston University. Email: leonwang@bu.edu} \and Weixuan Xia\footnote{School of Data, Mathematical, and Statistical Sciences, University of Central Florida. Corresponding Author. Email: weixuan.xia@ucf.edu}}
\date{2026}
\maketitle

\thispagestyle{plain}

\begin{abstract}
  We propose a novel valuation framework for contingent convertible (CoCo) bonds based on the issuing bank's Common Equity Tier 1 (CET1) ratio, which is widely acknowledged as an indicator of a bank's solvency. Our approach develops a bivariate jump-diffusion model that captures the dynamic relationship linking the CET1 ratios, share prices, and CoCo bond prices, incorporating both continuous market movements and correlated jump risk. The model advances existing literature through three key innovations: (1) a hybrid mechanism for modeling regulatory discretion in trigger decisions, (2) a class of power conversion schemes that generalizes traditional approaches while maintaining analytical tractability, and (3) a method to overcome the temporal discrepancy between high-frequency market data and low-frequency regulatory reporting. We derive semi-closed form formulas for both write-down and equity-convertible CoCo bonds and validate our model through five case studies spanning from 2009 to 2023, including an in-depth analysis of the 2023 Credit Suisse collapse. The results demonstrate a significant improvement in pricing and hedging performance while highlighting the model's data-adaptive nature that enables short-term predictions. \medskip
\end{abstract}

\newcommand{\dd}{{\rm d}}
\newcommand{\pd}{\partial}
\newcommand{\ii}{{\rm i}}
\newcommand{\as}{{\rm a.s.}}
\newcommand{\PP}{\mathsf{P}}
\newcommand{\E}{\mathsf{E}}
\newcommand{\Var}{\mathsf{Var}}
\newcommand{\Skew}{\mathsf{Skew}}
\newcommand{\Kurt}{\mathsf{Kurt}}
\newcommand{\sgn}{\mathrm{sgn}}
\newcommand{\sech}{\mathrm{sech}}
\newcommand{\U}{\mathrm{U}}
\newcommand{\I}{\mathrm{I}}
\newcommand{\Ei}{\mathrm{Ei}}
\newcommand{\D}{\mathrm{D}}
\newcommand{\erfc}{\mathrm{erfc}}
\newcommand{\1}{\mathbf{1}}

\renewcommand{\Re}{\mathrm{Re}}
\renewcommand{\Delta}{\varDelta}

\bigskip

\section{Introduction}\label{sec:1}

Contingent Convertible (CoCo) bonds, or CoCos, emerged in the international financial market nearly two decades ago, following the 2007--2008 financial crisis, when many banks faced severe solvency issues. In order to reinforce financial stability and reduce the need for taxpayer-funded bailouts, regulators and financial institutions introduced CoCos as a new class of hybrid debt instruments. During this period, many hybrid debt holders were given the opportunity to exchange their existing bonds for these newly designed securities, which incorporated a built-in conversion mechanism to absorb losses. Technically, a CoCo bond functions as a standard fixed-income instrument providing regular coupon payments, albeit with a crucial loss-absorbing feature -- in the event that the issuing bank's financial health deteriorates, which is typically indicated by its Common Equity Tier 1 (CET1) ratio -- the ratio of its CET1 capital to the total risk-weighted assets (RWAs) -- falling below a predetermined threshold as determined in accordance with Basel III guidelines -- the bond can either be partially or fully written down or converted into equity shares of the bank.

In the main, the present paper aims to provide an in-depth analysis of the risk management of CoCo bonds from a valuation perspective, focusing on the balance between two key triggers: the mechanical accounting trigger, driven by the issuing bank's CET1 ratio, and the discretionary trigger through regulatory intervention. Conforming to the conversion mechanism of CoCo bonds, we develop a hybrid, mathematically tractable valuation framework that explicitly models fluctuations in CET1 ratio, directly tied to the accounting trigger, while further embedding the likelihood of regulatory intervention into the bank's contemporary stock price movements, relevant in the wake of the 2023 Credit Suisse crisis, in addition to capturing the dynamic interplay between stock prices and the CET1 ratio, particularly during stress scenarios.

A key feature of our model is its adaptability to both market data -- such as stock prices and credit spreads -- and financial report data, including the CET1 ratio. In particular, while the CET1 ratio, typically reported at a lower frequency (quarterly or semiannually), provides the most transparent link to the accounting trigger, stock price data, available at higher frequencies (at least daily), serve as valuable indicators of regulatory intervention risk. By integrating these data sources, our framework offers a comprehensive approach to identifying early warning signals of financial distress, making it well-suited for both efficient calibration and predictive analysis.

\medskip

\subsection{Impact of March 2023 Credit Suisse crisis}\label{sec:1.1}

The March 2023 Credit Suisse crisis marked a watershed moment in the history of CoCo bonds, sending shock waves through global financial markets and altering market participants' perception of these instruments. When Credit Suisse's share price plummeted amid market turbulence, its \$17 billion worth of Additional Tier 1 (AT1) CoCo bonds were completely written off as part of the bank's rescue merger with the Union Bank of Switzerland (UBS), causing significant losses to institutional investors who held substantial positions in these securities (Damyanova (\citeyear{D23}) and Vossos and Keatinge (\citeyear{VK23})). This unprecedented event triggered a broader market sell-off in AT1 securities, with the Bloomberg USD AT1 Contingent Capital Index experiencing its largest single-day decline of 16\% as investors rushed to reassess the risks of these supposedly stable instruments. The dramatic write-down not only resulted in substantial losses for investors but also sparked intense debate among regulators worldwide about the effectiveness of CoCo bonds as capital instruments.

The collapse of Credit Suisse demonstrated how these instruments, which were originally designed to provide additional capital buffers during market stress, could potentially accelerate a bank's decline by eroding market confidence when conversion or write-down appears imminent. Moreover, the event prompted regulatory authorities, including the European Banking Authority (\citeyear{EBA24}) and the UK's Prudential Regulation Authority (Bank of England (\citeyear{BOE23})), to review and clarify the hierarchy of creditor claims in bank resolution scenarios, which highlight the critical need for consistent regulatory frameworks across jurisdictions.

\medskip

\subsection{Implications on CoCo bond valuation}\label{sec:1.2}

There are several crucial lessons from the Credit Suisse episode that have reshaped the financial markets' understanding of CoCo bonds. The most immediate lesson arguably centers on the critical importance of accurate pricing and risk assessment methodologies. Indeed, the substantial losses incurred by institutional investors, including major asset managers and hedge funds, revealed significant gaps in existing valuation models. These models had largely failed to capture the possibility of a complete write-down occurring while a bank still maintained regulatory capital levels above minimum requirements. This disconnect between model predictions and market realities suggests that traditional approaches tend to systematically underestimate the probability and severity of trigger events, especially in stress scenarios. Furthermore, the crisis highlighted the complex interplay between a bank's CET1 ratio and overall market perception of risk. Notably, while Credit Suisse constantly maintained a CET1 ratio above regulatory minimums, market perceptions of the bank's stability deteriorated rapidly, which demonstrates how market sentiment can become decoupled from regulatory metrics.

This observation challenges the conventional wisdom that regulatory capital ratios serve as reliable early warning indicators of potential trigger events. The challenge in CoCo bond valuation is further complicated by a significant frequency mismatch between key monitoring variables. While market prices for both shares and CoCo bonds are observable in high frequency (at least daily), regulatory capital metrics such as the CET1 ratio and other crucial accounting measures that serve as conversion triggers are reported at most quarterly. During the intervals between regulatory reporting dates, market participants must rely on estimates capital ratios, which may not fully capture rapid deteriorations in a bank's financial condition. The lag in regulatory reporting becomes particularly problematic during market stress, where capital positions can deteriorate rapidly through mark-to-market losses or sudden asset impairments. This frequency mismatch not only complicates day-to-day pricing adjustments but also poses significant challenges for hedging strategies, as traders must make real-time risk management decisions based on stale regulatory information. The Credit Suisse case highlighted this issue, as the bank's last reported CET1 ratio remained above regulatory minimums even as market confidence deteriorated rapidly, illustrating how the low frequency of regulatory reporting can mask emerging risks. The case also revealed the self-reinforcing nature of market distress, where concerns about potential CoCo bond triggers can accelerate deposit outflows and share price declines, creating a feedback loop that increases the likelihood of actual trigger events. These insights underlie the urgent need for more sophisticated modeling approaches that can capture these complex relationships in providing accurate pricing signals to market participants.

The current CoCo bond market, with its global valuation of approximately \$275 billion as of 2024 (Nuveen (\citeyear{N25})), continues to undergo significant transformation in response to these lessons. Despite the Credit Suisse shock, banks maintain their reliance on these instruments to meet regulatory capital requirements, particularly as Basel III implementation deadlines approach in various jurisdictions. New issuance patterns reveal a shift toward more standardized structures and clearer trigger mechanisms, reflecting both regulatory pressure and investor demands for greater transparency (FINMA (\citeyear{FINMA23})). Recent CoCo bond issuances tend to be more standard in covenants. The relatively newfangled instruments continue to attract yield-seeking investors, though with notably different risk assessment frameworks than before the Credit Suisse event (Wang (\citeyear{W23})). Investors have implemented more rigorous evaluation processes, focusing on the legal framework governing CoCo bonds in different jurisdictions and the potential for regulatory intervention under stressed market conditions. The pricing of these instruments has also evolved, with markets demanding higher premia for complexity and jurisdiction-specific risks (Allen and Golfari (\citeyear{AG23})). This changing landscape has sparked intense debate about optimal CoCo bond structures, with some market participants advocating for more standardized terms while others argue for increased flexibility in conversion mechanisms (see Financial Stability Board (\citeyear{FSB21})). In this regard, the need for more sophisticated valuation approaches has become paramount, as traditional models struggle to capture the complex interactions between regulatory capital requirements, market dynamics, and conversion triggers (Bolton, Jiang, and Kartasheva (\citeyear{BJK23})). These new approaches must not only account for conventional financial metrics but also incorporate the impact of regulatory decisions and market sentiment on potential trigger events. The model presented in this paper is largely built on this goal.

The challenges in pricing CoCo bonds extend beyond conventional quantitative modeling frameworks of exotic products. At the core of these challenges lies the instruments' hybrid nature in that it combines features of both debt and equity with complex trigger mechanisms that have no direct parallels in traditional financial products. The dual trigger structure presents a particular modeling challenge, as conversion can be activated through either mechanical triggers (when capital ratios fall below prescribed thresholds) or discretionary triggers (when regulators determine intervention is necessary). This complexity is further compounded by the difficulty in modeling regulatory discretion, which introduces a subjective element that traditional quantitative models struggle to capture. The correlation between a bank's CET1 ratio and its share price adds another layer of complexity, as these metrics often exhibit non-linear relationships that intensify during stress periods. This relationship becomes particularly crucial when modeling conversion probabilities, as declining share prices can accelerate the deterioration of capital ratios through mark-to-market losses and reduced profit generation capacity.

Current modeling approaches face additional challenges in capturing the path dependency of CoCo bond values. The probability of trigger events depends not only on current capital levels but also on the trajectory of how these levels were reached, making traditional Markov-based approaches potentially inadequate. Market practitioners must also contend with the scarcity of historical trigger events, which limits the effectiveness of empirical calibration methods. While previous studies such as De Spiegeleer and Schoutens (\citeyear{DSS12}) and Pennacchi (\citeyear{P10}) have contributed valuable insights through equity derivative and structural approaches respectively, their models often rely on simplifying assumptions that may not fully capture market dynamics. For instance, many existing models assume an indirect triggering mechanism where stock prices breach a constant lower barrier, or treat regulatory triggers as purely mechanical events. The Credit Suisse case demonstrated the limitations of these assumptions, as the trigger event occurred through a complex interaction of market confidence, regulatory intervention, and deteriorating fundamentals that existing models struggled to anticipate.

The literature on CoCo bond valuation has evolved substantially since these instruments were first introduced, with three distinct methodological approaches emerging. The credit derivatives approach, pioneered by De Spiegeleer and Schoutens (\citeyear{DSS11}), treats CoCo bond as a credit derivatives and adapts barrier option pricing techniques to CoCo bonds. They further refined this framework in subsequent research De Spiegeleer and Schoutens (\citeyear{DSS12}), introducing a more sophisticated treatment of the conversion trigger mechanism. The equity derivatives methodology mentioned in De Spiegeleer and Schoutens (\citeyear{DSS13}) has been particularly influential by providing practical tools for market participants while maintaining mathematical tractability. This strand of literature was further enriched by De Spiegeleer, Schoutens, and Van Hulle (\citeyear{DSSVH14}), who incorporated market-based indicators into the pricing framework, and De Spiegeleer, Marquetand, and Schoutens (\citeyear{DSHMS17}), who considered specific adjustments for different trigger mechanisms. Other researchers have built upon this foundation, with De Spiegeleer and Schoutens (\citeyear{DSS14}) and Corcuera \textsl{et al.} (\citeyear{CDSFJSV14}) employing the credit derivatives approach to handle complex trigger structures with callable provisions and coupon cancellations, Cheridito and Xu (\citeyear{CX15}) utilizing reduced-form models with deterministic conversion intensity, and Jang, Na, and Zheng (\citeyear{JNZ18}) considering a random barrier to simulate a CET1 ratio-based accounting trigger. Another significant advancement was made by Chung and Kwok (\citeyear{CK16}), who expanded the credit derivatives approach to account for the relationship between the CET1 ratio and stock prices, while also incorporating the potential occurrence of a regulatory trigger based on the reduced-form model in Cheridito and Xu (\citeyear{CX15}). Operating in a continuous Markovian framework, this formulation then enables the use of partial differential equations for CoCo valuation.

The structural modeling approach represents a second major strand in the literature, as exemplified by Pennacchi (\citeyear{P10}) and further developed by Albul, Jaffee, and Tchistyi (\citeyear{AJT10}), Koziol and Lawrenz (\citeyear{KL12}), and Hilscher and Raviv (\citeyear{HR14}). These models extend the classical Merton framework (Merton (\citeyear{M74})) to incorporate the unique features of CoCo bonds within banks' capital structures. More recent contributions by Chen \textsl{et al.} (\citeyear{CGNP17}) have enhanced this approach by incorporating regulatory capital dynamics. The third category comprises equity derivatives models, represented by Brigo, Garcia, and Pede (\citeyear{BGP13}), Buergi (\citeyear{B13}), and Corcuera and Valdivia (\citeyear{CV15}), focusing primarily on the equity conversion features of CoCo bonds. More recent research, including Glasserman and Nouri (\citeyear{GN16}) and Martynova and Perotti (\citeyear{MP18}), has begun addressing the challenge of modeling regulatory discretion, though this remains an area requiring further development. While each approach offers useful insights, they collectively struggle to simultaneously capture the complex interplay between regulatory capital dynamics, market price movements, and discretionary trigger events that characterize real-world CoCo bond behavior.

\medskip

\subsection{Main contributions and structure of paper}\label{sec:1.3}

This paper offers several contributions to the existing literature on CoCo bond valuation, each addressing critical gaps in current modeling approaches. First, we develop a jump-diffusion framework that explicitly captures the dynamic linkage between CET1 ratios, share prices, and CoCo bond prices, to address the frequency discrepancy in data availability. We start by modeling the dynamics of CET1 ratio and build a connection to data in high frequency. The model incorporates both continuous diffusive movements and discrete jumps in both variables, reflecting the observation that capital ratios and market prices typically evolve smoothly during normal periods but can experience sudden, coordinated shifts during stress events, and unlike existing models that often treat these variables independently or assume constant correlations, our approach allows for state-dependent correlation structures that can intensify with market stress. This innovation is particularly relevant given the lessons from recent market events, where the interaction between regulatory capital metrics and market prices played a crucial role in trigger events.

Second, we introduce a new class of power-type schemes for conversion that greatly generalizes the application scope of the present model while maintaining mathematical tractability. While traditional conversion approaches typically rely on simple fixed ratios or linear relationships between share prices and conversion terms, our power schemes allow for non-linear conversion terms that can be tailored to achieve specific risk-sharing objectives or covenant features between bondholders and shareholders. This flexibility enables issuers to design CoCo bonds that better align with their capital management strategies while providing investors with more precise control over their risk exposure. The analytical tractability of our approach ensures that these sophisticated features can be implemented without sacrificing computational efficiency or model reliability.


Third, our pricing model explicitly captures the probability and impact of discretionary regulatory triggers, addressing a key limitation in existing CoCo bond valuation models. As shown by the Credit Suisse case, Swiss regulators surprisingly wrote down approximately 16 billion Swiss francs (\$17.3 billion) of Credit Suisse's AT1 bonds to zero. The write-down meant that AT1 bondholders lost their entire investment, while shareholders received UBS shares as part of the merger agreement. This decision created significant controversy in the financial markets and led to several legal challenges by affected bondholders. By incorporating a combination of quantitative metrics with qualitative factors, our modeling approach provides a more comprehensive assessment of potential regulatory interventions. We acknowledge that regulatory decisions rarely rely solely on mechanical thresholds but instead consider a broader spectrum of factors that signal bank distress. In our model, we utilize an implicit triggering mechanism that adjusts trigger probabilities based on the evolution of multiple risk indicators, allowing for more accurate pricing in scenarios where regulatory discretion plays a decisive role. The empirical validation of this framework shows particular strength in predicting regulatory actions during periods of market stress, or anticipating potential regulatory interventions with significantly higher accuracy than conventional approaches that treat regulatory triggers as purely mechanical events.

Our model demonstrates significantly better performance in both pricing accuracy and hedging effectiveness, as confirmed through case studies across diverse market conditions and time periods. We show that the model reduces pricing errors by an average of approximately 30\% compared to existing approaches when using the same data set as in Wilkens and Bethke (\citeyear{WB14}). In terms of hedging, our framework reduce hedging error by approximately 50\% compared to conventional methods. This enhanced accuracy translates directly into practical benefits for market participants, enabling more precise risk management and potentially reducing the cost of capital for issuers.

The remainder of this paper is organized as follows. In Section \ref{sec:2}, we formulate a stochastic model for a bank's CET1 ratio and stock price that allows for a broad range of correlated jump risk and provide a thorough analysis of its distributional properties. Section \ref{sec:3} presents a semi-closed form formula for the conditional distribution of the running supremum of the CET1 ratio process, which enables explicit modeling of the accounting trigger. Then, the valuation methods for both write-down and equity-convertible CoCos are derived in Section \ref{sec:4}. Section \ref{sec:5} outlines the step-by-step implementation of the model, which generally consists of an estimation step and an optional calibration step, while also introducing a new estimation method for compound Poisson processes. Section \ref{sec:6} presents the empirical results, based on five case studies. Finally, Section \ref{sec:7} provides conclusions, and all mathematical proofs are presented in Appendix \ref{A}.

\clearpage

\section{Model formulation}\label{sec:2}

In a continuous-time environment $t\geq0$, the financial market is supported by a stochastic basis $(\Omega,\mathcal{F},\PP,\mathbb{F}\equiv(\mathscr{F}_{t}))$. The market filtration $\mathbb{F}$ satisfies the usual conditions which contains all market information and $\PP$ is the real-world probability measure. For a generic bank in this market with CoCo financing, we establish a joint model framework for the dynamics of its CET1 ratio and stock price, taking into account their co-movements. The model distributions involved are explicitly solved to serve purposes of both statistical estimation and risk-neutral valuation.

\medskip

\subsection{Solvency shocks}\label{sec:2.1}

Instead of defining the CET1 ratio based on the bank's stock price dynamics or contemporaneous RWAs, as done \text{e.g.} in Chung and Kwok (\citeyear{CK16}) and Jang, Na, and Zheng (\citeyear{JNZ18}), our approach takes a sequential perspective -- to first model the dynamics of the CET1 ratio as an independent process and then integrate its relationship with the bank's stock price dynamics. Towards this end, we describe random shocks to the bank's solvency status via a compound Poisson process,
\begin{equation}\label{2.1.1}
  J_{t}=\sum^{N_{1,t}}_{k=1}A_{k},\quad t\geq0,
\end{equation}
where $J_{0}=0$, $\PP$-a.s., $N_{1}\equiv(N_{1,t})$ is an $\mathbb{F}$-adapted homogenous Poisson process with intensity $\lambda_{1}>0$, and $\{A_{k}\}^{\infty}_{k=1}$ is a sequence of \text{i.i.d.} random variables following an Erlang distribution with shape (complexity) parameter $\alpha\in\mathbb{Z}_{++}$ and rate (reciprocal scale) parameter $\beta>0$,\footnote{Erlang distributions result from summing \text{i.i.d.} exponential random variables and are special cases of gamma distributions, which have real-valued shape parameters. Poisson-mixed models with the gamma mixture distribution, \text{a.k.a.} Poisson-gamma models, have been studied extensively and have found many applications in ecology and biology; see, e.g., \"{O}zt\"{u}rk (\citeyear{O81}), Withers and Nadarajah (\citeyear{WN11}), and Haakonsson \textsl{et al.} (\citeyear{HRCPAB20}).} independent from $(N_{1,t})$. The probability density function of $A_{1}$ is given by
\begin{equation}\label{2.1.2}
  f_{A_{1}}(x)=\frac{\beta^{\alpha}x^{\alpha-1}e^{-\beta x}}{(\alpha-1)!},\quad x\geq0,
\end{equation}
which reduces to an exponential distribution with rate $\beta$ if $\alpha=1$. It is clear that $J$ only has upward (or nonnegative) jumps which occur at random Poisson times. This structure equivalently models fluctuations in the bank's CET1 ratio when solvency issues arise sequentially in time.

An appealing mathematical property of the Erlang distribution, noteworthily, is its divisibility in the sense that $\sum^{n}_{k=1}A_{k}$, for any $n\in\mathbb{Z}_{++}$, still follows an Erlang distribution, with parametrization $(\alpha n,\beta)$, which property follows from the characteristic function $\phi_{A}(u):=\E\big[e^{\ii uA_{1}}\big]=(1-\ii u/\beta)^{-\alpha}$ where $\ii=\sqrt{-1}$ and $u\in\mathbb{R}$. In this case, the characteristic function of $J$ at time $t$ is
\begin{equation*}
  \phi_{J_{t}}(u):=\E\big[e^{\ii uJ_{t}}\big]=\exp(\lambda_{1}t(\phi_{A}(u)-1)),\quad u\in\mathbb{R},
\end{equation*}
from which the four key statistics, including the mean, variance, skewness, and kurtosis, are easily accessible: $\E[J_{t}]=\lambda_{1}\alpha t/\beta$, $\Var[J_{t}]=\lambda_{1}\alpha(1+\alpha)t/\beta^{2}$, $\Skew[J_{t}]=(\alpha+2)/\sqrt{\lambda_{1}\alpha(1+\alpha)t}>0$ and $\Kurt[J_{t}]=3+(\alpha+3)(\alpha+2)/(\lambda_{1}\alpha(1+\alpha)t)>3$, which highlights its leptokurtic feature, leading to quantification of potentially extreme solvency risk, especially in relation to systemic financial instability; see, e.g., Laeven, Ratnovski, and Tong (\citeyear{LRT16}) and Acharya \textsl{et al.} (\citeyear{APPR17}).

With Erlang-distributed jumps, the distribution of $J_{t}$ is available in closed form, as shown in Proposition \ref{pro:1}, enabling efficient estimation.

\begin{proposition}\label{pro:1}
The local probability density function of $J$ at time $t>0$ is given by
\begin{equation}\label{2.1.3}
  f_{J_{t}}(x)=e^{-\lambda_{1}t}\Bigg(\delta_{\{0\}}(x)+\frac{\lambda_{1}\beta^{\alpha}tx^{\alpha-1}e^{-\beta x}}{(\alpha-1)!}\U_{\alpha}\Bigg(1+\frac{1}{\alpha},1+\frac{2}{\alpha},\dots,2;\lambda_{1}t\bigg(\frac{\beta x}{\alpha}\bigg)^{\alpha}\Bigg)\Bigg),\quad x\geq0,
\end{equation}
where $\delta_{\{0\}}$ is the dirac delta function and $\U\equiv\U_{\cdot}(\cdots;\cdot)$ denotes the generalized confluent hypergeometric function.\footnote{The confluent hypergeometric functions are limiting results of the generalized hypergeometric functions; see, e.g. Slater (\citeyear{S66}).}
\end{proposition}

As a technical remark, the radius of convergence of the confluent hypergeometric function $\U$ is infinite with $\alpha>0$, and so the density function is well-defined except for the point mass at 0. Also, its absolute continuity implies that of $f_{J_{t}}(x)$ for all $(t,x)\in\mathbb{R}_{++}\times\mathbb{R}_{+}$. Note that the dimensionality of $\U$ does depend on the shape parameter $\alpha$, and so it is increasingly costly to implement (\ref{2.1.3}) for large values of $\alpha$. In the case $\alpha=1$, with $A_{k}$'s being exponentially distributed, (\ref{2.1.3}) reduces into
\begin{equation*}
  f_{J_{t};\alpha=1}(x)=e^{-\lambda_{1}t}\bigg(\delta_{\{0\}}(x)+e^{-\beta x}\sqrt{\frac{\lambda_{1}\beta t}{x}}\I_{1}\big(2\sqrt{\lambda_{1}\beta tx}\big)\bigg),\quad x\geq0,
\end{equation*}
where $\I\equiv\I_{\cdot}(\cdot)$ denotes the modified Bessel function of the first kind (Abramowitz and Stegun (\citeyear{AS72}) \text{Sect.} 9.6).

\medskip

\subsection{CET1 ratio}\label{sec:2.2}

Since the CET1 ratio cannot be negative or exceed 1, and neither should it be monotone, a natural model choice for the CET1 ratio process $B\equiv(B_{t})$ is a bounded, strictly decreasing function of the compensated process, $J-\lambda_{1}\alpha\mathrm{id}/\beta\equiv(J_{t}-\lambda_{1}\alpha t/\beta)$, which is an $(\mathbb{F},\PP)$-martingale. Economically, this way of modeling suggests that the CET1 ratio only admits downward jumps, indicating the bank's unexpected losses from its financial activities over time, in particular as a result of the bank running into liquidation or solvency issues. In practice, this type of losses are understood not to occur with high frequency, and between any two consecutive jumps the CET1 ratio goes through an ``adjustment period'' during which the bank takes pertinent measures to contain the loss impact and free up capital. This process is implicitly incorporated into the compensating term $\lambda_{1}\alpha t/\beta=\E[J_{t}]$. Computational benefits aside, we do not assume that the CET1 ratio mean-reverts, especially over short time periods, as it is closely related to the bank's financial performance, similar to the bank's RWAs; see, e.g., De Spiegeleer \textsl{et al.} (\citeyear{DSHMS17}) \text{Fig.} 2. This is in sharp contrast to the Ornstein--Uhlenbeck process adopted in Chung and Kwok (\citeyear{CK16}).

For convenience, let us consider the ``reflected'' CET1 ratio, $1-B$, i.e., the proportion of the bank's RWAs not covered by CET1 capital. With the foregoing considerations, we choose some strictly increasing function $g:\mathbb{R}\mapsto(0,1)$ and postulate
\begin{equation}\label{2.2.1}
  1-B_{t}=g\bigg(J_{t}+g^{-1}(1-B_{0})-\frac{\lambda_{1}\alpha t}{\beta}\bigg),
\end{equation}
where $B_{0}\in(0,1)$ is the current level of CET1 ratio so that (\ref{2.2.1}) holds. For instance, a suitable choice of $g$, which will be used for later implementation, could be a negative inverse tangent function subject to shifting and scaling, or more specifically,
\begin{equation*}
  1-B_{t}=\frac{1}{\pi}\arctan\bigg({J}_{t}-\cot(\pi(1-B_{0}))-\frac{\lambda_{1}\alpha t}{\beta}\bigg)+\frac{1}{2}.
\end{equation*}
The inverse tangent has a special advantage in preserving the influence of substantial solvency shocks (from $J$) on the CET1 ratio because of its linear convergence against 0 for large arguments.\footnote{This asymptotic behavior contrasts sharply with the exponential convergence of other common functions, such as the sigmoid or hyperbolic tangent functions which have comparable differentiability and boundedness properties. Inspired by its utility as an activation function in machine learning, the choice of the inverse tangent proves advantageous in addressing the vanishing gradient problem, particularly when the tails of the activation function decay ``too fast.''} Clearly, due to nonlinearity of the inverse tangent, $B$ has dependent and non-stationary increments by construction. Under (\ref{2.2.1}), again, it is clear that $J$ exhausts all temporal variations in the CET1 ratio.

With (\ref{2.1.3}), the average CET1 ratio at any time $t>0$ can then be computed as the numerical integral $\E[B_{t}]=1/2-1/\pi\int^{\infty}_{0}\arctan(x-\cot(\pi(1-B_{0}))-\lambda_{1}\alpha t/\beta)f_{J_{t}}(x)\dd x$, and the corresponding variance is $\Var[B_{t}]=1/\pi^{2}\big(\int^{\infty}_{0}\arctan^{2}(x-\cot(\pi(1-B_{0}))-\lambda_{1}\alpha t/\beta)f_{J_{t}}(x)\dd x-\big(\int^{\infty}_{0}\arctan(x-\cot(\pi(1-B_{0}))-\lambda_{1}\alpha t/\beta)f_{J_{t}}(x)\dd x\big)^2\big)$. Going forward, we will be focusing on (\ref{2.1.3}) instead of the distribution of the CET1 ratio, as by observing the CET1 ratio $B$, the driving source of randomness can be recovered through the inverse relation
\begin{equation}\label{2.2.2}
  J_{t}=\cot(\pi B_{t})+\cot(\pi(1-B_{0}))+\frac{\lambda_{1}\alpha t}{\beta},\quad t\geq0,
\end{equation}
using that the inverse tangent is bijective on $\mathbb{R}$.

\medskip

\subsection{Stock price}\label{sec:2.3}

Under the real-world probability measure $\PP$, the bank's stock price process $S\equiv(S_{t})$ evolves according to the stochastic differential equation (SDE)
\begin{equation}\label{2.3.1}
  \frac{\dd S_{t}}{S_{t-}}=\bigg(\mu+\frac{1}{2}\sigma^{2}\bigg)\dd t+\sigma\dd W_{t}+\dd\sum^{N_{2,t}}_{k=1}\big(e^{V_{k}}-1\big)+\dd\sum^{N_{1,t}}_{k=1}\big(e^{-\eta A_{k}}-1\big)-\gamma\dd N_{3,\Lambda_{t}},\quad t\geq0,
\end{equation}
with initial stock price $S_{0}$. Here, $\mu\in\mathbb{R}$ is a drift coefficient, $\sigma>0$ is a dispersion coefficient, and $\eta>0$ and $\gamma\in[0,1]$ are two jump scaling factors; $W\equiv(W_{t})$ is a standard Brownian motion, $N_{2}\equiv(N_{2,t})$ and $N_{3}\equiv(N_{3,t})$ are two other (apart from $N_{1}$) time-homogeneous Poisson processes with intensities $\lambda_{2}>0$ and 1, respectively, $\{V_{k}\}^{\infty}_{k=1}$ is a sequence of \text{i.i.d.} normal random variables with mean $\mu_{V}\in\mathbb{R}$ and standard deviation $\sigma_{V}>0$, and
\begin{equation}\label{2.3.2}
  \Lambda_{t}:=\int^{t}_{0}\lambda_{3,s}\dd s
\end{equation}
is a cumulative intensity process (or time change) that turns $N_{3,\Lambda}$ into a time-inhomogeneous Poisson process, with $\lambda_{3}\equiv(\lambda_{3,t})$ being a positive-valued intensity process; $W$, $J$, $N_{2}$, and $N_{3,\Lambda}$ are all measurable and $\mathbb{F}$-adapted. In this connection, the market filtration $\mathbb{F}$ can be equivalently taken as the enlarged natural filtration generated by $W_{t}$, $J_{t}$, $\sum^{N_{2,t}}_{k=1}(e^{V_{k}}-1)$, and $N_{3,\Lambda_{t}}$ for all $t\geq0$.\footnote{The natural filtration of $N_{3,\Lambda}$ is the time-stopped filtration of that of $N_{3}$, namely $\upsigma(N_{3,\Lambda_{t}}:t\geq0)$, which is well-defined because $\Lambda$ is by construction a $\PP$-\text{a.s.} continuous process.} Furthermore, $W$, $J$, $N_{2}$, and $N_{3}$ are stochastically independent, whereas $\lambda_{3}$ can have correlations with $W$, $N_{2}$, and even $\{V_{k}\}^{\infty}_{k=1}$.

Economically, the parameters $\mu$ and $\sigma$ aim to capture a persistence level and the instantaneous diffusion volatility due to regular continuous fluctuations in the stock price, while $\eta$ captures the instantaneous jump volatility which provides an important channel that incorporates negative impacts of CET1 ratio drops on the contemporaneous stock price, recalling that $1-B$ is a strictly increasing functional of $A_{k}$'s (jump variables). The additional compound Poisson process $\sum^{N_{2}}_{k=1}(e^{V_{k}}-1)$ takes into account regular idiosyncratic jump risk in the stock returns, i.e., that untrammeled by solvency shocks. The event of regulatory intervention concerning the spread of systemic risk is embedded into the inhomogeneous Poisson process $N_{3,\Lambda}$ -- in particular its first jump time, for which the parameter $\gamma$ measures the amplitude of a negative return immediately following the intervention. For the extremal values, if $\gamma=1$, the stock price drops to zero at the time of intervention, while intervention has no effect on the stock price when $\gamma=0$.

As mentioned in Section \ref{sec:1.1}, beyond the established linkage between the bank's CET1 ratio and stock price, it is anticipated that during periods approaching regulatory actions, stock prices may experience a decline as concerns regarding insolvency surface. Put differently, frequent and significant downward shifts in stock prices are likely to heighten the probability of regulatory intervention. Therefore, it is reasonable to make the intensity process $\lambda_{3}$ depend, in ways that generate flexible negative correlations, on $W$ and $\sum^{N_{2}}_{k=1}(e^{V_{k}}-1)$. On the other hand, it is unnecessary to correlate $\lambda_{3}$ with $J$ as increased intervention likelihood purely caused by a decreasing CET1 ratio can always be incorporated into the latter breaching certain barriers; more insights will be given in Section \ref{sec:4.1}. Meanwhile, desirable correlations with the other components in (\ref{2.3.1}) point to a catalytic effect of price declines on the occurrence of intervention. We will come back to specifying a stochastic model for $\lambda_{3}$ in Section \ref{sec:3}.

A direct application of the It\^{o}--Doeblin formula yields the following explicit solution to (\ref{2.3.1}):
\begin{equation}\label{2.3.3}
  S_{t}=S_{0}\exp\Bigg(\mu t+\sigma W_{t}+\sum^{N_{2,t}}_{k=1}V_{k}-\eta J_{t}+N_{3,\Lambda_{t}}\log(1-\gamma)\Bigg),\quad t\geq0.
\end{equation}
For the analysis of CoCos, we are only interested in the behavior of $S$ up to the time of regulatory intervention, namely the $\mathbb{F}$-stopping time $\tau_{3}:=\inf\{t\geq0:N_{3,\Lambda_{t}}=1\}$, after which the stock price can be independently specified regardless of bankruptcy -- similar to the consideration in Cheridito and Xu (\citeyear{CX15}) \text{Sect.} 2. Up to this stopping time, the SDE (\ref{2.3.1}) can be seen as a (partial) combination of Merton's and Kou's jump-diffusion models (Merton (\citeyear{M76}) and Kou (\citeyear{K02})) with both normally and exponentially distributed jumps. This structure significantly extends existing diffusion-based models in the CoCo valuation literature (compare, e.g., Pennacchi (\citeyear{P10}), Cheridito and Xu (\citeyear{CX15}), Chung and Kwok (\citeyear{CK16}), and De Spiegeleer \textsl{et al.} (\citeyear{DSHMS17})), which are understandably overly restrictive for capturing stock price and accounting ratio dynamics; such limitations become especially evident considering market stress events, which are of critical importance in CoCo bond analysis.

When it comes to estimating the six new parameters, $\mu$, $\sigma$, $\lambda_{2}$, $\mu_{V}$, $\sigma_{V}$, and $\eta$, directly using the distribution of the stock price is unmanageable as it involves summing independent normal and compound Poisson-Erlang random variables. Instead, with return data observed at high frequencies (at least daily), we adopt the following approximation of the instantaneous log-return for any $t\geq0$:
\begin{equation}\label{2.3.4}
  \log\frac{S_{t+\Delta}}{S_{t}}\approx X_{\Delta}=\mu\Delta+\sigma\xi\sqrt{\Delta}+V_{1}\zeta_{2}-\eta A_{1}\zeta_{1},
\end{equation}
for some small time step $\Delta\in(0,1/(\lambda_{1}\vee\lambda_{2}))$, where $\xi$ is a standard normal random variable and $\zeta_{1}$ and $\zeta_{2}$ are two standard Bernoulli random variables with respective probabilities of success $\lambda_{1}\Delta$ and $\lambda_{2}\Delta$, with $\xi$, $A_{1}$, $V_{1}$, $\zeta_{1}$ and $\zeta_{2}$ being mutually independent. The approximation (\ref{2.3.4}) is valid up to the time $\tau_{3}$ of intervention and is analogous to what was used for the double-exponential jump-diffusion model in Kou (\citeyear{K02}) \text{Sect.} 3. It allows to express the distribution of $X_{\Delta}$ explicitly, as the next proposition shows.

\begin{proposition}\label{pro:2}
For $\Delta\in(0,1/(\lambda_{1}\vee\lambda_{2}))$, the probability density function of $X_{\Delta}$ can be written as
\begin{align}\label{2.3.5}
  f_{X_{\Delta}}(x)&=\lambda_{1}\lambda_{2}\Delta^{2}\Psi(x;\nu,\upsilon)+\frac{\lambda_{2}\Delta(1-\lambda_{1}\Delta)}{\sqrt{2\pi\upsilon^{2}\Delta}} e^{-(x-\nu\Delta)^{2}/(2\upsilon^{2}\Delta)} \nonumber\\
  &\qquad+\lambda_{1}\Delta(1-\lambda_{2}\Delta)\Psi(x;\mu,\sigma)+\frac{(1-\lambda_{1}\Delta)(1-\lambda_{2}\Delta)}{\sqrt{2\pi\sigma^{2}\Delta}} e^{-(x-\mu\Delta)^{2}/(2\sigma^{2}\Delta)},\quad x\in\mathbb{R},
\end{align}
where $\nu=\mu+\mu_{V}/\Delta$ and $\upsilon=\sqrt{\sigma^{2}+\sigma^{2}_{V}\big/\Delta}$, and
\begin{align*}
  \Psi(x;\nu,\upsilon)&=\frac{\Delta^{(\alpha-1)/2}}{\sqrt{2\pi\upsilon^{2}}}\bigg(\frac{\beta\upsilon}{\eta}\bigg)^{\alpha} \exp\bigg(-\frac{(x-\nu\Delta)^{2}}{2\upsilon^{2}\Delta}+\frac{(\eta(x-\nu\Delta)+\beta\upsilon^{2}\Delta)^{2}}{(2\eta\upsilon)^{2}\Delta}\bigg)\\
  &\qquad\times\D_{-\alpha}\bigg(\frac{\eta(x-\nu\Delta)+\beta\upsilon^{2}\Delta}{\eta\upsilon\sqrt{\Delta}}\bigg),
\end{align*}
in which $\D\equiv\D_{\cdot}(\cdot)$ denotes the parabolic cylinder function (see, e.g., Abramowitz and Stegun (\citeyear{AS72}) \text{Sect.} 19).
\end{proposition}

In (\ref{2.3.5}), if $A_{k}$'s are exponentially-distributed with $\alpha=1$, then the above density function can be further simplified using the well-known relation $\D_{-1}(c)=\sqrt{\pi/2}e^{c^{2}/4}\erfc(c/\sqrt{2})$, for $c\in\mathbb{R}$, where $\erfc\equiv\erfc(\cdot)$ denotes the complementary Gauss error function. Also, under (\ref{2.3.4}), the mean and variance of the log price at time $t$ are easily found to be $\E[\log S_{t}|\tau_{3}>t]=\log S_{0}+\mu t+\lambda_{2}\mu t-\eta\lambda_{1}\alpha t/\beta$ and $\Var[\log S_{t}|\tau_{3}>t]=\sigma^{2}t+\lambda_{2}\big(\mu^{2}_{V}+\sigma^{2}_{V}\big)t+\eta^{2}\lambda_{1}\alpha(1+\alpha)t/\beta^{2}$, respectively.

Suppose that the bank's stock pays dividends at a continuously compounded yield $q$, and the short risk-free rate is $r$. With only equity conversion, it is reasonable to assume that $q\geq0$ is constant, while $r\equiv(r_{t})_{t\geq0}$ is deterministic but may be time-varying.\footnote{It is possible to generalize the model framework to incorporate stochastic dividend yields and interest rates, but the benefits of this are marginal considering our current emphasis on default risk analysis.} The presence of the Brownian motion allows us to construct an risk-neutral probability measure $\PP^{\ast}\sim\PP$ under which the discounted post-dividend stock price process $(e^{-(r-q)t}S_{t})$ is a martingale using a simple mean-correcting argument.\footnote{See Yao, Yang, and Yang (\citeyear{YYY11}) for an important note on the choice of mean-correcting martingale measures.} To that end, we define the continuous density process
\begin{equation}\label{2.3.6}
  \exp\bigg(-\frac{1}{2}\int^{t}_{0}\theta^{2}_{s}\dd s-\int^{t}_{0}\theta_{s}\dd W_{s}\bigg)=\frac{\dd\PP^{\ast}}{\dd\PP}\bigg|_{\mathscr{F}_{t}},\quad t\geq0,
\end{equation}
with implicit market price of risk
\begin{equation}\label{2.3.7}
  \theta_{t}=\frac{\mu-r_{t}+q+\lambda_{1}\psi_{1}(1)+\lambda_{2}\psi_{2}(1)-\gamma\lambda_{3,t}}{\sigma}+\frac{\sigma}{2},\quad t\geq0,
\end{equation}
where $\psi_{1}$ and $\psi_{2}$ are two jump-correcting functions given by
\begin{equation*}
  \psi_{1}(u):=\E\big[e^{-\eta uA_{1}}-1\big]=\phi_{A}(\ii\eta u)-1=\bigg(1+\frac{\eta u}{\beta}\bigg)^{-\alpha}-1,\quad u>-\frac{\beta}{\eta}
\end{equation*}
and
\begin{equation*}
  \psi_{2}(u):=\E\big[e^{uV_{1}}-1\big]=e^{\mu_{V}u+\sigma^{2}_{V}u^{2}/2}-1,\quad u\in\mathbb{R}.
\end{equation*}
According to the Girsanov theorem, under $\PP^{\ast}$, $W^{\ast}\equiv\big(W^{\ast}_{t}:=W_{t}+\int^{t}_{0}\theta_{s}\dd s\big)$ becomes a standard Brownian motion, and the structures of the compound Poisson processes $J$ and $\sum^{N_{2}}_{k=1}V_{k}$ are unaltered, both of which are still independent from $W^{\ast}$. Hence, the risk-neutral stock price satisfies the SDE
\begin{align*}
  \frac{\dd S_{t}}{S_{t-}}&=(r_{t}-q-\lambda_{1}\psi_{1}(1)-\lambda_{2}\psi_{2}(1)+\gamma\lambda_{3,t})\dd t+\sigma\dd W^{\ast}_{t}+\dd\sum^{N_{2,t}}_{k=1}\big(e^{V_{k}}-1\big)+\dd\sum^{N_{1,t}}_{k=1}\big(e^{-\eta A_{k}}-1\big)\\
  &\qquad-\gamma\dd N_{3,\Lambda_{t}},\quad t\geq0,
\end{align*}
whose solution is given by
\begin{equation}\label{2.3.8}
  S_{t}=S_{0}\exp\Bigg(\int^{t}_{0}\bigg(r_{s}-q-\frac{1}{2}\sigma^{2}-\lambda_{1}\psi_{1}(1)-\lambda_{2}\psi_{2}(1)+\gamma\lambda_{3,s}\bigg)\dd s+\sigma W^{\ast}_{t}+\sum^{N_{2,t}}_{k=1}V_{k}-\eta J_{t}+N_{3,\Lambda_{t}}\log(1-\gamma)\Bigg),
\end{equation}
where the change with $W^{\ast}$ also takes place in $r$ and $\lambda_{3}$ (with time cumulation $\Lambda$), namely the dynamics of the short rate and the regulatory intervention intensity under $\PP^{\ast}$. In light of the fundamental theorem of asset pricing, (\ref{2.3.6}) underlies the no-arbitrage valuation of CoCos.

\bigskip

\section{Default mechanism}\label{sec:3}

We introduce the mathematical tools necessary for analyzing the default trigger of CoCos issued by the bank, which manifest through two channels: (i) solvency shocks and (ii) regulatory intervention. As aforementioned, the first channel corresponds to an intrinsic, accounting-based default mechanism governed by prescribed conditions, while the second arises from extrinsic, regulatory influences under abnormal conditions.

First, since the cumulation of solvency shocks with corrections is perfectly controlled by the process $J-\lambda_{1}\alpha\mathrm{id}/\beta$, the accounting trigger is directly linked to the running supremum
\begin{equation}\label{3.1}
  \bar{J}_{t}:=\sup_{s\in[0,t]}\bigg\{J_{s}-\frac{\lambda_{1}\alpha s}{\beta}\bigg\},\quad t>0,
\end{equation}
with the understanding that $\bar{J}_{0}=0$, $\PP$-a.s. The following proposition details its distribution.

\begin{proposition}\label{pro:3}
For $t>0$ and $x>0$, let $f^{(\mathrm{c})}_{J_{t}}(x)=f_{J_{t}}(x)-e^{-\lambda_{1}t}\delta_{\{0\}}(x)$ be the continuous part of the local density function of $J_{t}$ in (\ref{2.1.3}). Then, the probability distribution of $\big(\bar{J}_{t}\big)$ up to time $t>0$ is given by
\begin{align}\label{3.2}
  P(t,x):=\PP\big[\bar{J}_{t}>x\big]&=1-e^{-\lambda_{1}t}-\int^{x+\lambda_{1}\alpha t/\beta}_{0}f^{(\mathrm{c})}_{J_{t}}(y)\dd y +\frac{\lambda_{1}\alpha}{\beta}\int^{t}_{0}\mathfrak{I}(t-s)f^{(\mathrm{c})}_{J_{s}}\bigg(x+\frac{\lambda_{1}\alpha s}{\beta}\bigg)\dd s,\quad x>0,
\end{align}
where
\begin{equation}\label{3.3}
  \mathfrak{I}(t):=e^{-\lambda_{1}t}+\int^{\lambda_{1}\alpha t/\beta}_{0}\bigg(1-\frac{\beta y}{\lambda_{1}\alpha t}\bigg)f^{(\mathrm{c})}_{J_{t}}(y)\dd y,
\end{equation}
\end{proposition}

What follows immediately is the next result for the conditional distribution.

\begin{corollary}\label{cor:1}
For generic $0\leq t_{0}<t$, the probability distribution of $\bar{J}_{t}$ conditional on the $\upsigma$-field $\mathscr{F}_{t_{0}}$ is given by
\begin{equation}\label{3.4}
  P_{t_{0}}(t,x):=\PP\big[\bar{J}_{t}>x\big|\mathscr{F}_{t_{0}}\big]=P\bigg(t-t_{0},x-J_{t_{0}}+\frac{\lambda_{1}\alpha t_{0}}{\beta}\bigg)
\end{equation}
on $\{x\geq\bar{J}_{t_{0}}\}$, or $P_{t_{0}}(t,x)=1$ otherwise.
\end{corollary}

This corollary allows dynamic evaluation of the default trigger and hence to compute the value of a CoCo at any time prior to expiry or termination. In (\ref{3.4}), note also that by measurability,
\begin{equation*}
  P_{t_{0}}(t,x):=\1_{\{x<\bar{J}_{t}\}},\quad t_{0}\geq t>0.
\end{equation*}

Second, there are many possible ways to design an analytically tractable model for the intensity process $\lambda_{3}$ measuring the propensity for regulatory intervention (regulation trigger). For instance, in Cheridito and Xu (\citeyear{CX15}), $\lambda_{3}$ is a deterministic function of time, while in Chung and Kwok (\citeyear{CK16}) it is assumed to be a geometric Brownian motion inversely related to the stock price. Here, we propose an orthogonal structure which not only captures desired correlations with the idiosyncratic risk of the stock but is also parsimonious in parametrization. A major technical advantage of this structure, compared with the popular class of affine intensity models (see, e.g., Gourieroux, Monfort, and Polimenis (\citeyear{GMP06})), is its preservation of closed form under changes of measure.\footnote{In the case of the well-known Cox--Ingersoll--Ross process, for example, the inclusion of a linear drift coefficient (in the defining SDE) quickly breaks down the affine structure.} To be specific, we write
\begin{equation*}
  \lambda_{3,t}=\lambda^{(1)}_{3,t}+\lambda^{(2)}_{3,t},\quad t\geq0,
\end{equation*}
on the right side of which the first component is given by a squared Brownian motion with drift,\footnote{The squaring operation is analogous to exponentiation (as in Chung and Kwok (\citeyear{CK16}) \text{Sect.} 2.5) or reflection that the resulting processes are both positive and yet does not exhibit mean reversion (explained later). However, the integral of a geometric or reflected Brownian motion with drift has no closed-form characteristic function (see Matsumoto and Yor (\citeyear{MY05}) \text{Corol.} 3.2 and Xia and Zhang (\citeyear{XZ26}) \text{Thm.} 1).}
\begin{equation}\label{3.5}
  \lambda^{(1)}_{3,t}=(\kappa^{(1)}t+\varsigma^{(1)}W^{\ast}_{t})^{2},
\end{equation}
while the second component satisfies the following Ornstein--Uhlenbeck-type SDE:
\begin{equation}\label{3.6}
  \dd\lambda^{(2)}_{3,t}=-\kappa^{(2)}\lambda^{(2)}_{3,t-}\dd t+\varsigma^{(2)}\dd\sum^{N_{2,t}}_{k=1}h(V^{-}_{k}),
\end{equation}
which solves for
\begin{equation*}
  \lambda^{(2)}_{3,t}=\lambda^{(2)}_{3,0}e^{-\kappa^{(2)}t}+\varsigma^{(2)}\int^{t}_{0}e^{-\kappa^{(2)}(t-s)}\dd\sum^{N_{2,s}}_{k=1}h(V^{-}_{k}).
\end{equation*}
There are five parameters including $\kappa^{(1)}>0$, $\kappa^{(2)}>0$, $\varsigma^{(1)}<0$, $\varsigma^{(2)}>0$, and $\lambda^{(2)}_{3,0}>0$; $(\cdot)^{-}\equiv-\min\{\cdot,0\}$ stands for the negative part, and $h:\mathbb{R}_{+}\mapsto\mathbb{R}_{+}$ is some customizable increasing function.

Within this framework, we can analyze how the stock price influences the likelihood of regulatory intervention, considering both diffusion and jump risks. Specifically, abrupt and substantial price declines (jumps) exert a more significant impact on intervention tendencies compared to continuous price movements (diffusion). Mathematically, such a ``leverage effect'' occurring between $\lambda_{3}$ and $W^{\ast}$ and $\sum^{N_{2}}_{k=1}V_{k}$ can be analyzed by computing the quadratic covariation\footnote{Recall that the jump components are unchanged when $\PP$ is changed to $\PP^{\ast}$, as a result of which the covariation is the same under the two probability measures.}
\begin{equation}\label{3.7}
  \Bigg[\lambda_{3},W^{\ast}+\sum^{N_{2}}_{k=1}V_{k}\Bigg]_{t}=\kappa^{(1)}\varsigma^{(1)}t^{2}+2\varsigma^{(1)2}\int^{t}_{0}W^{\ast}_{s}\dd s+\varsigma^{(2)}\sum^{N_{2,t}}_{k=1}h(V^{-}_{k})V_{k},\quad t\geq0.
\end{equation}
It is familiar that the random variable $\int^{t}_{0}W^{\ast}_{s}\dd s$ for fixed $t$ has a normal distribution with zero mean and variance $t^{3}/3$, so the first two terms on the right side of (\ref{3.7}) combined are negative with $\PP^{\ast}$-probability
\begin{equation*}
  \int^{\kappa^{(1)}\varsigma^{(1)}t^{2}}_{-\infty}\sqrt{\frac{3}{8\pi\varsigma^{(1)4}t^{3}}}e^{-3x^{2}/(8\varsigma^{(1)4}t^{3})}\dd x=\frac{1}{2}\erfc\frac{\kappa^{(1)}\sqrt{6t}}{4\varsigma^{(1)}}.
\end{equation*}
Since $\erfc$ is strictly decreasing on $\mathbb{R}$ and $\kappa^{(1)}\varsigma^{(1)}<0$, the last line implies that diffusion-driven leverage takes time to manifest, at faster speed for larger values of $\kappa^{(1)}$. In contrast, jump-driven leverage always prevails, simply because $h(V^{-}_{1})V_{1}\leq0$, $\PP$-a.s. It is important that the intensity process $\lambda_{3}$ exhibits mean reversion with respect to its jumps because regulatory intervention is typically triggered by sustained or systemic financial stress -- linked to extreme or consecutive price declines -- not by transient, non-fundamental spikes; on the other hand, it is deemed optional for small, continuous price fluctuations, which generally reflect regular market activity or noise rather than fundamental solvency concerns.

Note that in general, if $\lambda^{(1)}_{3}\not\equiv0$, $\lambda_{3}$ is not a Markov process with respect to $\mathbb{F}$;\footnote{The Markov property clearly holds if $\kappa^{(1)}=0$, which case will be inadequate to establish the desired leverage effect however.} nonetheless, as Proposition \ref{pro:4} demonstrates, this does not pose technical challenges for dynamic default evaluation, because the required additional information can be incorporated into a single parameter replacing the contemporaneous value of $\lambda^{(1)}_{3}$.

\begin{proposition}\label{pro:4}
For generic $0\leq t_{0}<t$, the function
\begin{equation}\label{3.8}
  E^{\ast}_{t_{0}}(t,u):=\E^{\ast}\Big[e^{-u\int^{t}_{t_{0}}\lambda_{3,v}\dd v}\Big|\mathscr{F}_{t_{0}}\Big]=E^{(1)}_{t_{0}}(t,u)E^{(2)}_{t_{0}}(t,u),\quad u\geq0,
\end{equation}
is well-defined, where
\begin{align}\label{3.9}
  E^{(1)}_{t_{0}}(t,u)&=\exp\Bigg(\frac{\kappa^{(1)2}(t-t_{0})}{2\varsigma^{(1)2}} \bigg(\frac{\tanh\sqrt{2\varsigma^{(1)2}(t-t_{0})^{2}u}}{\sqrt{2\varsigma^{(1)2}(t-t_{0})^{2}u}}-1\bigg) +\frac{\kappa^{(1)}(\sech\sqrt{2\varsigma^{(1)2}(t-t_{0})^{2}u}-1)}{\varsigma^{(1)2}} \nonumber\\
  &\qquad\times(\kappa^{(1)}t_{0}+\varsigma^{(1)}W^{\ast}_{t_{0}}) -\frac{\lambda^{(1)}_{3,t_{0}}\sqrt{u}\tanh\sqrt{2\varsigma^{(1)2}(t-t_{0})^{2}u}}{\sqrt{2\varsigma^{(1)2}}}\Bigg) \sqrt{\sech\sqrt{2\varsigma^{(1)2}(t-t_{0})^{2}u}}
\end{align}
and
\begin{equation}\label{3.10}
  E^{(2)}_{t_{0}}(t,u)=\exp\bigg(-u\lambda^{(2)}_{3,t_{0}}\frac{1-e^{-\kappa^{(2)}(t-t_{0})}}{\kappa^{(2)}} +\int^{t}_{t_{0}}\log\phi_{\sum^{N_{2,1}}_{k=1}h(V^{-}_{k})}\bigg(\frac{\ii u\varsigma^{(2)}(1-e^{-\kappa^{(2)}s})}{\kappa^{(2)}}\bigg)\dd s\Bigg).
\end{equation}
\end{proposition}

Again, for analytical tractability, it is helpful to design the function $h$ in such a way that the integral on the right side of (\ref{3.10}) has a closed-form expression. While there may be plenty of such choices, the ones that immediately come to mind, noted that $V_{1}$ is normally distributed, are simple functions, or piecewise constant functions. More specifically, without introducing additional parameters, we define
\begin{equation}\label{3.11}
  h(x):=4\1_{\{-x\leq-3\sigma_{V}\}}+\sum^{3}_{i=1}i\1_{\{-x\in(-i\sigma_{V},(1-i)\sigma_{V}]\}}.
\end{equation}
This structure classifies the negative jump $-V^{-}_{1}$ into four levels based on the standard deviation of $V$ and assigns linearly increasing weights in terms of increasing magnitudes of $V^{-}_{1}$. We then have the next corollary.

\begin{corollary}\label{cor:2}
Let the function $h$ be given by (\ref{3.11}). Then, (\ref{3.10}) specializes for generic $0\leq t_{0}<t$ into
\begin{align}\label{3.12}
  E^{(2)}_{t_{0}}(t,u)&=\exp\Bigg(-u\lambda^{(2)}_{3,t_{0}}\frac{1-e^{-\kappa^{(2)}(t-t_{0})}}{\kappa^{(2)}} +\frac{\lambda_{2}}{\kappa^{(2)}}\sum^{4}_{i=1}a_{i}e^{-i\varsigma^{(2)}u/\kappa^{(2)}} \bigg(\Ei\bigg(\frac{iu\varsigma^{(2)}e^{-\kappa^{(2)}t_{0}}}{\kappa^{(2)}}\bigg) -\Ei\bigg(\frac{iu\varsigma^{(2)}e^{-\kappa^{(2)}t}}{\kappa^{(2)}}\bigg)\bigg) \nonumber\\
  &\qquad-\lambda_{2}(t-t_{0})\Bigg),\quad u\geq0,
\end{align}
where
\begin{align}\label{3.13}
  &a_{1}=1-\frac{1}{2}\erfc\frac{\mu_{V}+\sigma_{V}}{\sqrt{2\sigma^{2}_{V}}},\quad a_{2}=\frac{1}{2}\bigg(\erfc\frac{\mu_{V}+\sigma_{V}}{\sqrt{2\sigma^{2}_{V}}}-\erfc\frac{\mu_{V}+2\sigma_{V}}{\sqrt{2\sigma^{2}_{V}}}\bigg), \nonumber\\
  &a_{3}=\frac{1}{2}\bigg(\erfc\frac{\mu_{V}+2\sigma_{V}}{\sqrt{2\sigma^{2}_{V}}}-\erfc\frac{\mu_{V}+3\sigma_{V}}{\sqrt{2\sigma^{2}_{V}}}\bigg),\quad
  a_{4}=\frac{1}{2}\erfc\frac{\mu_{V}+3\sigma_{V}}{\sqrt{2\sigma^{2}_{V}}},
\end{align}
and $\Ei\equiv\Ei(\cdot)$ denotes the exponential integral function (Abramowitz and Stegun (\citeyear{AS72}) \text{Sect.} 5.1).
\end{corollary}

\bigskip

\section{Valuation of CoCo bonds}\label{sec:4}

Now we discuss valuation methods for the two types of CoCos (write-down and equity-convertible) under the established model framework. For convenience, we assume that these CoCos have a fixed expiry date $T>0$.\footnote{For perpetual, callable CoCos, especially those written on AT1 capital, the date $T$ can well be thought of as the first expected date to redeem the bond, as considered in De Spiegeleer and Schoutens (\citeyear{DSS14}) \text{Sect.} 4.} As discussed in Section \ref{sec:3}, the default mechanism is completely determined by two independent triggers, with the corresponding default time given by\footnote{Note that $\tau>0$ $\PP$-\text{a.s.} as long as $B_{0}>\underline{B}$, and it is agreed that $\inf\emptyset=\infty$.}
\begin{equation*}
  \tau=\tau_{1}\wedge\tau_{3}=\inf\big\{t\in(0,T]:\;B_{t}<\underline{B}\text{ or }N_{3,\Lambda_{t}}=1\big\}\in(0,T]\cup\{\infty\},
\end{equation*}
which combines the time ($\tau_{1}$) when the CET1 ratio breaches a predetermined threshold (a lower barrier) $\underline{B}$ and that ($\tau_{3}$) of an independent shock due to regulatory intervention. It is clear that $\tau$ is a stopping time of the market filtration $\mathbb{F}$ and is by definition measurable with respect to $\mathscr{F}_{T}$. The randomness of $\tau$ is fully covered by the information set $\upsigma(J_{t},N_{3,\Lambda_{t}}:t\in[0,T])$, and we have the relation
\begin{equation}\label{4.1}
  \{\tau>t\}=\big\{\bar{J}_{t}\leq\overline{J}\big\}\cap\big\{N_{3,\Lambda_{t}}=0\big\}\in\mathscr{F}_{t},\quad t\in(0,T],
\end{equation}
where the equality holds up to null sets and $\overline{J}$ is the corresponding upper barrier governing the compensated process $(J_{t}-\lambda_{1}\alpha t/\beta)$ over $(0,T]$, as implied from $\underline{B}$, or more precisely, in light of the relation (\ref{2.2.2}),
\begin{equation}\label{4.2}
  \overline{J}=\cot(\pi\underline{B})+\cot(\pi(1-B_{0})).
\end{equation}
The structure in (\ref{4.1}) also allows to distinguish usual default controlled by the accounting trigger from regulation-forced default.

\medskip

\subsection{Write-down CoCos}\label{sec:4.1}

Suppose that the bank issues a CoCo with face value $K$. The CoCo is scheduled to pay a total of $M$ coupons $\{c_{i}\}^{M}_{i=1}$ at a sequence of dates $\{t_{i}\}^{M}_{i=1}\in\prod^{M}_{i=1}(0,T]$. In the event of default, the CoCo terminates immediately, paying its face value written down by some fraction while canceling all subsequent coupon payments. For specificity, we assume that if the default is due to the regulatory trigger, such a fraction always equals 1, in which case the CoCo is considered written off -- in line with the Credit Suisse collapse (Vossos and Keatinge (\citeyear{VK23})); otherwise it is fixed at some $w\in[0,1]$.\footnote{The CoCo is said to be of write-off type if $w=1$, in which case it simply evaporates upon default as the face value is written down in full.} As discussed in Section \ref{sec:2.3}, to analyze possible interactions between the two types of default triggers, we also introduce a scaling factor $\varpi\in(0,1]$, interpreted as the probability of a usual default on an accounting trigger, which is to say, with probability $1-\varpi$ an accounting trigger leads to regulatory intervention. In connection with this, we remark that in the present model framework, the process $N_{3,\Lambda}$ only captures regulatory intervention independent of the bank's solvency, and the independence gap is exactly bridged by this scaling factor as a proxy for insolvency-induced intervention. As a consequence, the recovery rate of the CoCo is understood to be the random variable
\begin{equation}\label{4.1.1}
  R=\varpi(1-w)\1_{\{\tau_{1}<\tau_{3}\}},
\end{equation}
with which the face value of the CoCo becomes $KR$ at the time of default.

Let us denote by $\Pi^{\mathrm{(WD)}}_{t_{0}}$ the value of this write-down CoCo at any date $t_{0}\in[0,T)$ prior to termination, which is presented in the theorem below.

\begin{theorem}\label{thm:1}
Given $\tau>t_{0}$, we have that
\begin{align}\label{4.1.2}
  \Pi^{\mathrm{(WD)}}_{t_{0}}&=Ke^{-\int^{T}_{t_{0}}r_{v}\dd v}(1-P_{t_{0}}(T,\overline{J}))E^{\ast}_{t_{0}}(T,1)+\sum^{M}_{i=1}c_{i}e^{-\int^{t_{i}}_{t_{0}}r_{v}\dd v} \big(1-P_{t_{0}}\big(t_{i},\overline{J}\big)\big)E^{\ast}_{t_{0}}(t_{i},1)\mathbf{1}_{\{t_{0}<t_{i}\}} \nonumber\\
  &\qquad+\varpi(1-w)K\int^{T}_{t_{0}}e^{-\int^{s}_{t_{0}}r_{v}\dd v}E^{\ast}_{t_{0}}(s,1)\dd_{s}P_{t_{0}}(s,\overline{J}),\quad t_{0}\in[0,T),
\end{align}
where $E^{\ast}_{t_{0}}(\cdot,\cdot)$ is given by (\ref{3.8}).
\end{theorem}

The value of the CoCo is decomposed into three parts. On the right side of (\ref{4.1.2}), the first term corresponds to redemption of the face value at expiry, while the second comes from the stream of coupon payments, and together they form the non-default leg of the CoCo. The third term gives the payoff of the CoCo at the time of default, forming the default leg. This decomposition fits into the standard structure as \text{e.g.} explained in Cheridito and Xu (\citeyear{CX15}) \text{Sect.} 3. It is substantially identical to the those of traditional defaultable bonds, except that the default mechanism of the CoCo explicitly depends on an accounting ratio. As the CoCo approaches its expiry date, the probability of default tends to 0, and we have that $\Pi^{\mathrm{(WD)}}_{T-}=K\mathbf{1}_{\{\bar{J}_{T}\leq\overline{J}\}\cap\{N_{3,\Lambda_{T}}=0\}}$.

\medskip

\subsection{Equity-convertible CoCos}\label{sec:4.2}

An equity-convertible CoCo has the same scheduled payments as the write-down type considered in Section \ref{sec:4.1}, except that in the event of a usual (accounting-triggered) default, the CoCo is converted into stock shares of the bank at a trigger-dependent conversion ratio $\chi_{\tau}$.\footnote{In practice, although conversion is a pre-agreement, the conversion ratio can well be random, possibly depending on the stock price at the time of default; see, e.g., De Spiegeleer and Schoutens (\citeyear{DSS12}) \text{Sect.} 2.4.} In this paper, we introduce an adaptive structure for the conversion ratio factoring in benefits for both bondholders and shareholders. This structure strikes a delicate balance, with implications for flexible risk exposure adjustment: Bondholders receive partial protection against substantial stock price decreases near default, while conversion also does not excessively dilute current shareholders' ownership stake. The key idea is to use a (conditionally) random conversion ratio which takes the following form of a weighted geometric average:\footnote{The geometric average weighted by a power coefficient $p$ draws motivation from the risk-adjusting functionalities of power-type options in related derivatives research; see Blenman and Clark (\citeyear{BC05}), Xia (\citeyear{X17}), Xia (\citeyear{X19}), Wang and Xia (\citeyear{WX22}), and Li and Xia (\citeyear{LX25}).}
\begin{equation}\label{4.2.1}
  \chi_{\tau}=
  \begin{cases}
    \displaystyle \frac{KR}{S^{p}_{0}S^{1-p}_{\tau}} &\quad\text{if }\tau\leq T\\
    \displaystyle 0 &\quad\text{else},
  \end{cases}
\end{equation}
where $R$ is the conditional recovery rate given by (\ref{4.1.1}), $S_{0}$ is the stock price at the inception of the CoCo, and $p\in[0,1]$ is a customizable weight factor. It is equivalent to regard the conversion ratio as being determined from a convex combination of the corresponding log-prices, namely $p\log S_{0}+(1-p)\log S_{\tau}$. Since $R=0$ on $\{\tau_{3}<\tau_{1}\}$, $\chi_{\tau}$ is actually measurable with respect to $\mathscr{F}_{\tau\wedge T}$. Importantly, $p$ serves as a protection against the relative reduction in the shareholders' ownership: When $p=0$, bondholders are able to acquire shares at prevailing market prices, exposing shareholders to high dilution risks; conversely, when $p=1$, conversion becomes fixed based on the initial share price, potentially resulting in significant losses for bondholders and hindering their conversion capacity.

Note that with the power scheme (\ref{4.2.1}), the conversion price at the time $\tau$ of default is simply
\begin{equation*}
  K\chi_{\tau}S_{\tau}=(1-w)K\bigg(\frac{S_{\tau}}{S_{0}}\bigg)^{p}
\end{equation*}
conditional on $\tau_{1}\leq T<\tau_{3}$, precluding regulatory intervention, and is seen to directly generalize De Spiegeleer and Schoutens (\citeyear{DSS12}) \text{Eq.} (1), when the conversion price depends either on $S_{0}$ or on $S_{\tau}$. The use of the weight factor $p$ also provides an effective approximation for various nonstandard, floored or averaged conversion prices that have emerged in recent years, such as those previously adopted by Credit Suisse (see, e.g., Wilkens and Bethke (\citeyear{WB14}) \text{Tab.} 1); some details are given in Appendix \ref{B}.

Under (\ref{4.2.1}), the time-$t_{0}$ value of the equity-convertible CoCo, prior to termination, is detailed in Theorem \ref{thm:2}.

\begin{theorem}\label{thm:2}
Given $\tau>t_{0}$, we have that
\begin{align}\label{4.2.2}
  \Pi^{\mathrm{(EC)}}_{t_{0}}&=Ke^{-\int^{T}_{t_{0}}r_{v}\dd v}(1-P_{t_{0}}(T,\overline{J}))E^{\ast}_{t_{0}}(T,1)+\sum^{M}_{i=1}c_{i}e^{-\int^{t_{i}}_{t_{0}}r_{v}\dd v} \big(1-P_{t_{0}}\big(t_{i},\overline{J}\big)\big)E^{\ast}_{t_{0}}(t_{i},1)\mathbf{1}_{\{t_{0}<t_{i}\}} \nonumber\\
  &\qquad+\varpi(1-w)K\bigg(\frac{S_{t_{0}}}{S_{0}}\bigg)^{p}\int^{T}_{t_{0}}e^{-\int^{s}_{t_{0}}\tilde{q}^{\circ}_{v}\dd v}\tilde{E}_{t_{0}}(s,1-p\gamma)\dd_{s}\tilde{P}_{t_{0}}(s,\overline{J}), \quad t_{0}\in[0,T),
\end{align}
where
\begin{equation}\label{4.2.3}
  \tilde{q}^{\circ}=pq+(1-p)r+\frac{1}{2}p(1-p)\sigma^{2}+\lambda_{1}(p\psi_{1}(1)-\psi_{1}(p))+\lambda_{2}(p\psi_{2}(1)-\psi_{2}(p)),
\end{equation}
and $\tilde{E}_{t_{0}}(\cdot,\cdot)$ and $\tilde{P}_{t_{0}}(\cdot,\cdot)$ are just $E^{\ast}_{t_{0}}(\cdot,\cdot)$ and $P_{t_{0}}(\cdot,\cdot)$, respectively, evaluated with the parameter changes $\tilde{\kappa}^{(1)}=\kappa^{(1)}+p\sigma\varsigma^{(1)}$, $\tilde{\lambda}_{2}=\lambda_{2}(\psi_{2}(p)+1)$, $\tilde{\mu}_{V}=\mu_{V}+p\sigma^{2}_{V}$, $\tilde{\lambda}_{1}=\lambda_{1}(\psi_{1}(p)+1)$, and $\tilde{\beta}=\beta+p\eta$.\footnote{All unmentioned parameters are unchanged, i.e., equal to their counterparts under $\PP^{\ast}$.}
\end{theorem}

A comparison between Theorem \ref{thm:1} with Theorem \ref{thm:2} shows that while the non-default legs are trivially identical, the default leg of the equity-convertible CoCo boils down to the same recovered cash payment $KR$ if $p=0$; if $p=1$, then it is perfectly dependent on the stock price at the conversion time, as in the case of standard equity-convertible CoCos. We state the next corollary.

\begin{corollary}\label{cor:3}
Given $\tau>t_{0}$, then for a standard convertible CoCo,
\begin{align}\label{4.2.4}
  \Pi^{\mathrm{(EC)}}_{t_{0}}&=Ke^{-\int^{T}_{t_{0}}r_{v}\dd v}(1-P_{t_{0}}(T,\overline{J}))E^{\ast}_{t_{0}}(T,1)+\sum^{M}_{i=1}c_{i}e^{-\int^{t_{i}}_{t_{0}}r_{v}\dd v} \big(1-P_{t_{0}}\big(t_{i},\overline{J}\big)\big)E^{\ast}_{t_{0}}(t_{i},1)\mathbf{1}_{\{t_{0}<t_{i}\}} \nonumber\\
  &\qquad+\frac{\varpi(1-w)KS_{t_{0}}}{S_{0}}\int^{T}_{t_{0}}e^{-q(s-t_{0})}\tilde{E}_{t_{0}}(s,1-\gamma)\dd_{s}\tilde{P}_{t_{0}}(s,\overline{J}), \quad t_{0}\in[0,T),
\end{align}
where the evaluation of $\tilde{E}_{t_{0}}(\cdot,\cdot)$ and $\tilde{P}_{t_{0}}(\cdot,\cdot)$ is as explained in Theorem \ref{thm:2} with $p=1$.
\end{corollary}

In an ideal situation, dynamically hedging convertible CoCos requires positions linked to the underlying risk factors, which specifically include fluctuations in the CET1 ratio, stock price, and the propensity for regulatory intervention under market stress. In practice, while the stock itself serves as a natural hedging instrument, additional risks can be managed using credit default swaps (CDSs) or other credit derivatives; see the detailed discussions in De Spiegeleer and Schoutens (\citeyear{DSS12}) and Wilkens and Bethke (\citeyear{WB14}). In the presence of jump risks, constructing perfect hedges is not feasible. Therefore, we concentrate on delta-gamma hedging for implementation, using only the stock. The hedge is constructed at a frequency at least matching that of the observed market data (typically daily) and is based on computed model prices. Most importantly, even with this partial hedge, our empirical analysis (Section \ref{sec:6}) will demonstrate that sole reliance on delta-gamma hedging already significantly reduces hedging errors compared to various results in the literature under the same settings.



\bigskip

\section{Model implementation}\label{sec:5}

As mentioned earlier, our valuation approach represents a hybrid methodology by integrating financial data on CET1 ratios and stock returns while accounting for forced, regulation-based default risk. Unlike a purely calibrated model that takes no advantage of these observed fluctuations, implementation can involve both statistical estimation and model calibration, which is a key benefit when it comes to data-driven CoCo price predictions. This section aims to elucidate the procedures for implementing the proposed model, to showcase its validity, efficiency, and to set the groundwork for subsequent case studies.

We structure the implementation into three phases: Preparation (Phase O), Estimation (Phase I), and Calibration (Phase II). Phase O is a preliminary stage focused on collecting fundamental data pertaining to the CoCo under examination. Much of this information is accessible from the CoCo's term sheets. Also, during this phase, efforts are made to secure data on the risk-free interest rate term structure and to derive reasonable estimates for dividend yields spanning the CoCo's lifespan. In particular, regarding international banks, we employ (currency-specific) swap rates to outline the interest rate term structure (with short rate process $r$). As a consequence, for any fixed coupon date $t_{i}$, we have the transformation $\int^{t_{i}}_{0}r_{v}\dd v=r(t_{i})t_{i}$ within the valuation formulas in Section \ref{sec:4}, where $r(t_{i})$ represents the $t_{i}$-year spot rate. The constant dividend yield $q$ is estimated using current dividend data and adjusted for continuous compounding.

According to Section \ref{sec:2}, the model parameters can be broadly divided into three groups: those associated with the CET1 ratio, the stock price, and the propensity for regulatory intervention, respectively. Phase I then focuses on estimating the CET1 ratio and stock price parameters using time series data, directly available from the bank's (quarterly) financial reports and market data, respectively, while Phase II involves determining, by leveraging credit spread or CoCo price data, optimal parameter values for the intensity process $\lambda_{3}$ that gauges the propensity for intervention, including the auxiliary parameter $\varpi$ for the probability of no accounting-triggered intervention. Notably, in Phase I, it is crucial to limit data inclusion up to time $\tau_{3}$ marking the onset of intervention; this can be reasonably characterized by public announcements of bankruptcy or merger plans accompanied by significant stock price drops (as in the case of Credit Suisse). In addition, depending on data availability, the CET1 ratio parameters can be alternatively estimated using stock return data, along with the stock price parameters.\footnote{Yet another option is to calibrate the CET1 ratio parameters directly on contemporaneous CoCo price data during Phase II, along with the parameters of $\lambda_{3}$.}

Table \ref{tab:1} provides a summary of the parameters in terms of their roles and sources; only those directly used in Theorem \ref{thm:1} and Theorem \ref{thm:2} are included.

\clearpage

\begin{table}[H]\small
  \centering
  \caption{Parameter categorization for CoCo valuation}
  \label{tab:1}
  \begin{tabular}{c|ccc}
    \toprule
    \textbf{Phase} & \textbf{parameter} & \textbf{category} & \textbf{data source} \\ \midrule
    \multirow{8}{*}{O} & CoCo notional ($K$) & \multirow{5}{*}{static input} & \multirow{5}{*}{term sheet} \\
    & CoCo maturity ($T$) & & \\
    & coupon numbers ($M$) & & \\
    & coupon dates ($t_{i}$) & & \\
    & coupon payments ($c_{i}$) & & \\ \cline{2-4}
    & risk-free interest rate ($r$) & dynamic input & swap rates \\ \cline{2-4}
    & dividend yield ($q$) & dynamic input & \makecell{future dividend \\ yield estimates} \\ \hline
    \multirow{3}{*}{I or II} & solvency shock intensity ($\lambda_{1}$)  & \multirow{3}{*}{\makecell{dynamic input \\ or calibration}} & \multirow{3}{*}{\makecell{financial reports, return data \\ or CoCo price data}} \\
    & solvency shock shape ($\alpha$) & & \\
    & solvency shock rate ($\beta$) & & \\ \hline
    \multirow{5}{*}{I} & stock volatility ($\sigma$) & \multirow{5}{*}{dynamic input} & \multirow{5}{*}{return data} \\
    & stock jump intensity ($\lambda_{2}$) & & \\
    & stock jump mean ($\mu_{V}$) & & \\
    & stock jump scale ($\sigma_{V}$) & & \\
    & stock-CET1 leverage ($\eta$) & & \\ \hline
    \multirow{7}{*}{II} & \makecell{regulatory intervention \\ propensity parameters \\ ($\kappa^{(1,2)},\varsigma^{(1,2)},\lambda^{(2)}_{3,0}$)} & \makecell{dynamic input \\ or calibration} & \makecell{credit spread data \\ or CoCo price data} \\ \cline{2-4}
    & transformed trigger barrier ($\overline{J}$) & \makecell{dynamic input \\ or calibration} & \makecell{term sheet and financial reports \\ or CoCo price data} \\ \cline{2-4}
    & independence probability ($\varpi$) & calibration & CoCo price data \\
    & intervention shock scale ($\gamma$) & calibration & CoCo price data \\ \cline{2-4}
    & write-down fraction ($w$) & \multirow{2}{*}{\makecell{static input \\ or calibration}} & \multirow{2}{*}{\makecell{term sheet \\ or CoCo price data}} \\
    & conversion power ($p$) & & \\
    \bottomrule
  \end{tabular}
\end{table}

\medskip

\subsection{Estimation procedures}\label{sec:5.1}

We describe the key procedures for conducting statistical estimation in Phase I. For a target bank, we collect $I_{B}$ quarterly CET1 ratio data $\check{B}=\{\check{B}_{i}\}^{I_{B}}_{i=1}$ and $I_{Y}$ daily stock return (excluding dividends) data $\check{Y}=\{\check{Y}_{i}\}^{I_{Y}}_{i=1}$.\footnote{These are simple returns which can be alternatively computed from daily closing stock price data, say $\check{S}$.} The CET1 ratio data and stock return data should cover the same ending period (in quarters), while the starting periods are generally allowed to be different. Then, to recover the source of solvency shocks from the CET1 ratio data based on (\ref{2.2.1}), we compute the the following inverse-transformed series:
\begin{equation}\label{5.1.1}
  \breve{J}_{i}=\tan\frac{\pi(1-2\check{B}_{i})}{2}+\cot(\pi(1-\check{B}_{0})),\quad i\in\mathbb{Z}\cap[1,I_{B}].
\end{equation}
The inversion is always valid because $|\check{B}_{i}|<1$ by definition, and the second addend in (\ref{5.1.1}) is merely a correcting term. Thus, we can focus on the difference series
\begin{equation}\label{5.1.2}
  \breve{H}_{i}=\breve{J}_{i}-\breve{J}_{i-1},\quad i\geq2.
\end{equation}
It is also beneficial to conduct the augmented Dickey--Fuller test to ensure that $\breve{H}$ is a stationary time series. In the meantime, we compute the series of daily log-returns via the standard relation
\begin{equation*}
  \check{X}_{i}=\log(1+\check{Y}_{i}),\quad i\geq2,
\end{equation*}
which assumes the (real-valued) distribution specified by (\ref{2.3.5}). The augmented Dickey--Fuller test can be used similarly to guarantee the stationarity of $\check{X}$.

Note that the difference random variables $\breve{H}_{i}$'s from (\ref{5.1.2}) have a (negatively) shifted compound Poisson-Erlang distribution under the model assumptions, which is challenging to estimate because it is a mixed discrete--continuous distribution whose only point mass is parameterized at a time multiple of $-\lambda_{1}\alpha/\beta$.\footnote{This case is inherently different from a non-shifted distribution whose point mass is fixed at 0 (not parameterized). In this latter case, MLE has been studied before -- see Withers and Nadarajah (\citeyear{WN11}) \text{Sect.} 4. We remark that the procedures proposed herein represent a new approach to estimating compensated compound Poisson processes, and the theoretical underpinnings of this method will be explored in a separate paper to avoid digression.} To overcome this difficulty, we propose a verification-based method. Specifically, due to the positive support of the (non-shifted) compound Poisson-Erlang distribution, the negative values observed in $\breve{H}_{i}$'s must undergo the compensation effect (or negative drifting). This justifies the following sign correction:
\begin{equation*}
  \check{H}_{i}=\breve{H}_{i}-\min\breve{H}+\varepsilon,
\end{equation*}
for a customizable correction factor $\varepsilon>0$. As a result, $\check{H}_{i}$'s are guaranteed to have positive values.

With the closed-form formulas in Proposition \ref{pro:1} and Proposition \ref{pro:2}, implementation of Maximum Likelihood Estimation (MLE) is generally very efficient. In particular, by setting a small value (say 0.01) for $\varepsilon$, first we estimate the three CET1 ratio parameters according to
\begin{equation}\label{5.1.3}
  (\hat{\lambda}_{1},\hat{\alpha},\hat{\beta})\in \underset{(\lambda_{1},\alpha,\beta)\in\mathbb{R}_{++}\times\mathbb{Z}_{++}\times\mathbb{R}_{++}}{\mathrm{argmax}}\sum^{I_{B}}_{i=2}\log f^{(\mathrm{c})}_{J_{1/4}}(\check{H}_{i};\lambda_{1},\alpha,\beta),
\end{equation}
where $1/4$ corresponds to the quarterly frequency and $f^{(\mathrm{c})}_{J_{1/4}}$ is recalled to be the continuous part of the (local) density function $J_{1/4}$ (see Proposition \ref{pro:3}). After this, we check if $\hat{\lambda}_{1}\geq4\log100\approx18.4207$, which condition is equivalent to $e^{-\hat{\lambda}_{1}/4}\leq0.01$. It is deemed that successful verification of this condition implies that the point mass of the distribution of $J_{1/4}$ is negligible for practical purposes, leading to the adoption of $(\hat{\lambda}_{1},\hat{\alpha},\hat{\beta})$ as the final parameter estimates; in case the verification fails, we further compute
\begin{align}\label{5.1.4}
  \hat{\lambda}'_{1}&=-4\log\bigg(1-\int^{\infty}_{0+}f^{(\mathrm{c})}_{J_{1/4}}(x;\hat{\alpha},\hat{\beta})\dd x\bigg), \nonumber\\
  (\hat{\alpha}',\hat{\beta}')&\in\underset{(\alpha,\beta)\in\mathbb{Z}_{++}\times\mathbb{R}_{++}}{\mathrm{argmax}}\sum^{I_{B}}_{i=2}\log f^{(\mathrm{c})}_{J_{1/4}}(\check{H}_{i};\hat{\lambda}_{1},\alpha,\beta),
\end{align}
and adopt $(\hat{\lambda}'_{1},\hat{\alpha}',\hat{\beta}')$ as the refined parameter estimates. Then, given the CET1 ratio parameter estimates, the six stock-specific parameters are directly estimated via
\begin{equation}\label{5.1.5}
  (\hat{\mu},\hat{\sigma},\hat{\lambda}_{2},\hat{\mu}_{V},\hat{\sigma}_{V},\hat{\eta})\in \underset{(\mu,\sigma,\mu_{V},\sigma_{V},\eta)\in\mathbb{R}\times\mathbb{R}_{++}\times\mathbb{R}\times\mathbb{R}^{2}_{++}} {\mathrm{argmax}}\sum^{I_{S}}_{i=2}\log f_{X_{1/252}}(\check{X}_{i};\hat{\lambda}_{1},\hat{\alpha},\hat{\beta},\mu,\sigma,\lambda_{2},\mu_{V},\sigma_{V},\eta),
\end{equation}
where $1/252$ corresponds to the daily frequency and $(\hat{\lambda}_{1},\hat{\alpha},\hat{\beta})$ should be replaced by $(\hat{\lambda}'_{1},\hat{\alpha}',\hat{\beta}')$ if $\hat{\lambda}_{1}<4\log100$. While (\ref{5.1.5}) yields six parameter estimates, only five are needed for the pricing--hedging of CoCos, as shown in Table \ref{tab:1}, as the drift coefficient $\mu$ is removed by the measure change.

If it is infeasible to obtain CET1 ratio data ($\check{B}$) for the period under consideration, the CET1 ratio parameters and the stock price parameters can be estimated in a single step, i.e., combining (\ref{5.1.3}) and (\ref{5.1.5}), via
\begin{align}\label{5.1.6}
  (\hat{\lambda}_{1},\hat{\alpha},\hat{\beta},\hat{\mu},\hat{\sigma},\hat{\lambda}_{2},\hat{\mu}_{V},\hat{\sigma}_{V},\hat{\eta}) &\in\underset{(\lambda_{1},\alpha,\beta,\mu,\sigma,\lambda_{2},\mu_{V},\sigma_{V},\eta) \in\mathbb{R}_{++}\times\mathbb{Z}_{++}\times\mathbb{R}_{++}\times\mathbb{R}\times\mathbb{R}^{2}_{++}\times\mathbb{R}\times\mathbb{R}^{2}_{++}} {\mathrm{argmax}} \nonumber\\
  &\qquad \sum^{I_{S}}_{i=2}\log f_{X_{1/252}}(\check{X}_{i};\lambda_{1},\alpha,\beta,\mu,\sigma,\lambda_{2},\mu_{V},\sigma_{V},\eta),
\end{align}
provided that $\hat{\lambda}_{1}\geq4\log 100$.

\medskip

\subsection{Calibration procedures}\label{sec:5.2}

In Phase II, the parameters pertaining to regulatory intervention are optimized through calibration using available CoCo price data. We start by outlining the main steps to efficiently implement the CoCo valuation formulas (\ref{4.1.2}) and (\ref{4.2.2}).

The distribution formulas in Proposition \ref{pro:3} only involve proper integrals, with integrands possessing no singularities or oscillations over the domains of integration, and so they are very amenable to numerical computations; in particular, the standard Gauss--Kronrod quadrature rule (see, e.g., Ma, Rokhlin, and Wandzura (\citeyear{MRW96})) is applicable to approximate the integral in (\ref{3.3}) as
\begin{equation*}
  \hat{\mathfrak{I}}(t)=\frac{\lambda_{1}\alpha t}{\beta M_{1}}\sum^{M_{1}}_{k=0}\bigg(1-\frac{k}{M_{1}}\bigg)f^{(\mathrm{c})}_{J_{t}}\bigg(\frac{\lambda_{1}\alpha kt}{\beta M_{1}}\bigg)
\end{equation*}
for $t>0$ and $x>0$, where $M_{1},M_{2}\gg1$ are the numbers of quadratures. A similar approximation for (\ref{3.2}) is then
\begin{align*}
  \hat{P}(t,x)&=1-e^{-\lambda_{1}t}-\frac{x+\lambda_{1}\alpha t/\beta}{M_{3}}\sum^{M_{3}}_{k=0}f^{(\mathrm{c})}_{J_{t}}\bigg(\frac{k(x+\lambda_{1}\alpha t/\beta)}{M_{3}}\bigg)+\frac{\lambda_{1}\alpha t}{\beta M_{4}}\sum^{M_{4}}_{k=0}\hat{\mathfrak{I}}\bigg(\frac{(M_{4}-k)t}{M_{4}}\bigg)f^{(\mathrm{c})}_{J_{t}}\bigg(x+\frac{\lambda_{1}\alpha kt}{\beta M_{4}}\bigg)
\end{align*}
with $M_{3},M_{4}\gg1$. In applications, it is reasonable to require that the relative numerical error of each approximation be less than $10^{-4}$. For example, $|\hat{\mathfrak{I}}(t)-\mathfrak{I}(t)|<10^{-4}$; we shall stick to this precision goal during implementation.

For the time integrals contained in the valuation formula (\ref{4.1.2}), we may directly discretize the time differential $\dd_{s}P(s,\overline{J})$, by using the approximation
\begin{align*}
  \int^{T}_{t_{0}}e^{-\int^{s}_{t_{0}}r_{v}\dd v}E^{\ast}_{t_{0}}(s,1)\dd_{s}P_{t_{0}}(s,\overline{J})&\approx\sum^{M_{5}}_{k=\lfloor M_{5}t_{0}/T\rfloor+1}\exp\bigg(-\int^{kT/M_{5}}_{t_{0}}r_{v}\dd v\bigg) E^{\ast}_{t_{0}}\bigg(\frac{kT}{M_{5}},1\bigg) \\
  &\qquad\times\bigg(P_{t_{0}}\bigg(\frac{kT}{M_{5}},\overline{J}\bigg)-P_{t_{0}}\bigg(\frac{(k-1)T}{M_{5}},\overline{J}\bigg)\bigg)
\end{align*}
with $M_{5}\gg1$. This construction implies that the default event can only occur at the time points $\{kT/M_{5}\}^{M_{5}}_{k=1}\subset[0,T]$ rather than in a continuous manner. A similar discretization scheme applies for (\ref{4.2.2}) as well, or precisely,
\begin{align*}
  \int^{T}_{t_{0}}e^{-\int^{s}_{t_{0}}\tilde{q}^{\circ}_{v}\dd v}\tilde{E}_{t_{0}}(s,1-p\gamma)\dd_{s}\tilde{P}_{t_{0}}(s,\overline{J})&\approx\sum^{M_{5}}_{k=\lfloor M_{5}t_{0}/T\rfloor+1}\exp\bigg(-\int^{kT/M_{5}}_{t_{0}}\tilde{q}^{\circ}_{v}\dd v\bigg)\tilde{E}_{t_{0}}\bigg(\frac{kT}{M_{5}},1-p\gamma\bigg) \\
  &\qquad\times\bigg(\tilde{P}_{t_{0}}\bigg(\frac{kT}{M_{5}},\overline{J}\bigg) -\tilde{P}_{t_{0}}\bigg(\frac{(k-1)T}{M_{5}},\overline{J}\bigg)\bigg).
\end{align*}
To balance between accuracy and computational effort, we set the criterion for choosing suitable values of $M_{5}$ as to accommodate default with at least daily frequency, i.e., $M_{5}\geq252(T-t_{0})$ given any $t_{0}\in[0,T]$.

With the above implementation steps in mind, for subsequent calibration, we collect $I_{\rm CoCo}$ price data, $\{\check{\Pi}_{i}\}^{I_{\rm CoCo}}_{i=1}$, for a certain CoCo issued by the target bank, which can be either of the write-down type or the equity convertible type. Suppose that this CoCo has a fixed expiry date $T$ (recalling the discussion in Section \ref{sec:4}). Observations are made for the same CoCo on consecutive dates $\big\{t^{(i-1)}_{0}\big\}^{I_{\rm CoCo}}_{i=1}$, with $t^{(0)}_{0}=0$.

The parameters attached to regulatory intervention are directly optimized through calibration using the CoCo price data. For this purpose, we also need to collect $I_{\rm cs}$ data $\check{\rm cs}=\{\check{\rm cs}_{i,j}\}^{I_{\rm cs}}_{i=1}$ for the bank's credit spread term structure, where each data point represents a different observation date $t^{(i-1)}_{0}$ and maturity $T$, with $I_{\rm cs}\geq I_{\rm CoCo}$ and $T$ matching the CoCo expiry date. To ensure that the calibration reflects contemporaneous information along with the above estimation, it is important for both the CoCo price data and the credit spread data to align with the latest period for observing the CET1 ratio data. Towards this end, the dates $\big\{t^{(i-1)}_{0}\big\}^{I_{\rm cs}}_{i=1}$ are chosen to match the exact same month when $\check{B}_{I_{B}}$ is available and regulatory intervention takes place. Based on Proposition \ref{pro:4} and Theorem \ref{thm:2}, the associated intervention probabilities can then be dynamically estimated in terms of
\begin{align}\label{5.2.1}
  E^{\ast}_{t^{(i-1)}_{0}}(T,1)&\approx\exp\big(-\check{\rm cs}_{t^{(i-1)}_{0}}\big)E^{\ast}_{0}(T-t^{(i-1)}_{0},1), \nonumber\\
  \tilde{E}_{t^{(i-1)}_{0}}(T,1)&\approx\exp\big(-\check{\rm cs}_{t^{(i-1)}_{0}}\big)\tilde{E}_{0}(T-t^{(i-1)}_{0},1),\quad i\in\mathbb{Z}\cap[1,I_{\rm cs}],
\end{align}
where on the right sides, both terms $E^{\ast}_{0}(T-t^{(i-1)}_{0},1)$ and $\tilde{E}_{0}(T-t^{(i-1)}_{0},1)$ are evaluated with $\kappa^{(1)}{t^{(i-1)}_{0}}+\varsigma^{(1)}W_{t^{(i-1)}_{0}}=\lambda^{(2)}_{3,t^{(i-1)}_{0}}=0$ (see (\ref{3.9}) and (\ref{3.10})). The approximation in (\ref{5.2.1}) makes use of real-time credit spreads as dynamic inputs while reducing the number of intervention-based parameters to just four: $(\kappa^{(1)},\varsigma^{(1)},\kappa^{(2)},\varsigma^{(2)})\in\mathbb{R}_{++}\times\mathbb{R}_{--}\times\mathbb{R}^{2}_{++}$.

Meanwhile, the independence probability parameter $\varpi$ and the shock scale parameter $\gamma$ should be optimized via the CoCo price data as well, with which the CET1 ratio and intervention parameters can be optionally tuned.

As the accounting trigger barrier $\underline{B}$ is typically set by regulatory authorities (in accordance with Basel III guidelines), a reasonable choice of $\overline{J}$, given data $\check{B}$ on the CET1 ratio, can be made according to (\ref{4.2}), with $B_{0}=\check{B}_{0}$ being the initial CET1 ratio. However, this calculation method is overly susceptible to the single observation $\check{B}_{0}$, inherent to the collected CET1 ratio data from Section \ref{sec:5.1}, and so a more robust method, which we recommend and will adopt later, is to optimize $\overline{J}$ using the CoCo price data.

Also can be optimized from the CoCo price data are the write-down fraction $w\in[0,1]$ and the conversion power $p>0$, especially where they are not explicitly stated in the term sheet. Therefore, in theory, calibration on CoCo price data can involve from $4$ to $4+2+3+3=12$ parameters, noting that the parameters tied to the CET1 ratio may be alternatively optimized in this way.\footnote{However, it is generally recommended to limit a single calibration to no more than 10 parameters for the sake of robustness, which is the case of the subsequent empirical analysis.} Then, with $\vartheta$ being a placeholder for to-be-calibrated parameters other than $\varpi$ and $\gamma$, we proceed with the following minimization problem for the relative pricing error:
\begin{equation*}
  \underset{(\varpi,\gamma,\vartheta)\in(0,1]\times[0,1]\times\Theta}{\min}\frac{1}{I_{\rm CoCo}}\sum^{I_{\rm CoCo}}_{i=1}\frac{\big|\Pi_{t^{(i-1)}_{0}}-\check{\Pi}_{i}\big|}{\check{\Pi}_{i}},
\end{equation*}
where the model prices $\Pi_{t^{(i-1)}_{0}}\equiv\Pi_{t^{(i-1)}_{0}}(\varpi,\gamma,\vartheta)$ are as shown in Theorem \ref{thm:1} and Theorem \ref{thm:2}, with all other parameters fixed, and $\Theta$ denotes the (floating) value range of $\vartheta$. We refer again to Table \ref{tab:1} for a summary of parameter categories. To facilitate comparison with existing literature, in our empirical analysis we shall use the root-mean-square error (RMSE), with
\begin{equation}\label{5.2.2}
  \big(\hat{\varpi},\hat{\gamma},\hat{\vartheta}\big)\in\underset{(\varpi,\gamma,\vartheta)\in(0,1]\times[0,1]\times\Theta}{\mathrm{argmin}} \sqrt{\frac{1}{I_{\rm CoCo}}\sum^{I_{\rm CoCo}}_{i=1}\big(\Pi_{t^{(i-1)}_{0}}-\check{\Pi}_{i}\big)^{2}},
\end{equation}
which has more emphasis on the absolute error than the relative error.

\bigskip

\section{Empirical analysis}\label{sec:6}

To validate the performance of our proposed valuation framework, for our empirical analysis we conducted five detailed case studies spanning three major global banking institutions -- Lloyds Banking Group (NYSE: LYG), Credit Suisse (NYSE: CS), and China Construction Bank (CCB | HKG: 0939) -- across different time periods. These banks were selected not only for their significance in the CoCo bond market but also to provide a diverse set of market conditions and regulatory environments for robust model setting and testing. For both Lloyds and Credit Suisse, we analyzed two distinct time horizons: a recent period reflecting significant market events in 2023, and an earlier period (2009--2011) allowing for direct comparison with existing literature, particularly the results in Wilkens and Bethke (\citeyear{WB14}).


\medskip

\subsection{Data and preparation}\label{sec:6.1}

For each period under consideration, daily stock return data were obtained from WRDS CRSP, while quarterly accounting ratios were hand-collected from the banks' annual or quarterly financial reports. All other data, including clean prices of CoCo bonds (all equity-convertible, based on a notional of \$100, in USD or HKD, with quarterly or semiannual coupon payments), credit spreads, interest rates, and contemporaneous dividend yield estimates, were sourced from Bloomberg financial terminals. Special care was taken to ensure temporal alignment across data series with differing frequencies, such as daily stock prices and quarterly accounting ratios. For the Credit Suisse 2021--2023 case study, the data are uniformly truncated as of March 17, 2023, which marks the onset of a regulatory-triggered default ($\tau_{3}$); this was done to reflect the significant market dislocation that followed, including a sharp decline of approximately 50\% in the stock price on March 20 (see again Damyanova (\citeyear{D23}) and Vossos and Keatinge (\citeyear{VK23})).\footnote{On March 19, UBS announced its agreement to acquire Credit Suisse for approximately \$3.2 billion -- a substantial discount to its prior market valuation -- in a deal orchestrated by the Swiss government and the Swiss National Bank to safeguard financial stability.}



\smallskip

As we have seen from Proposition \ref{pro:3}, computation of the default probability $P(t,x)$, for $t\in(0,T]$ and $x>0$, requires evaluating three numerical integrals, two of which are iterated, with the inner one contained in the function $\mathfrak{I}(t)$, in (\ref{3.3}). Moreover, based on our discussion in Section \ref{sec:5.2}, approximations for the default legs\footnote{Notably, the default legs are generally not expressible in closed form under the intensity-based approach (see, e.g., Cheridito and Xu (\citeyear{CX15})), which is adopted at the regulatory intervention level in the present paper.} in (\ref{4.1.2}) and (\ref{4.2.2}) are based on a large number $M_{5}$ of time steps in order to guarantee an adequate default schedule in practice. Therefore, the calibration task (\ref{5.2.2}) can be computationally costly to carry out with direct numerical integration. For this reason, as a means to significantly boost computational efficiency, we employ radial basis function (RBF) interpolation for the default probability with respect to its five input variables, including three CET1 ratio parameters: $(t,x;\lambda_{1},\alpha,\beta)$.

\begin{figure}[H]
  \centering
  \includegraphics[scale=0.3]{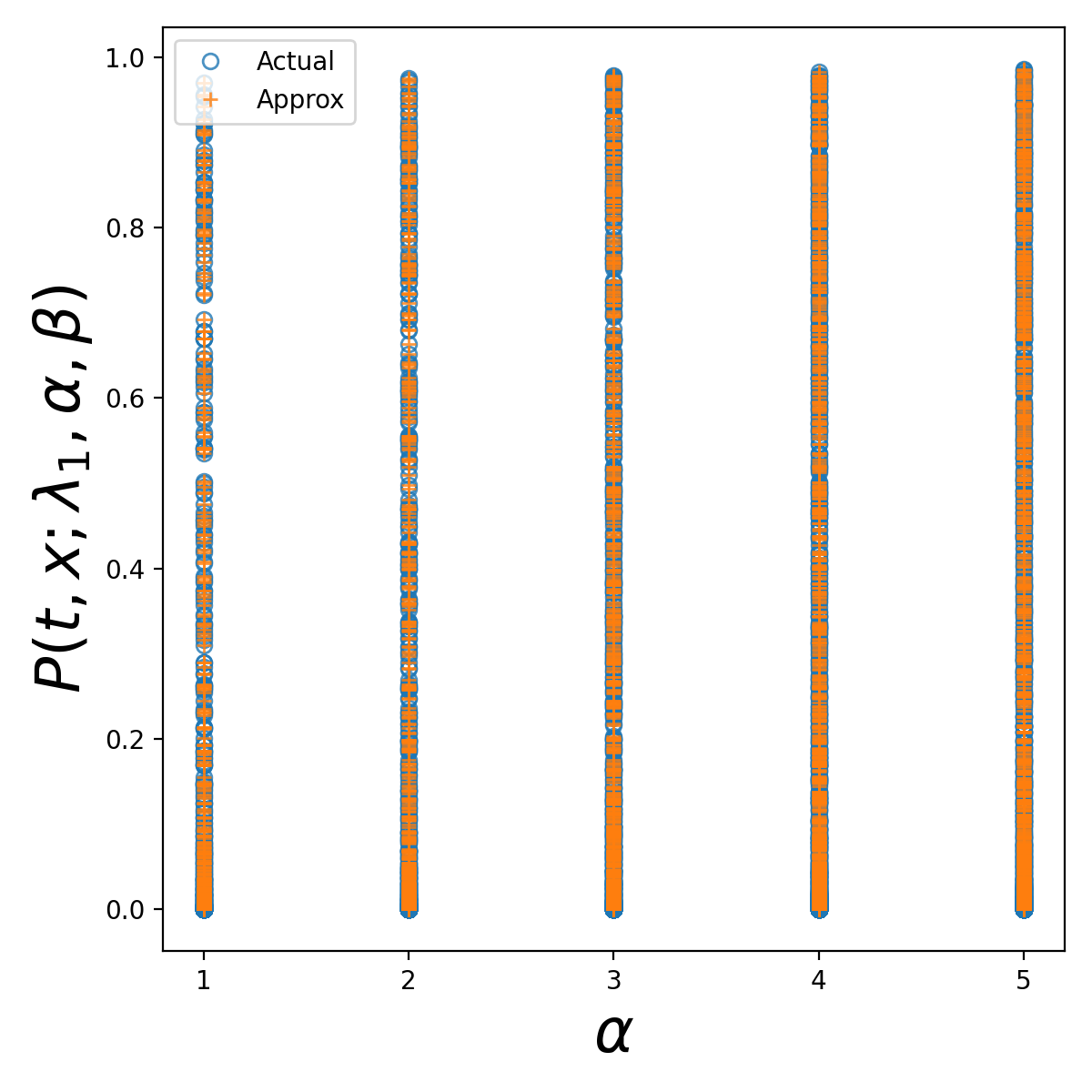}
  \includegraphics[scale=0.3]{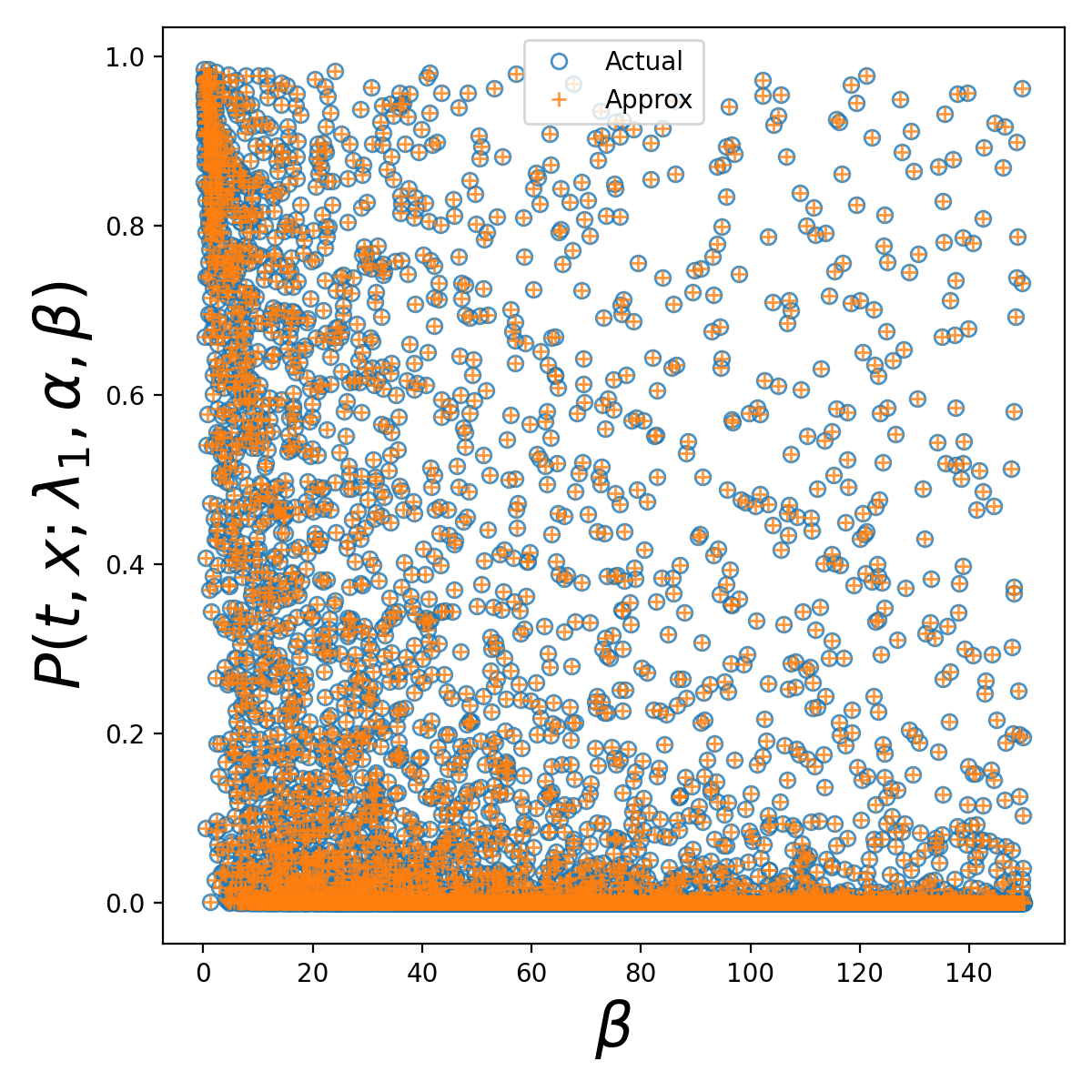}
  \includegraphics[scale=0.3]{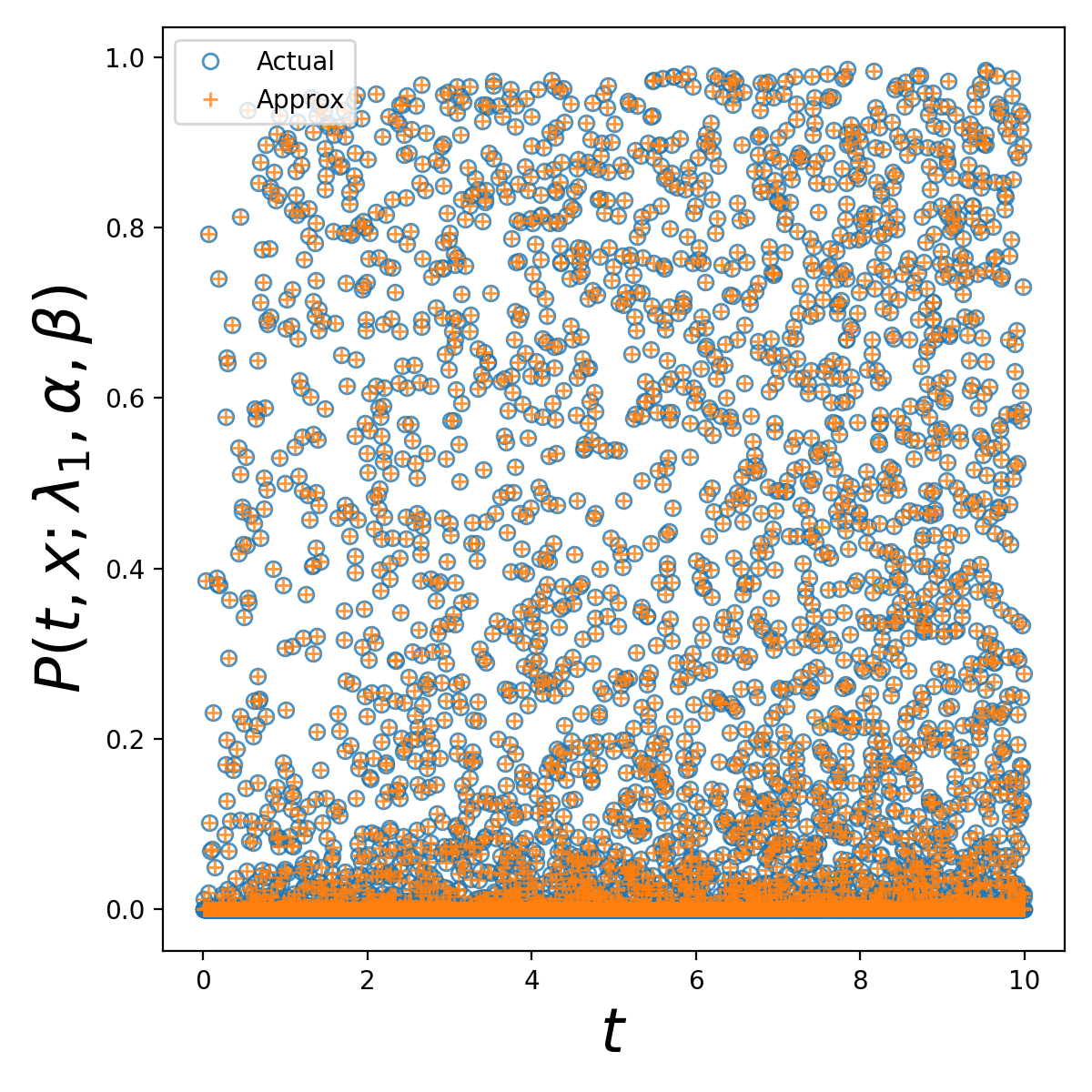}
  \bigskip \\
  \includegraphics[scale=0.3]{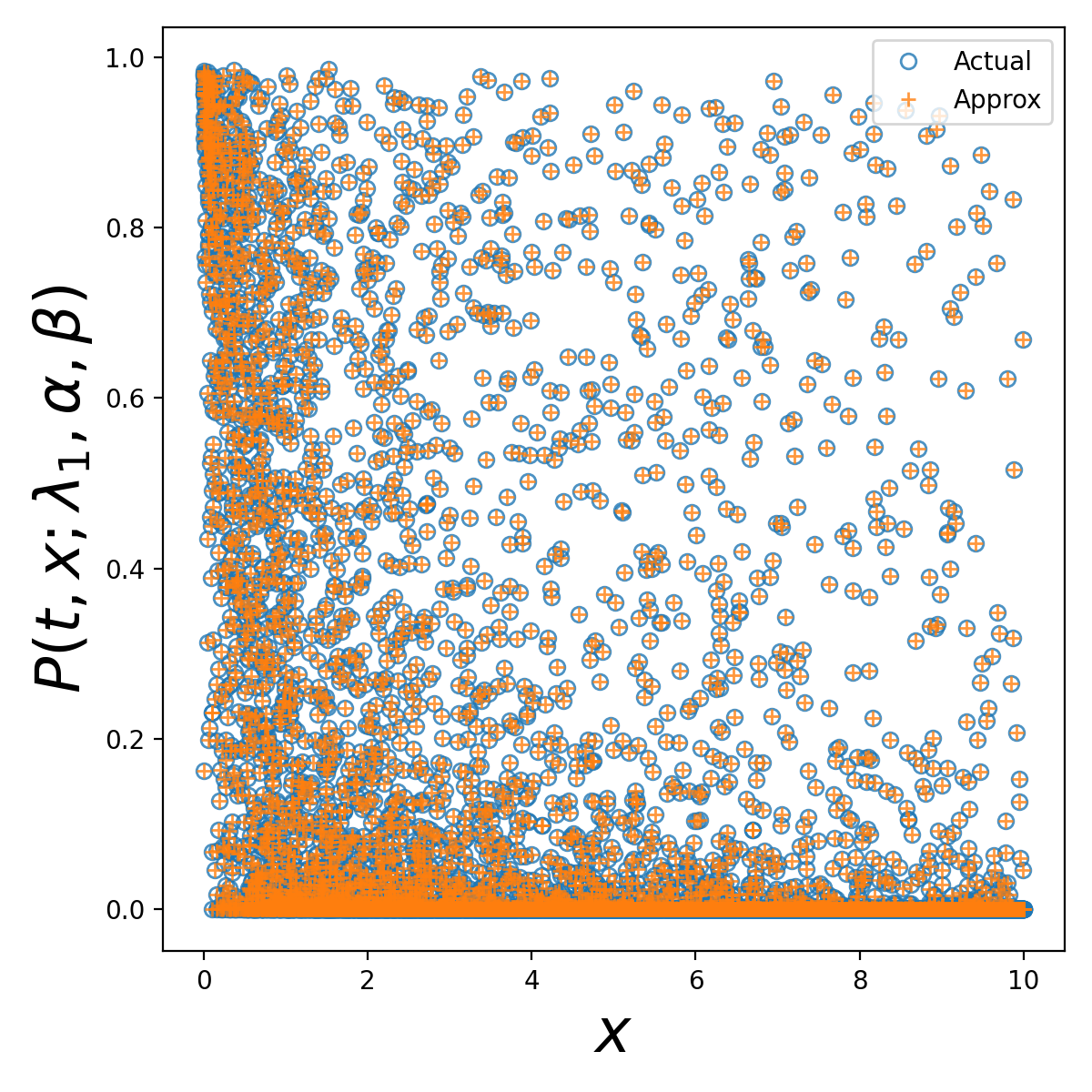}
  \includegraphics[scale=0.3]{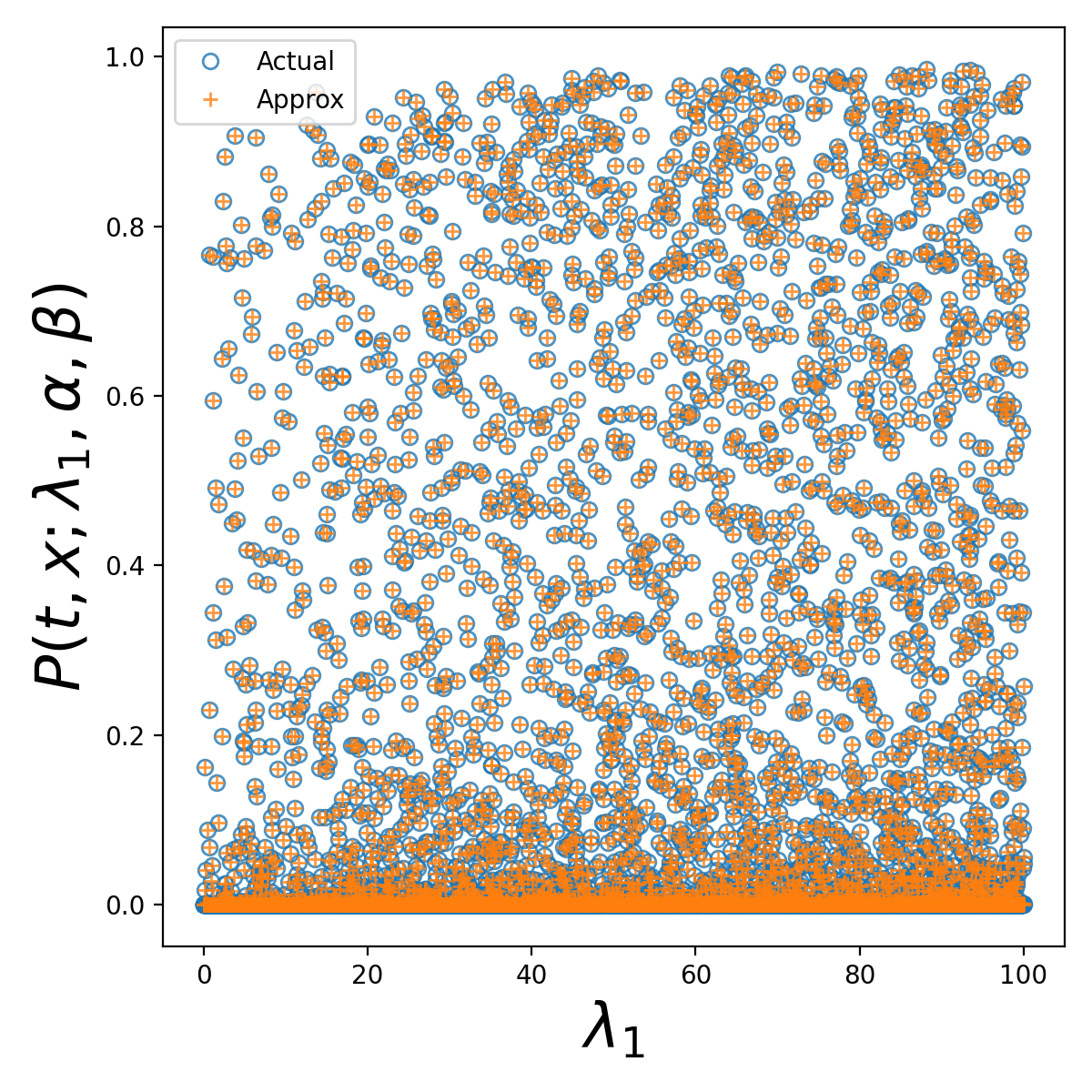}
  \caption{Spline approximation results for default probability (\ref{3.2}) with (\ref{3.3})}
  \label{fig:1}
\end{figure}

More specifically, to generate a reasonably large training set, we uniformly sample $\mathfrak{n}=10,000$ points based on the value ranges $\lambda_{1}\in(0,100)$, $\alpha\in\mathbb{Z}\cap[1,5]$, $\beta\in(0,150)$, $t\in[0,10]$, and $x\in(0,10]$ (for $\overline{J}$). Evaluating $P(t,x)\equiv P(t,x;\lambda_{1},\alpha,\beta)$ on each sample point gives rise to 10,000 function values accordingly. In particular, these ranges closely align with the economic interpretations of the parameters and are broad enough to accommodate extreme market conditions, such as when frequent negative shocks to the CET1 ratio occur during market downturns. Let us denote by $\vartheta_{i}=(t,x;\lambda_{1},\alpha,\beta)_{i}$ the $i$th sample point ($i\leq\mathfrak{n}$), understood as a 5-vector row, and $p_{i}$ the corresponding $i$th function value. Then, the RBF interpolation aims to find an approximation for the default probability $P$ of the following form:
\begin{equation}\label{6.1.1}
  \breve{P}(t,x;\lambda_{1},\alpha,\beta)=\sum_{i=1}^{\mathfrak{n}}\ell_{i}\varphi_{\epsilon}(\|(t,x;\lambda_{1},\alpha,\beta)-\vartheta_{i}\|),
\end{equation}
where $\ell_{i}$'s are to-be-determined weights and $\varphi_{\epsilon}$ is a radial basis function with smoothing parameter $\epsilon$, which depends on the Euclidean distance between the generic point $(t,x;\lambda_{1},\alpha,\beta)$ and the sample point $\vartheta_{i}$. Here, we choose specifically the multi-quadratic kernel, having the form $\varphi_{\epsilon}(x)=\sqrt{1+(x/\epsilon)^{2}}$, $x\geq0$, in which the smoothing parameter $\epsilon$ is determined as the average distance among all the sample points. The weights $\ell_{i}$'s are determined by the linear system
\begin{equation*}
  p_{j}\equiv\breve{P}(\vartheta_{j})=\sum^{\mathfrak{n}}_{i=1}\ell_{i}\varphi_{\epsilon}(\|\vartheta_{j}-\vartheta_{i}\|),\quad j\in\mathbb{Z}\cap[1,\mathfrak{n}],
\end{equation*}
ensuring nearly perfect match to the function values $p_{j}$'s. The approximation results are visualized in Figure \ref{fig:1}.

On average, each evaluation of $P(t,x;\lambda_{1},\alpha,\beta)$ takes about 0.6 second (wall time),\footnote{All implementation scripts were written in Python and executed on a personal MacBook Pro (2024) equipped with an Apple M3 Pro chip and 36 GB of RAM. The fully processed data set and the corresponding Python code are available from the authors upon reasonable request.} while the spline approximation with (\ref{6.1.1}) takes less than 0.1 millisecond, yielding a 6,000-fold improvement in computational speed relative to the original function.

\medskip

\subsection{Case studies}\label{sec:6.2}

We implemented our estimation and calibration procedures following the details in Section \ref{sec:5}, employing maximum likelihood estimation for the CET1 ratio and stock price parameters, followed by calibration of the remaining model parameters against observed (equity-convertible) CoCo prices and, where applicable, CDS spreads. For the 2009--2011 period, due to the unavailability of CET1 ratio data (largely because Basel III had not yet mandated the standardized reporting of such metrics), estimation of the associated parameters is carried out using Equation (\ref{5.1.6}) based directly on stock return data.\footnote{In this case, we set $E^{\ast}_{t_{0}}\equiv1$, $\tilde{E}_{t_{0}}\equiv1$, and the independent parameter $\varpi=1$ during calibration, signifying the absence of regulatory intervention.}

Table \ref{tab:2} reports the estimated or calibrated parameter values in each case study, with the RMSE computed based on (\ref{5.2.2}), measuring absolute pricing errors. Similar to Wilkens and Bethke (\citeyear{WB14}) \text{Tab.} 3, given the model-predicted CoCo prices, we computed hedging errors using only delta and gamma hedging -- as mentioned in Section \ref{sec:4.2} -- with the underlying stock being the sole hedging instrument in each case. The average and total hedging errors are reported in Table \ref{tab:2} as well. Figure \ref{fig:2a}, Figure \ref{fig:2b}, Figure \ref{fig:2c}, Figure \ref{fig:2d}, and Figure \ref{fig:2e} illustrate the model performance across all five case studies, comparing the model-implied density functions for the CET1 ratio and log-returns against the corresponding kernel density estimates, and contrasting the time series of model-predicted and observed CoCo prices.


\begin{table}[H]\small
  \centering
  \caption{Cross-comparison of estimation and calibration results}
  \label{tab:2}
  \begin{tabular}{c|ccccc}
    \toprule
    \textbf{Parameter estimation} & \makecell{LYG \\ (2021--2023)} & \makecell{LYG \\ (2009--2011)} & \makecell{CS \\ (2020--2023)} & \makecell{CS \\ (2011)} & \makecell{CCB \\ (2019--2023)} \\
    \midrule
    Solvency shock intensity ($\hat{\lambda}_1$) & 21.6405 & 8.1326 & 32.5280 & 16.7256 & 22.2757 \\
    Solvency shock shape ($\hat{\alpha}$) & 1 & 4 & 3 & 3 & 1 \\
    Solvency shock rate ($\hat{\beta}$) & 22.4895 & 46.0099 & 77.9160 & 39.4584 & 47.4683 \\
    Stock volatility ($\hat{\sigma}$) & 0.2062 & 0.3651 & 0.3089 & 0.2050 & 0.1347 \\
    Stock jump intensity ($\hat{\lambda}_2$) & 34.7319 & 79.1058 & 31.9521 & 99.3929 & 32.8870 \\
    Stock jump mean ($\hat{\mu}_{V}$) & $-0.0017$ & $-0.0019$ & $-0.0003$ & $-0.0011$ & $-0.0024$ \\
    Stock jump scale ($\hat{\sigma}_{V}$) & 0.0299 & 0.0365 & 0.0643 & 0.0380 & 0.0309 \\
    Stock-CET1 leverage ($\hat{\eta}$) & 0.2542 & 0.2196 & 0.7412 & 0.0839 & 0.4900 \\
    Transformed trigger barrier ($\hat{\overline{J}}$) & 0.4780 & 0.0101 & 1.8732 & 0.6276 & 0.8358 \\
    Write-down fraction ($\hat{w}$) & 0.0000 & 0.4582 & 0.0001 & 0.0000 & 0.0000 \\
    Conversion power ($\hat{p}$) & 0.9797 & 0.9999 & 0.6235 & 1.0000 & 1.0000 \\ \hline
    \makecell{Regulatory intervention \\ modeling} & no & no & yes & no & no \\ \hline
    \textbf{RMSE (\%)} & 5.04 & 7.95 & 6.47 & 3.28 & 1.51 \\
    \textbf{Average hedging error (\%)} & 0.00 & 0.02 & 0.00 & 0.01 & 0.03 \\
    \textbf{Total hedging error (\%)} & 0.30 & 0.78 & 0.53 & 0.17 & 0.22 \\
    \bottomrule
    \end{tabular} \\
    \begin{tablenotes}
      \small
      \item Hedge: delta and gamma (shares)
    \end{tablenotes}
\end{table}



According to Table \ref{tab:2}, examining the estimation and calibration results for both Lloyds and Credit Suisse across the two periods respectively reveals notable insights into the bank's risk profile and the adaptability of our model. First, in the LYG case, for the 2021--2023 period, the estimated CET1 ratio parameters indicate a moderate frequency of solvency shocks ($\hat{\lambda}_1=21.6405$), with single-stage (namely, exponential) shock magnitudes ($\hat{\alpha}=1$), while the relatively high value $\hat{\beta}=22.4895$ particularly suggests that these shocks tend to be moderate in size. The estimated stock price parameters speak to significant volatility ($\hat{\sigma}=0.2062$) with a substantial idiosyncratic jump component ($\hat{\lambda}_2=34.7319$), albeit with relatively small jump scales ($\hat{\sigma}_V=0.0299$). The leverage parameter estimate ($\hat{\eta}=0.2542$) implies strong association between Lloyds' stock price and CET1 ratio fluctuations over the two-year period, also providing empirical validation for the model's foundational premise -- that stock prices and regulatory capital metrics exhibit measurable interdependence. The calibrated power coefficient ($\hat{p}=0.9797$) in the conversion mechanism is remarkably near unity, which means that the LYG CoCo bond conversion terms are predominantly based on the initial stock price rather than the price at the time of default, which, as discussed before, effectively shields bondholders from excessive dilution during periods of financial stress; the zero write-down fraction ($\hat{w}=0.0000$) suggests full conversion into equity with no principal loss upon trigger.

During the second, earlier period (2009--2011), which encompasses the aftermath of the 2008 financial crisis, we observe markedly different parameter values. The CET1 ratio exhibits lower shock frequency ($\hat{\lambda}_1=8.1326$) but higher shock complexity ($\hat{\alpha}=4$), hinting at the presence of multi-stage solvency issues. The stock price parameters show higher volatility ($\hat{\sigma}=0.3651$) and more than double the jump intensity ($\hat{\lambda}_2=79.1058$) compared to the 2021--2023 period covering the COVID-19 pandemic. This particularly reflects that the 2008 financial crisis triggered severer turbulence in the CoCo bond market (especially in its early stages) than the global pandemic, mainly due to its banking-centered nature affecting the foundation of the bond issuers. The significant write-down fraction estimate ($\hat{w}=0.4582$) contrasts with the previous estimate as well, suggesting that earlier LYG CoCo designs incorporated substantial principal loss upon trigger -- a feature largely phased out in later issuances as the market matured and investor preferences evolved. The particularly low (transformed) trigger (upper) barrier ($\hat{\overline{J}}=0.0101$) implies these early CoCos were designed with substantial buffer against conversion, associated with both regulatory uncertainty and issuer caution during the formative stages of the market.

The two Credit Suisse (CS) cases offer particularly valuable perspectives, especially in light of the 2023 collapse. For the 2020--2023 period, the estimated-calibrated CS parameters reveal a different risk profile compared to the LYG parameters. The CET1 ratio exhibits both higher solvency shock frequency ($\hat{\lambda}_1 = 32.5280$) and substantially greater complexity ($\hat{\alpha} = 3$), indicating multi-stage solvency deterioration patterns. The higher value of the scale estimate $\hat{\beta} = 77.9160$ suggests that while shocks were frequent, they were also severe in magnitude, a pattern that our model successfully captures in the deteriorating trend leading to the bank's terminal decline. The most revealing in this case study are the intervention parameters, which were calibrated specifically for this case due to the 2023 collapse. The diffusion component of intervention intensity is near zero ($\hat{\kappa}^{(1)} \approx 0$), coupled with a strong negative correlation with stock price movements ($\hat{\varsigma}^{(1)} = -0.0821$), indicating that regulatory intervention was predominantly driven by jump events rather than continuous deterioration. The high jump-induced intervention sensitivity ($\hat{\varsigma}^{(2)}\approx 1$) and minimal mean reversion ($\hat{\kappa}^{(2)} \approx 0$) suggest that once intervention risk materialized, it remained persistently elevated -- a finding consistent with the rapid regulatory action that ultimately led to the Credit Suisse resolution, as documented in the chronology by Damyanova (\citeyear{D23}) and analyzed by Bolton, Jiang, and Kartasheva (\citeyear{BJK23}).

The CoCo-specific parameters reveal a substantially more complex default price structure ($\hat{p} = 0.6235$), which contrasts with LYG's simpler design, despite a comparably low write-down fraction ($\hat{w} = 0.0001$). This configuration exposed bondholders to maximum losses during stress scenarios, precisely matching what occurred when CS's AT1 bonds were completely written off during the UBS acquisition. Our model's ability to capture this deteriorating trajectory through the combination of frequent solvency shocks and elevated intervention probabilities demonstrates its potential to identify systemic risk patterns that traditional approaches miss. The calibrated independence parameter $\hat{\varpi} = 1.0000$ implies that throughout the period, markets perceived the causal link between accounting triggers and regulatory intervention to be relatively weak, consistent with the surprise element of the final regulatory decision documented by the European Banking Authority (\citeyear{EBA24}). The low estimate for the intervention shock scale ($\hat{\gamma}= 0.0212$) suggests minimal immediate stock price impact following regulatory intervention, highlighting the significant time delay between sharp bond price collapses and subsequent equity market reactions (observe the right plot in Figure \ref{fig:2c}). Such a pattern has been noted in the broader literature on bank resolution mechanisms (Financial Stability Board, \citeyear{FSB21}).

The earlier Credit Suisse data from 2011 show moderate solvency shock patterns ($\hat{\lambda}_1= 16.7256$ and $\hat{\alpha} = 3$) but with much higher stock price jump intensity ($\hat{\lambda}_2= 99.3929$) compared to the LYG case during the same period. Notably, the weak correlation between stock price and CET1 ratio ($\hat{\eta} = 0.0839$) -- nearly one-fourth of LYG's value -- points to a fundamental weakness in CS's risk structure that persisted over time: Stock prices and regulatory capital moved largely independently, thereby reducing the early warning capacity of market signals. This disconnect between market indicators and regulatory metrics, which our model quantifies through the leverage parameter, is consistent with findings by Martynova and Perotti (\citeyear{MP18}) on the limited effectiveness of market discipline in European banking. The conversion structure with $\hat{w} = 0.0000$ and $\hat{p} = 1.0000$ for these early CoCo bonds represents a design philosophy focused on capital preservation rather than loss absorption, which may explain why market participants failed to adequately price the risks that eventually materialized, as suggested by the broader literature on CoCo bond pricing challenges (see, e.g., Glasserman and Nouri, \citeyear{GN16}; Khah, Vermaelen, and Wolff, \citeyear{KVW19}).

On the other hand, analysis of the CCB case offers an interesting counterpoint to the Western banks. Based on Table \ref{tab:2}, estimates of the bank's CET1 ratio parameters ($\hat{\lambda}_1=22.2757$, $\hat{\alpha}=1$, and $\hat{\beta}=47.4683$) indicate frequent yet simple (single-stage) solvency shocks with moderate magnitude. The stock price exhibits lower volatility ($\hat{\sigma}=0.1347$) but relatively high stock-CET1 leverage ($\hat{\eta}=0.4900$), which suggests that despite overall stability, market prices remained sensitive to capital fluctuations. Also, the implied purely initial price-based conversion mechanism ($\hat{w}=0.0000$ and $\hat{p}=1.0000$) and a higher trigger barrier ($\hat{\overline{J}}=0.8358$) points to a more conservative CoCo structure, with potential reflection of different regulatory priorities in the Chinese banking system.

Overall speaking, the estimated and calibrated parameters across different time periods highlight a clear evolution in CoCo design, with earlier issues typically featuring higher write-down fractions and varying conversion mechanisms, which likely reflects the market's improvement process and regulatory refinements following the initial post-2008 implementations. In Table \ref{tab:2}, the stark differences in parameter estimates among banks, particularly in the CET1 ratio shock patterns ($\hat{\alpha}$) and stock-capital correlations ($\hat{\eta}$), highlight how the comprehensive model framework captures institution-specific risk characteristics that uniform approaches would miss. The CS case particularly demonstrates the model's capacity to identify intervention risk patterns, with the calibrated parameters suggesting predominantly jump-driven intervention mechanisms, consistent with how regulatory actions typically unfold within crisis scenarios. In addition, variations in the estimated power coefficient ($\hat{p}$) and write-down fraction ($\hat{w}$) across banks can be ascribed to their different approaches to balancing bondholder and shareholder interests, with some designs favoring capital preservation ($w$ low but $p$ high) and others emphasizing loss absorption capacity ($w$ high but $p$ low). This observation also agrees with the original motivation for introducing the power conversion mechanism.

\clearpage

\vspace*{0.2in}

\begin{figure}[H]
  \ContinuedFloat*
  \begin{minipage}[r]{0.33\linewidth}
  \centering
  \includegraphics[scale=0.3]{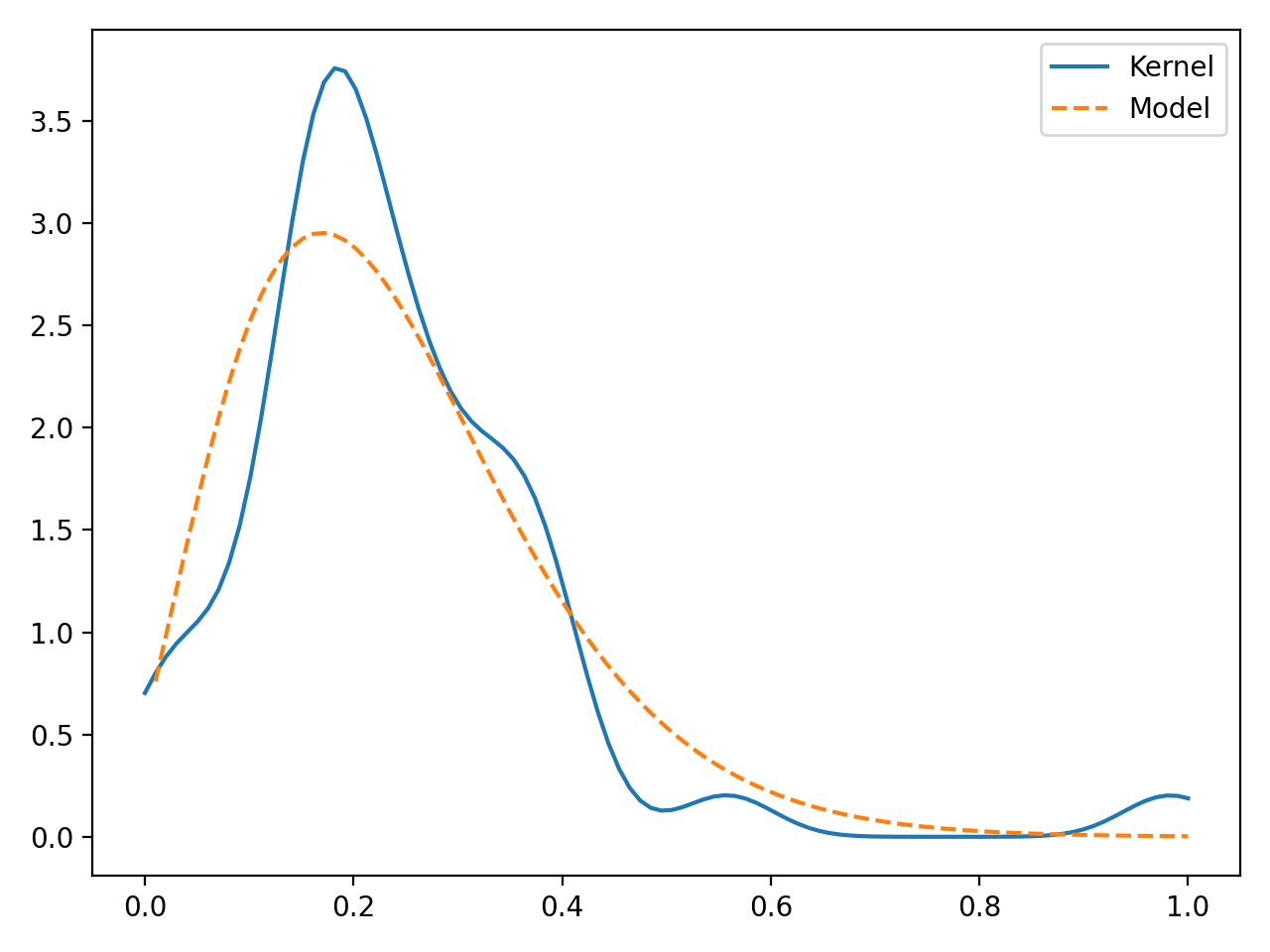}
  \caption*{\small Transformed CET1 ratio density}
  \end{minipage}
  \begin{minipage}[c]{0.33\linewidth}
  \centering
  \includegraphics[scale=0.3]{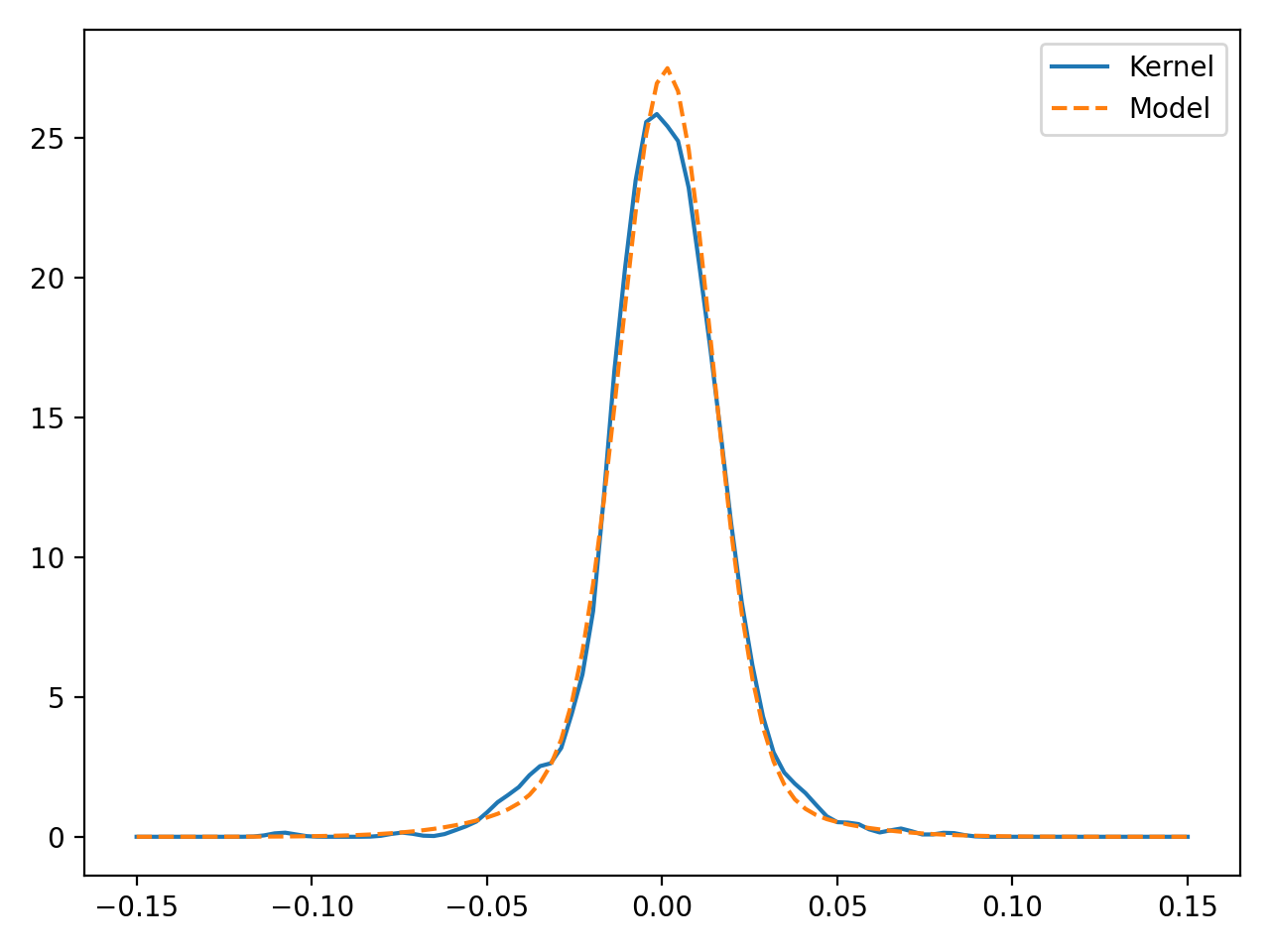}
  \caption*{\small log-return density}
  \end{minipage}
  \begin{minipage}[l]{0.33\linewidth}
  \centering
  \includegraphics[scale=0.3]{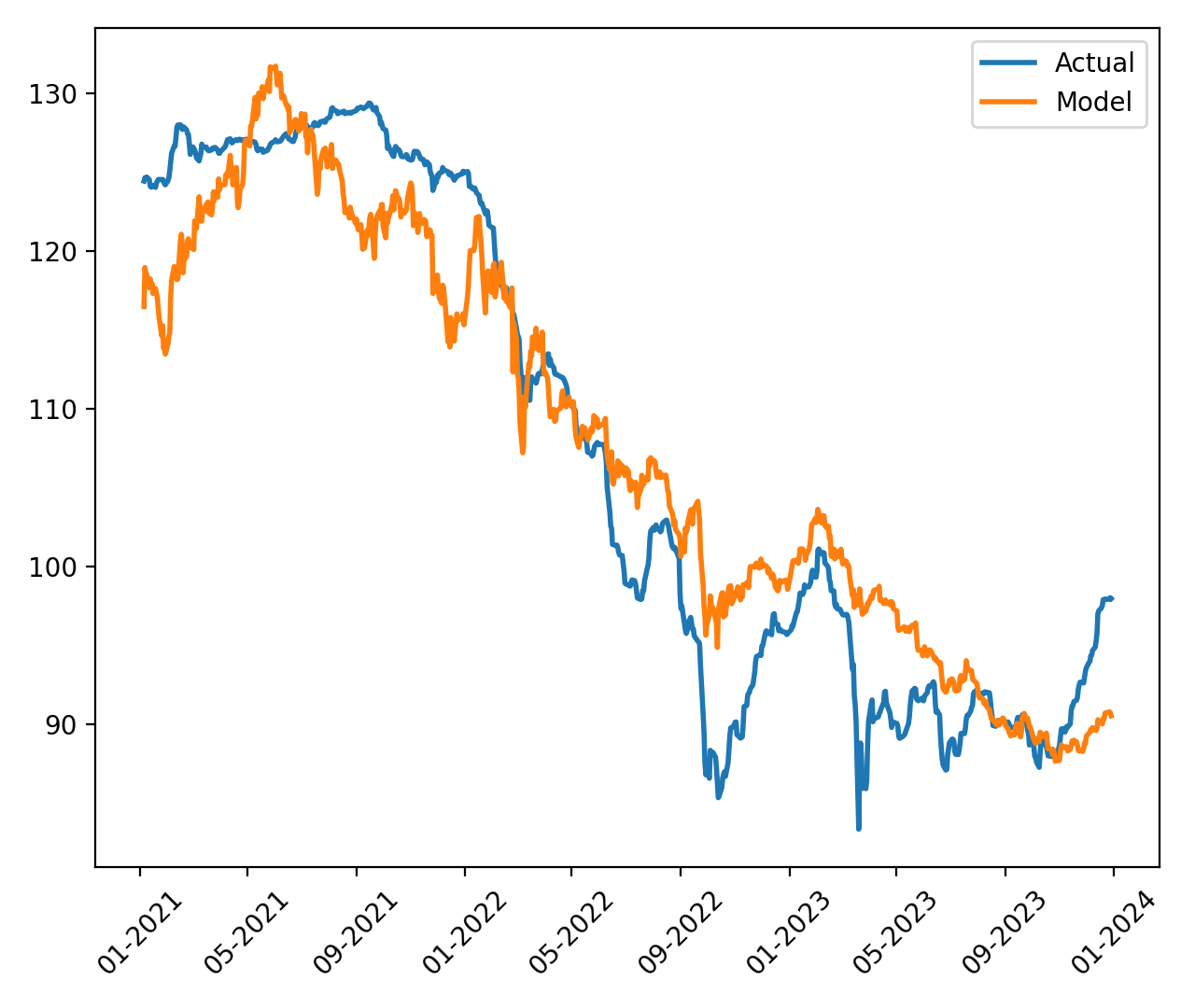}
  \caption*{\small CoCo price (clean, \%)}
  \end{minipage}
  \caption{Estimation and calibration results (LYG), CoCo price range: 01/04/2021--12/29/2023}
  \label{fig:2a}
\end{figure}

\begin{figure}[H]
  \ContinuedFloat
  \begin{minipage}[r]{0.49\linewidth}
  \centering
  \includegraphics[scale=0.3]{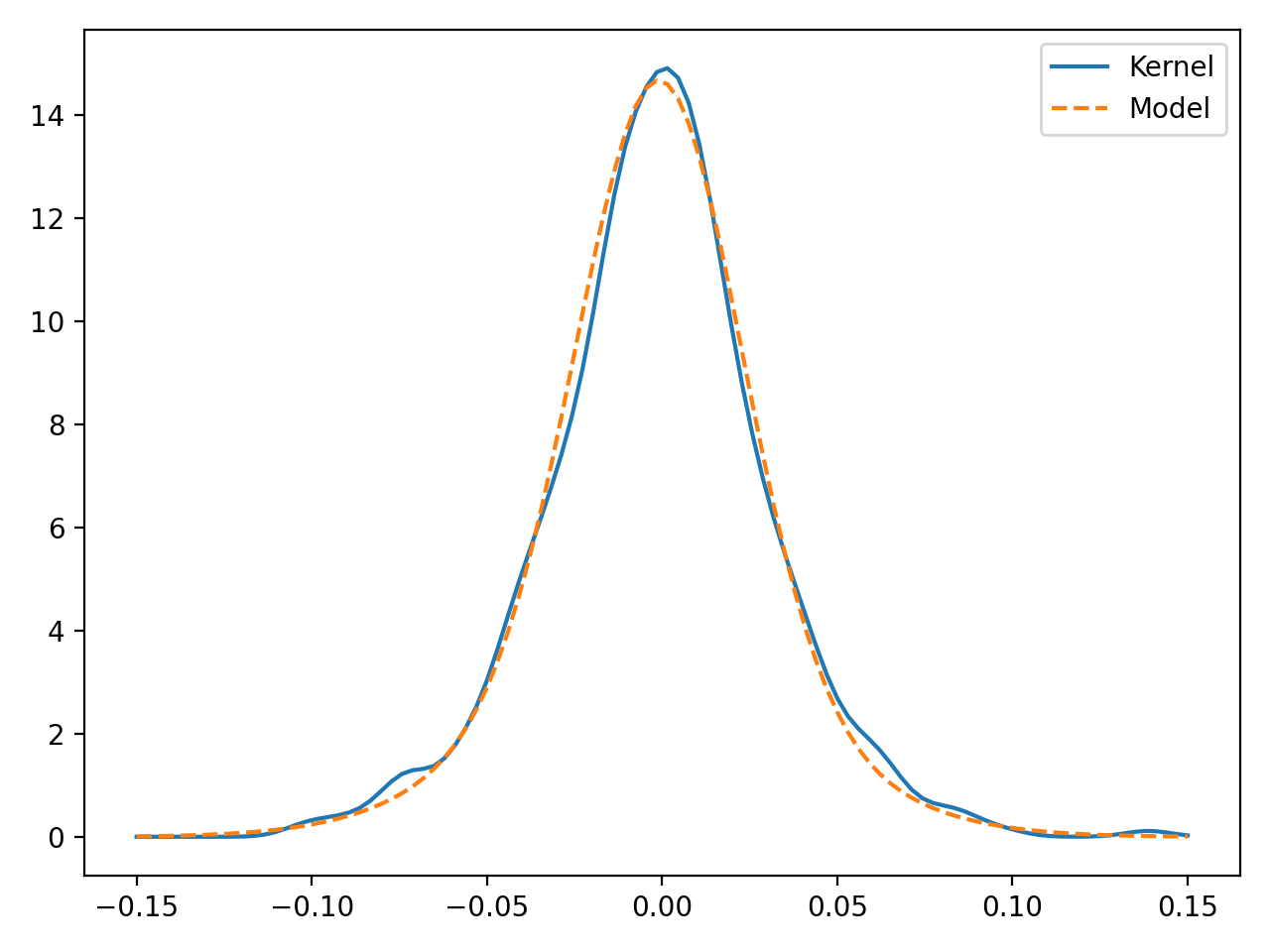}
  \caption*{\small Log-return density}
  \end{minipage}
  \begin{minipage}[l]{0.49\linewidth}
  \centering
  \includegraphics[scale=0.3]{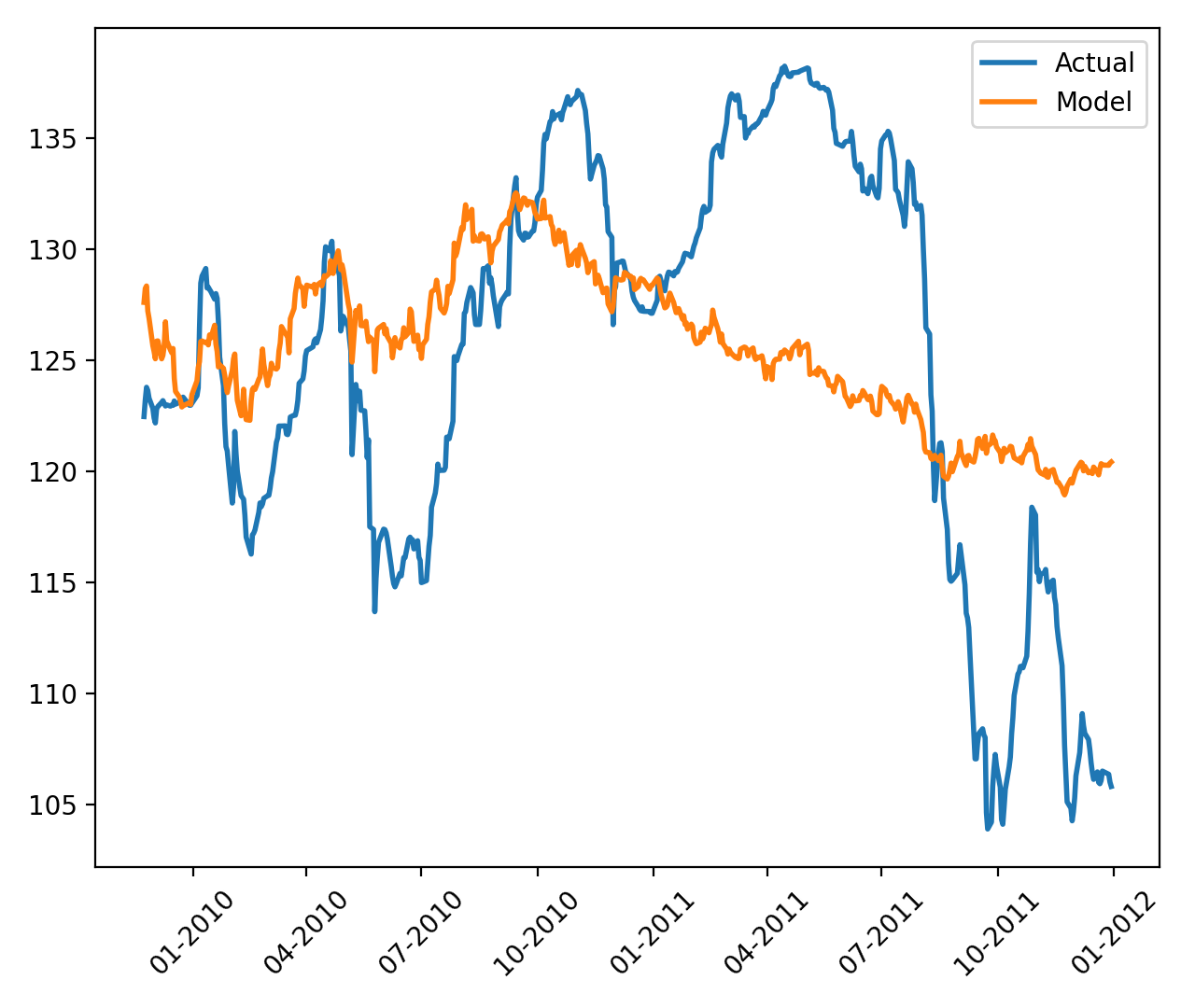}
  \caption*{\small CoCo price (clean, \%)}
  \end{minipage}
  \caption{Estimation and calibration results (LYG), CoCo price range: 11/23/2009--12/30/2011}
  \label{fig:2b}
\end{figure}

\begin{figure}[H]
  \ContinuedFloat
  \begin{minipage}[r]{0.33\linewidth}
  \centering
  \includegraphics[scale=0.3]{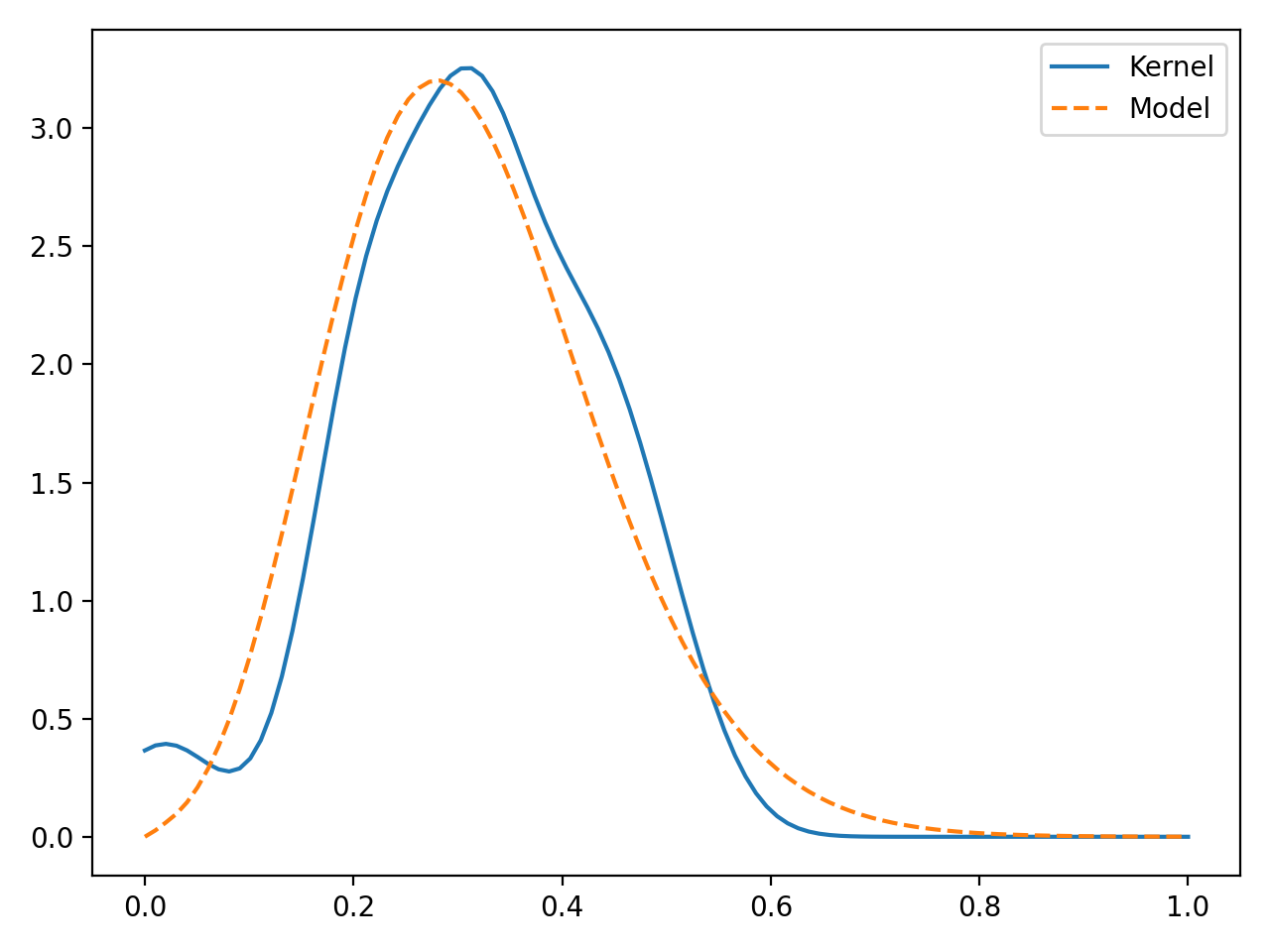}
  \caption*{\small Transformed CET1 ratio density}
  \end{minipage}
  \begin{minipage}[c]{0.33\linewidth}
  \centering
  \includegraphics[scale=0.3]{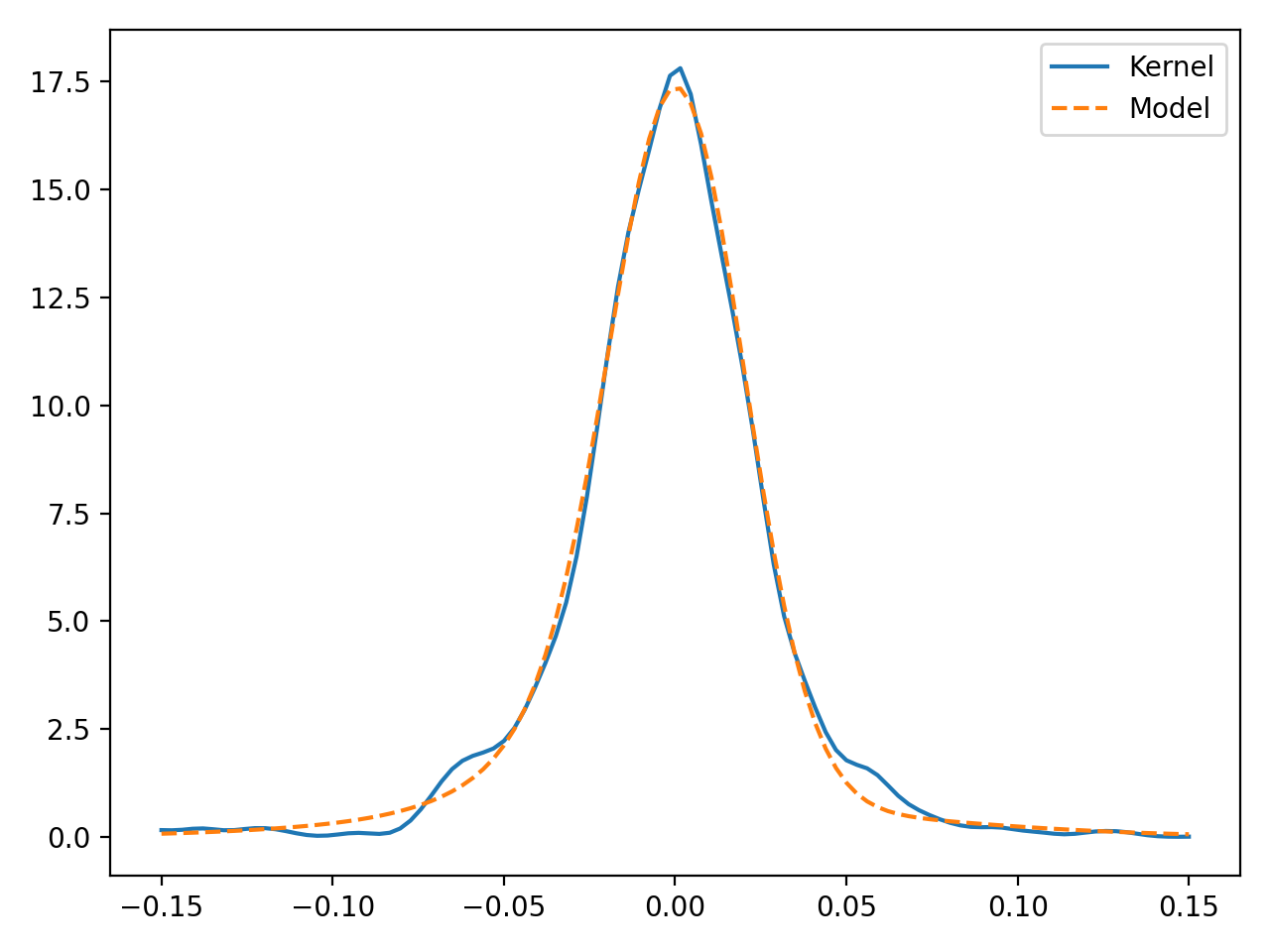}
  \caption*{\small log-return density}
  \end{minipage}
  \begin{minipage}[l]{0.33\linewidth}
  \centering
  \includegraphics[scale=0.3]{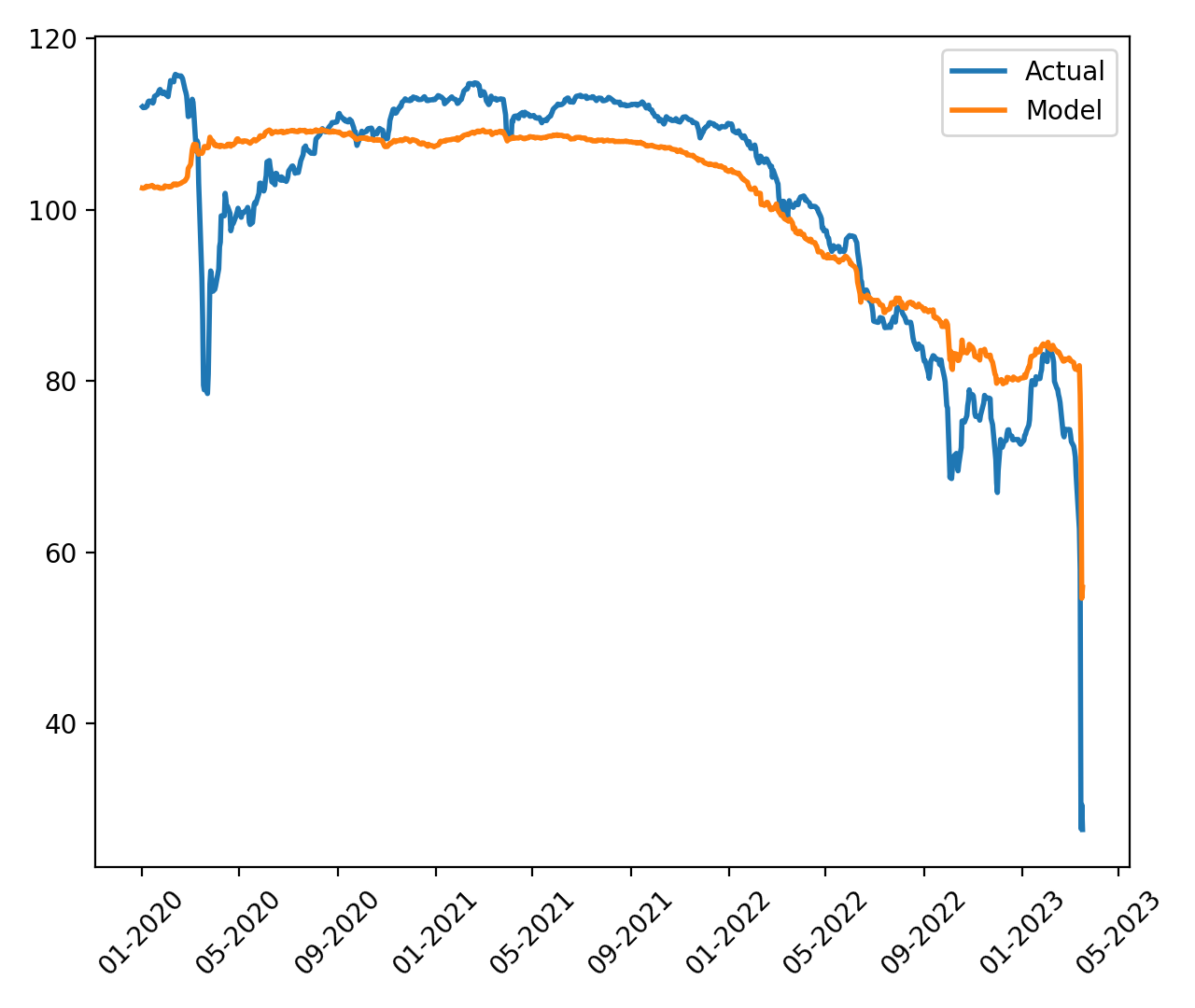}
  \caption*{\small CoCo price (clean, \%)}
  \end{minipage}
  \caption{Estimation and calibration results (CS), CoCo price range: 01/01/2020--03/17/2023}
  \label{fig:2c}
\end{figure}

\clearpage

\begin{figure}[H]
  \ContinuedFloat
  \begin{minipage}[r]{0.49\linewidth}
  \centering
  \includegraphics[scale=0.3]{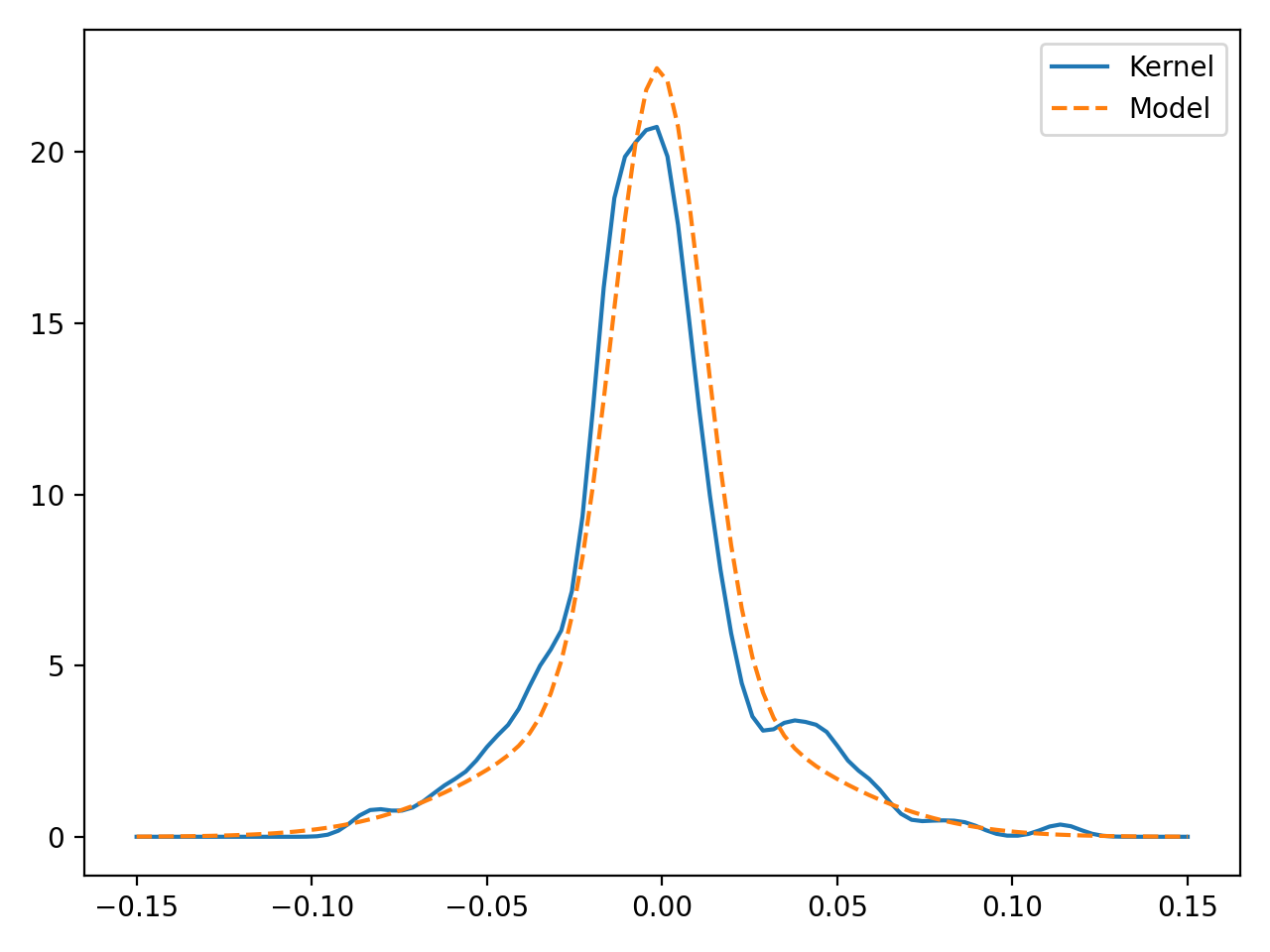}
  \caption*{\small Log-return density}
  \end{minipage}
  \begin{minipage}[l]{0.49\linewidth}
  \centering
  \includegraphics[scale=0.3]{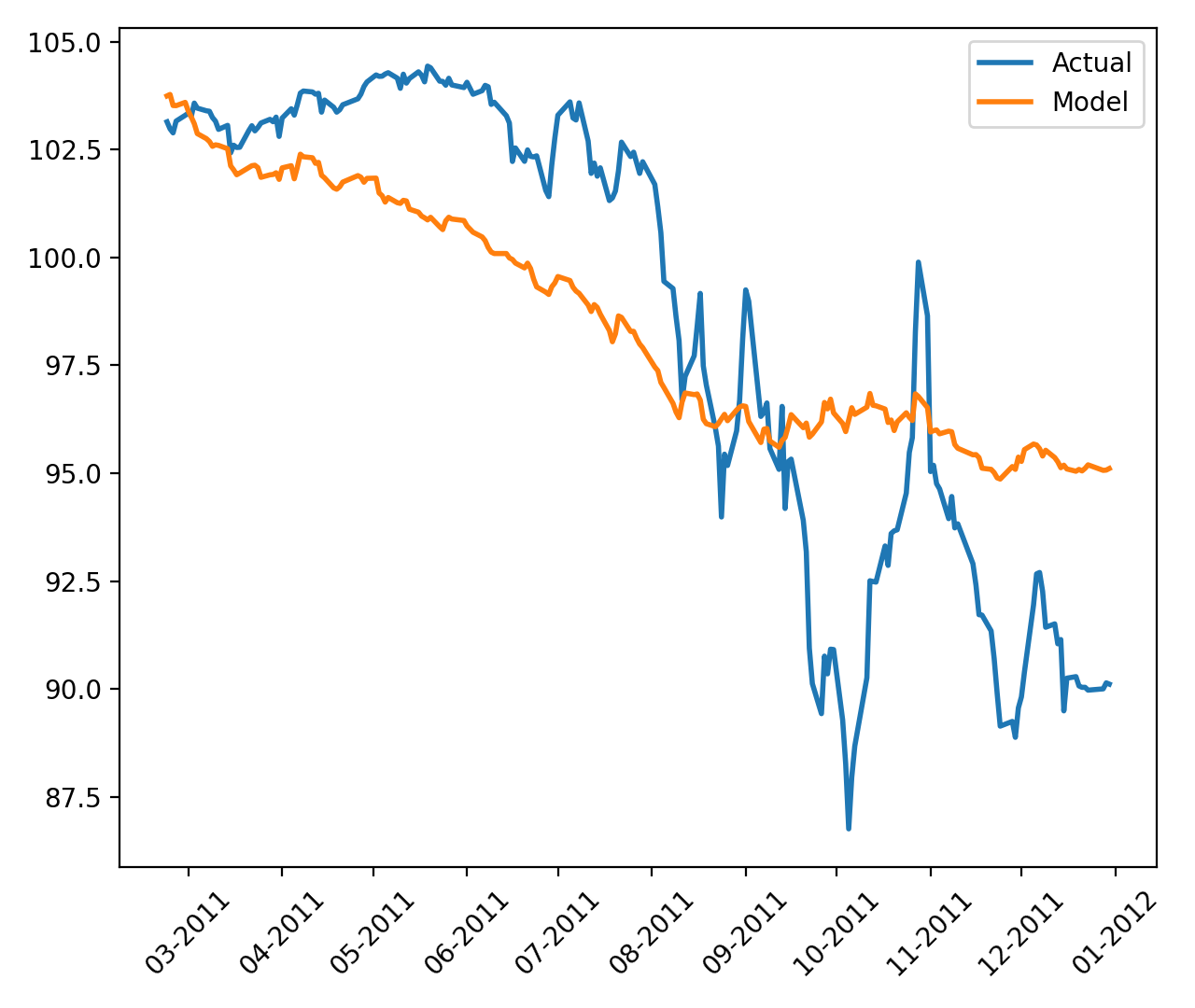}
  \caption*{\small CoCo price (clean, \%)}
  \end{minipage}
  \caption{Estimation and calibration results (CS), CoCo price range: 02/22/2011--12/30/2011}
  \label{fig:2d}
\end{figure}

\begin{figure}[H]
  \ContinuedFloat
  \begin{minipage}[r]{0.33\linewidth}
  \centering
  \includegraphics[scale=0.3]{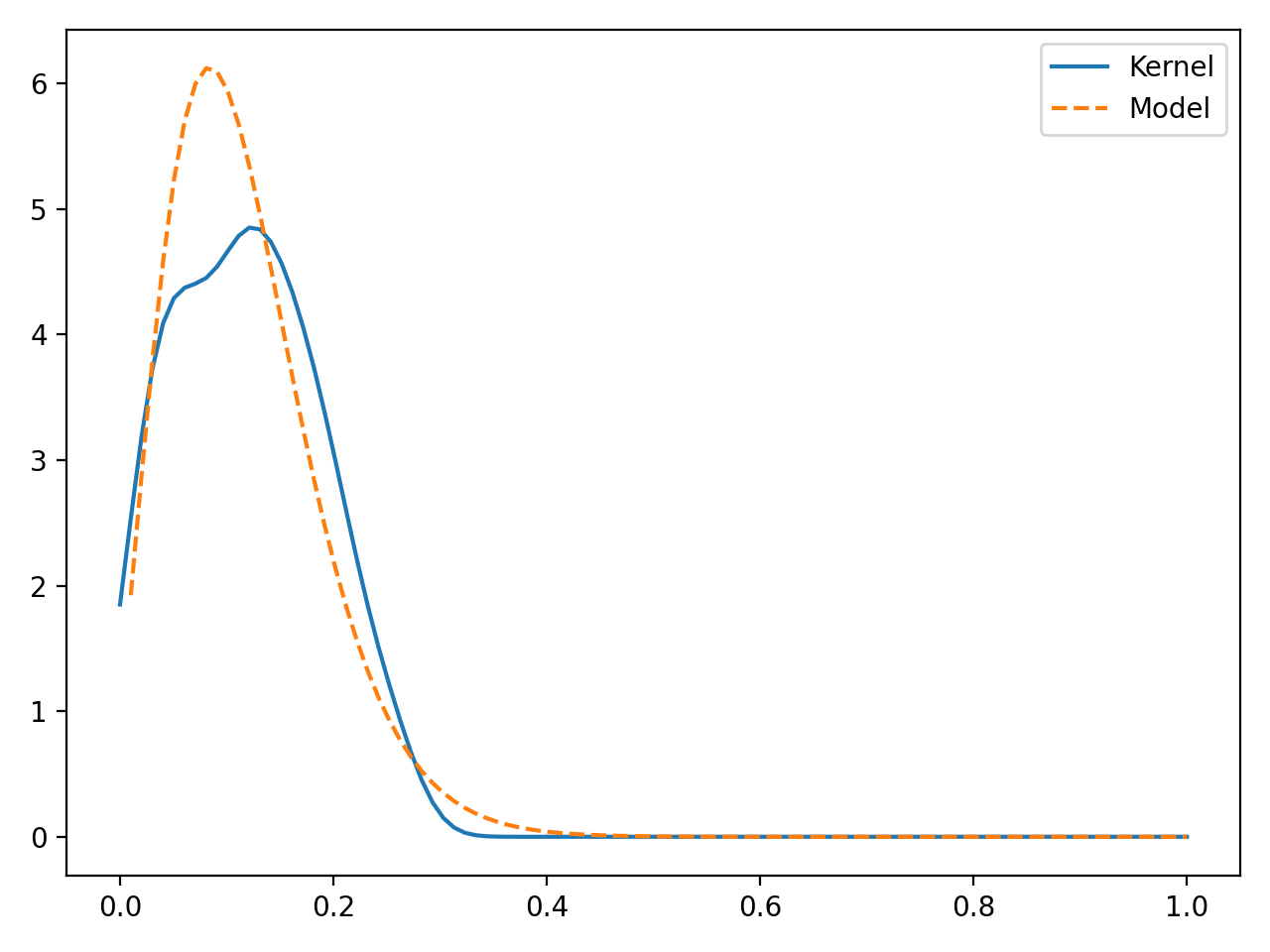}
  \caption*{\small Transformed CET1 ratio density}
  \end{minipage}
  \begin{minipage}[c]{0.33\linewidth}
  \centering
  \includegraphics[scale=0.3]{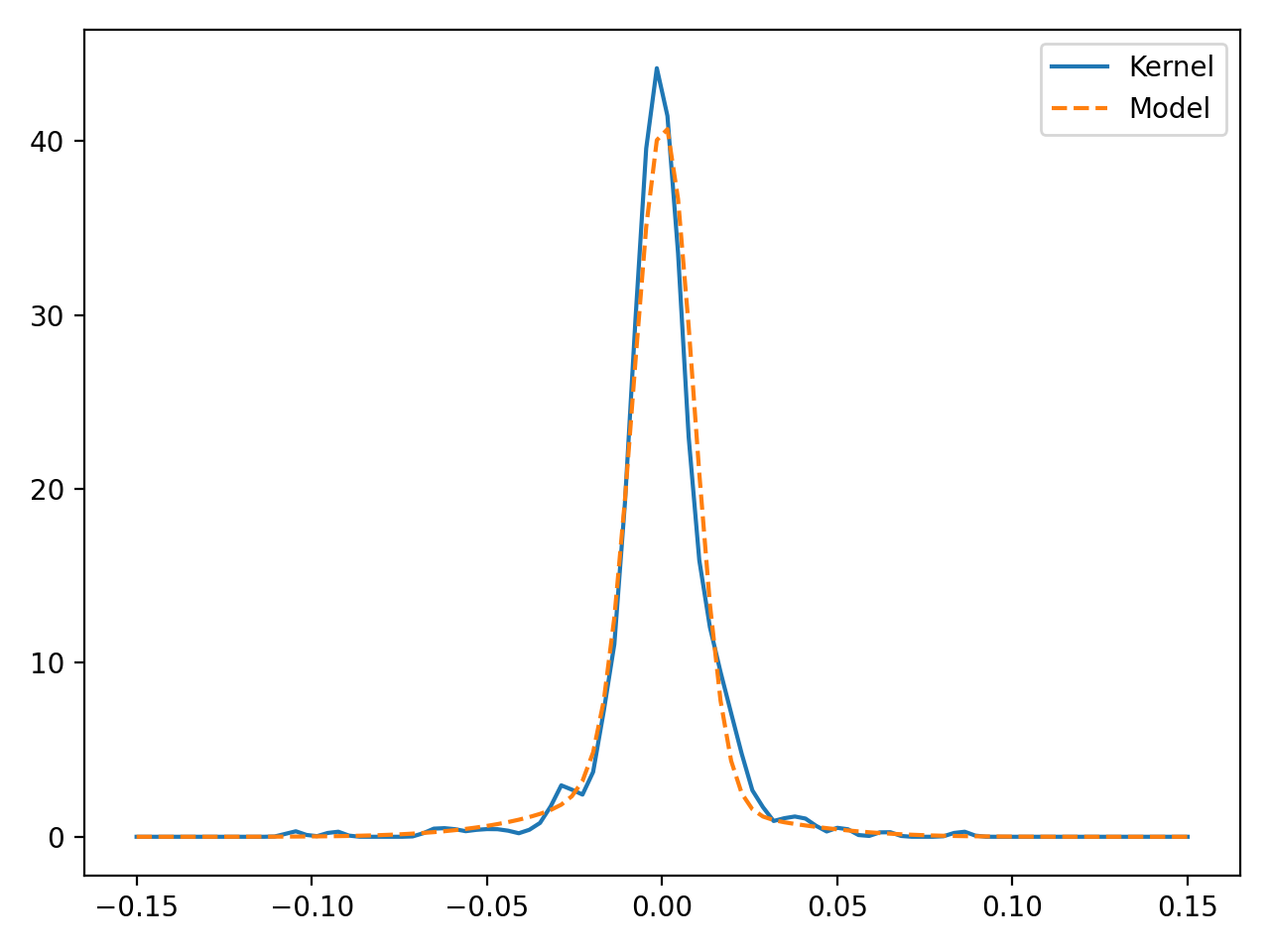}
  \caption*{\small log-return density}
  \end{minipage}
  \begin{minipage}[l]{0.33\linewidth}
  \centering
  \includegraphics[scale=0.3]{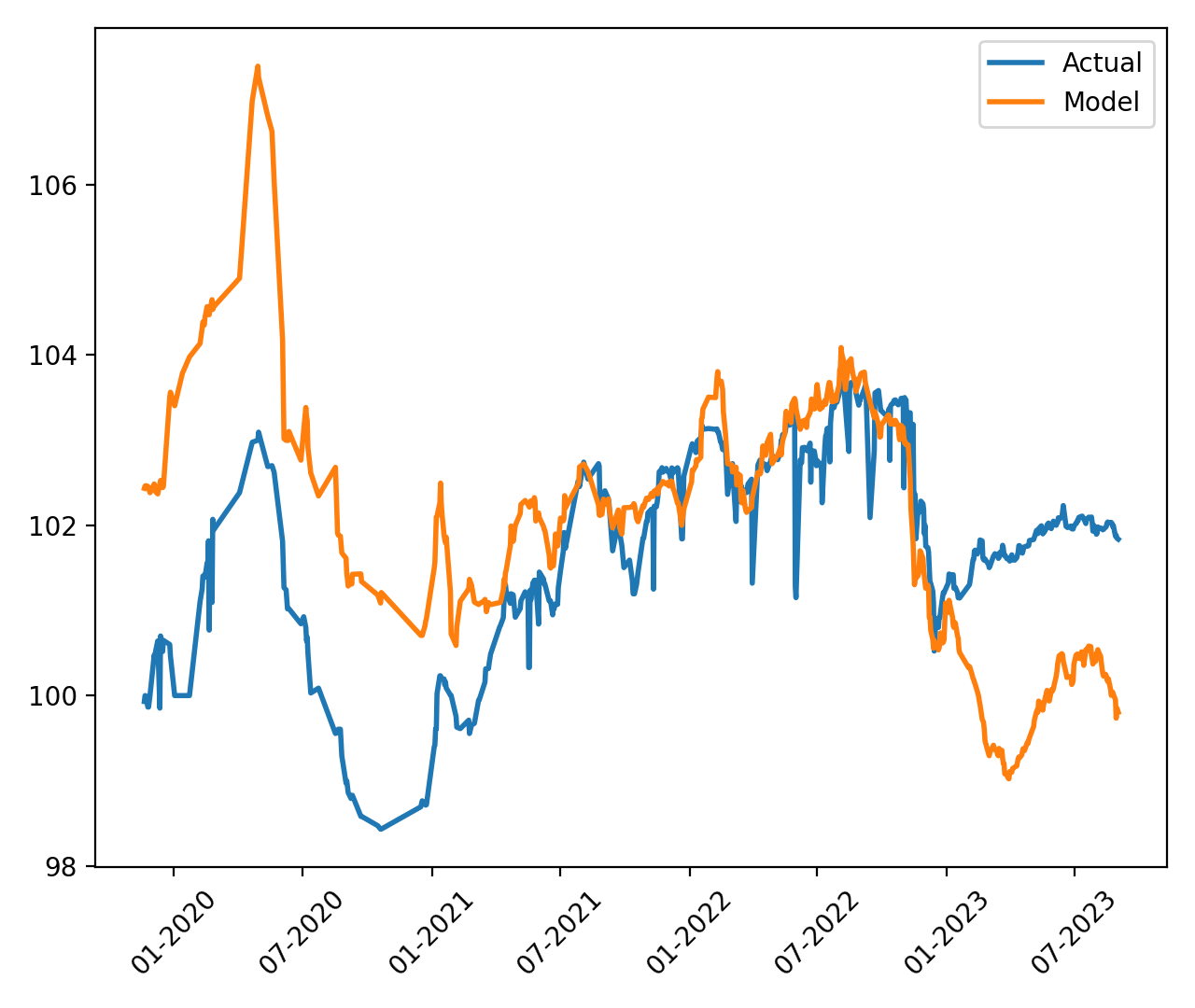}
  \caption*{\small CoCo price (clean, \%)}
  \end{minipage}
  \caption{Estimation and calibration results (CCB), CoCo price range: 11/20/2019--09/01/2023}
  \label{fig:2e}
\end{figure}

As shown in the last plots in Figures \ref{fig:2a} through \ref{fig:2e}, respectively, the empirical results confirm the strong overall fit of our model, which demonstrate significant improvements over existing approaches. While the RMSEs are generally low relative to bond price levels, in the 2011 Credit Suisse case, our model achieved an RMSE of 3.28\%, outperforming the credit derivatives modeling approach in the previous literature with its RMSE of 4.98\%; see Wilkens and Bethke (\citeyear{WB14}) \text{Tab.} 3. Similarly, for the 2009--2011 Lloyds case, our model delivers an RMSE of 7.95\% compared to the best RMSE of 11.32\% \text{ibid.}. These comparisons show a reduction of over 30\% in average pricing error.

Perhaps more importantly, beyond pricing accuracy, our framework exhibits significantly enhanced hedging capabilities across different metrics in that both average and total hedging errors are reduced by around 50\% (on average) relative to the modeling approaches examined in Wilkens and Bethke (\citeyear{WB14}). This substantial improvement in hedging performance directly translates to significantly lower rebalancing costs and more robust risk management under market stress, and given that we only used delta and gamma hedges, it enables institutional investors to maintain robust hedging portfolios across diverse market conditions, including during periods of heightened volatility when equity-capital correlation structures become increasingly unstable.

We remark that the above improvements stem from the following three main innovations: (1) The explicit modeling of the joint dynamics between CET1 ratios and stock prices, capturing correlations that drive conversion risk, (2) the inclusion of jump components in both processes, allowing for sudden changes in solvency status, and (3) flexible power-type conversion schemes, which accommodate diverse CoCo designs while maintaining analytical tractability.

The density plots in Figures \ref{fig:2a} through \ref{fig:2e} further show that the model managed to capture the key characteristics of solvency shock and stock return distributions, including a changeful but clear asymmetric leptokurtic feature, across different market conditions, which support the temporal evolution of CoCo prices. Particularly noteworthy is the model performance during periods of market stress, where traditional approaches exhibit the greatest pricing errors. Therefore, our findings suggest that accounting for the dynamic relationship between regulatory capital metrics and market prices is essential for accurate CoCo valuation and hedging, especially when financial instability intensifies. The significant improvement in pricing accuracy beyond doubt marks the importance of incorporating both accounting-based triggers and market-driven factors in developing CoCo valuation methods.

\bigskip

\section{Concluding remarks}\label{sec:7}

This paper has introduced a novel valuation framework for CoCo bonds that directly addresses several critical limitations in existing approaches. By building a bivariate jump-diffusion model that explicitly captures the dynamic relationship between banks' CET1 ratios, share prices, and CoCo bond prices, we put forward a comprehensive approach to understanding the complex conversion mechanisms that drive the behavior of these hybrid instruments. The proposed framework represents an advance over previous modeling approaches through three principal contributions.

First, our model successfully bridges the temporal discrepancy between high-frequency market data and low-frequency regulatory reporting by establishing an explicit mathematical connection (Equation (\ref{2.3.1})) between the CET1 ratio and contemporaneous stock price dynamics. This connection encompasses both continuous diffusion processes and correlated jump risks, allowing for more accurate representation of market dynamics during both normal and stressed conditions. The analytical tractability of our approach, based on the valuation formulas derived for both write-down and equity-convertible CoCos in Theorem \ref{thm:1} and Theorem \ref{thm:2} -- thanks to the semi-closed form distribution presented in Proposition \ref{pro:3} and Proposition \ref{pro:4} -- lays the foundation for computational efficiency without sacrificing model flexibility.

Second, the introduction of a power-type conversion scheme in (\ref{4.2.1}) represents a significant generalization of existing terms, offering greater adaptability to diverse CoCo designs while maintaining mathematical tractability. This innovation allows practitioners to model a spectrum of conversion mechanisms covering conversion ratios based on the geometric average of the price at inception and the price at default, which accommodate the evolving design preferences observed across different issuers and regulatory regimes. Specifically, the power coefficient provides a parsimonious way to reproduce the functionalities of complex conversion terms that require much more convoluted calculations (see again Appendix \ref{B}).

Third, the proposed framework for modeling regulatory discretion in trigger decisions has proven particularly valuable in light of the 2023 Credit Suisse collapse, where traditional accounting-based triggers failed to anticipate regulatory intervention. By employing a stochastic intensity structure for the intervention propensity that incorporates both diffusive components and jump-induced reactions, we have created an intervention mechanism to capture the complex, often nonlinear relationship between market indicators and regulatory actions. The empirical results in Section \ref{sec:6.2} demonstrate the significant enhancement of both pricing accuracy and hedging performance compared to various existing models.

The case studies spanning three major banking institutions across multiple time periods validate the robustness of our model framework under dissimilar market conditions, with substantial improvement in pricing accuracy -- up to 35\% reduction in RMSE compared to benchmark models, and especially encouraging is the superior performance during periods of market stress, precisely when timely and accurate valuation is most critical for risk management. These aspects together underline the practical value of this hybrid valuation approach for market participants. Besides, the empirical analysis reveals several important patterns in CoCo bond development. We observe a clear evolution in design features, with earlier issuances typically featuring higher write-down fractions and varying conversion mechanisms compared to more standardized recent structures. The major differences in the calibrated parameters among banks in turn justify the importance of bank-specific risk characteristics in determining CoCo bond valuations. The Credit Suisse case also demonstrates the possibility of identifying intervention risk patterns consistent with actual regulatory actions amid crisis.

Still, several limitations and avenues for future research merit consideration. While our model significantly improves on existing approaches, the task of forecasting regulatory actions during unprecedented systemic events remains inherently difficult. Future work could for instance explore the incorporation of additional macro-prudential indicators or network effects across financial institutions to further enhance the predictive capacity for systemic intervention scenarios. Also, as regulatory frameworks continue to evolve in response to the Credit Suisse episode, ongoing refinement of intervention propensity models will be understandably necessary.

From a practical standpoint, the proposed model framework provides guidance for the optimal design of CoCo bond terms, hence helping issuers balance loss-absorption capacity with market attractiveness. For investors, the data-adaptive nature and enhanced pricing accuracy support more informed risk assessment and effective hedging strategies. Meanwhile, for regulators, the model's ability to capture the relationship between market indicators and capital metrics provides a potential early warning mechanism for detecting emerging financial stress within institutions.

The events of March 2023 have demonstrated that CoCo bonds remain a vital yet complex component of the global banking capital structure. As these instruments keep evolving, there is an increasing need for valuation approaches that effectively account for both accounting-based triggers and market-driven factors. In this regard, the present work has provided a mathematically grounded as well as practically implementable framework for valuing and hedging CoCo bonds under both mechanical triggers and possible regulatory intervention.

\bigskip



\bigskip

\begin{appendices}

\section{Mathematical proofs}\label{A}

\renewcommand{\theequation}{A.\arabic{equation}}

\textbf{Proposition \ref{pro:1}.}

\begin{proof}
By the law of total probability, for generic $t\geq0$,
\begin{equation*}
  f_{J_{t}}(x)=\sum^{\infty}_{n=0}f_{J_{t}|N_{t}}(x|n)\PP[N_{t}=n],
\end{equation*}
in which $\PP[N_{t}=n]=e^{-\lambda_{1}t}(\lambda_{1}t)^{n}/n!$, $n\in\mathbb{Z}_{+}$, is the Poisson probability mass function, while the conditional (local) density function is
\begin{equation*}
  f_{J_{t}|N_{t}}(x|n)=
  \begin{cases}
    \displaystyle \delta_{\{0\}}(x) &\quad\text{if }n=0,\\
    \displaystyle \frac{\beta^{\alpha n}x^{\alpha n-1}e^{-\beta x}}{(\alpha n-1)!} &\quad\text{if }n\geq1.
  \end{cases}
\end{equation*}
For the summation starting from $n=1$, observe that
\begin{equation*}
  \sum^{\infty}_{n=1}\frac{(\lambda_{1}\beta^{\alpha}t)^{n}x^{\alpha n-1}e^{-\beta x}}{n!(\alpha n-1)!}=\frac{\lambda_{1}\beta^{\alpha} tx^{\alpha-1}e^{-\beta x}}{(\alpha-1)!} \sum^{\infty}_{n=0}\frac{(\lambda_{1}(\beta x)^{\alpha}t)^{n}}{(n+1)!(\alpha)_{\alpha n}},
\end{equation*}
where $(\cdot)_{\cdot}$ denotes the Pochhammer symbol (or rising factorial). Then, using the identity that $(n+1)(\alpha)_{\alpha n}=\alpha^{\alpha n}\prod^{\alpha}_{k=1}(1+k/\alpha)_{n}$,
\begin{equation}\label{A.1}
  \sum^{\infty}_{n=0}\frac{(\lambda_{1}(\beta x/\alpha)^{\alpha}t)^{n}}{n!\prod^{\alpha}_{k=1}(1+k/\alpha)_{n}}=:\U_{\alpha}\Bigg(1+\frac{1}{\alpha},1+\frac{2}{\alpha},\dots,2;\lambda_{1}t \bigg(\frac{\beta x}{\alpha}\bigg)^{\alpha}\Bigg)\Bigg),
\end{equation}
completing the proof.\footnote{The same formula can be derived by exploring the expansion in Withers and Nadarajah (\citeyear{WN11}) \text{Thm.} 3.1, for which one needs to evaluate the derivative of the confluent hypergeometric function with respect to its argument.}
\end{proof}

\medskip

\noindent \textbf{Proposition \ref{pro:2}.}

\begin{proof}
By the law of total probability, it is clear that the random variable $V_{1}\zeta_{2}$ admits the local density function
\begin{equation*}
  f_{V_{1}\zeta_{2}}(x)=\frac{\lambda_{2}\Delta}{\sqrt{2\pi\sigma^{2}_{V}}}e^{-(x-\mu_{V})^{2}/(2\sigma^{2}_{V})} +(1-\lambda_{2}\Delta)\delta_{\{0\}}(x), \quad x\in\mathbb{R},
\end{equation*}
and then, with the independence structure, $\mu\Delta+\sigma\xi\sqrt{\Delta}+V_{1}\zeta_{2}$ has the density function
\begin{align*}
  (f_{\mu\Delta+\sigma\xi\sqrt{\Delta}}\circ f_{V_{1}\zeta_{2}})(x)&=\frac{\lambda_{2}\Delta}{\sqrt{2\pi\upsilon^{2}\Delta}}e^{-(x-\nu\Delta)^{2}/(2\upsilon^{2}\Delta)} +\frac{1-\lambda_{2}\Delta}{\sqrt{2\pi\sigma^{2}\Delta}}e^{-(x-\mu\Delta)^{2}/(2\sigma^{2}\Delta)}\\
  &=:\lambda_{2}\Delta\ell_{1}(x)+(1-\lambda_{2}\Delta)\ell_{2}(x),\quad x\in\mathbb{R},
\end{align*}
where $\circ$ stands for convolution, and $\nu$ and $\upsilon$ are exactly as specified in the proposition.

By the chain rule, we write the local density function of the random variable $-\eta A_{1}\zeta_{1}$ as
\begin{equation}\label{A.2}
  f_{-\eta A_{1}\zeta_{1}}(x)=\frac{\lambda_{1}\Delta}{\eta}f_{A_{1}}\bigg(-\frac{x}{\eta}\bigg)+(1-\lambda_{1}\Delta)\delta_{\{0\}}(x),\quad x\leq0,
\end{equation}
where $f_{A_{1}}$ is given by (\ref{2.1.2}). Convoluting $\ell_{1}$ and $\ell_{2}$ with the point mass component in (\ref{A.2}) is straightforward; the nontrivial part is the convolution with the flipped scaled Erlang density function.

Since $\ell_{1}$ and $\ell_{2}$ are structurally identical, i.e., they are both normal densities, we will only consider $\ell_{1}$. We have
\begin{align}\label{A.3}
  \ell(x)&:=(\ell_{1}\circ f_{-\eta A_{1}})(x) \nonumber\\
  &=\int^{0}_{-\infty}\frac{1}{\sqrt{2\pi\upsilon^{2}\Delta}}\exp\bigg(-\frac{(x-y-\nu\Delta)^{2}}{2\upsilon^{2}\Delta}\bigg) \frac{\beta^{\alpha}(-y/\eta)^{\alpha-1}\exp(\beta y/\eta)}{\eta(\alpha-1)!}\dd y,\quad x\in\mathbb{R}.
\end{align}
First, we reformat (\ref{A.3}) into
\begin{equation*}
  \ell(x)=\varkappa_{1}\exp\bigg(-\frac{(x-\nu\Delta)^{2}}{2\upsilon^{2}\Delta}\bigg)\int^{\infty}_{0}y^{\alpha-1} \exp\bigg(-\frac{2(x-\nu\Delta)y+y^{2}}{2\upsilon^{2}\Delta}-\frac{\beta y}{\eta}\bigg)\dd y,
\end{equation*}
where
\begin{equation*}
  \varkappa_{1}=\frac{\beta^{\alpha}\eta^{-\alpha}}{(\alpha-1)!\sqrt{2\pi\upsilon^{2}\Delta}},
\end{equation*}
by isolating the terms free of $y$ and flipping the domain of integration. Next, we apply the substitution $y^{2}/(2\upsilon^{2}\Delta)\mapsto y$ to rewrite the last integral as
\begin{align}\label{A.4}
  &\quad\int^{\infty}_{0}\frac{(2\upsilon^{2}\Delta y)^{\alpha/2}}{2y}\exp\bigg(-y-\bigg(\frac{2(x-\nu\Delta)\sqrt{y}}{\sqrt{2\upsilon^{2}\Delta}}+\frac{\beta\upsilon\sqrt{2\Delta y}} {\eta}\bigg)\bigg)\dd y \nonumber\\
  &=\varkappa_{2}\int^{\infty}_{0}y^{\alpha/2-1}\exp(-y-2\varkappa_{3}\sqrt{y})\dd y,
\end{align}
where
\begin{equation*}
  \varkappa_{2}=\frac{(2\upsilon^{2}\Delta)^{\alpha/2}}{2}\quad\text{and}\quad\varkappa_{3}=\frac{\eta(x-\nu\Delta)+\beta\upsilon^{2}\Delta} {\eta\upsilon\sqrt{2\Delta}}.
\end{equation*}
By construction $\varkappa_{1}$ and $\varkappa_{2}$ are strictly positive, whereas the polarity of $\varkappa_{3}$ depends on $x$. Observably, the last integral in (\ref{A.4}) is precisely the Laplace transform of $h(y):=y^{\alpha/2-1}e^{-2\varkappa_{3}\sqrt{y}}$ evaluated at 1. By referring to Bateman (\citeyear{B54}) \text{Eq.} 4.5.35, it is easy to obtain (with minor modifications)
\begin{equation}\label{A.5}
  \int^{\infty}_{0}e^{-y}h(y)\dd y=2^{1-\alpha/2}(\alpha-1)!e^{\varkappa^{2}_{3}/2}\D_{-\alpha}(\sqrt{2}\varkappa_{3}).
\end{equation}
Note that for $\varkappa_{3}<0$ (\ref{A.5}) cannot be expressed in terms of a confluent hypergeometric function, but is actually an analytic continuation into the left half-plane.

We therefore obtain $\ell$ by plugging (\ref{A.5}) into (\ref{A.4}) and simplifying the coefficient $2^{1-\alpha/2}(\alpha-1)!\varkappa_{1}\varkappa_{2}>0$, eventually yielding (\ref{2.3.5}).
\end{proof}

\medskip

\noindent \textbf{Proposition \ref{pro:3}.}

\begin{proof}
In Tak\'{a}cs (\citeyear{T65}), it was shown that the distribution of the running supremum of a stochastic process with nonnegative, interchangeable increments over a generic time interval $[0,t]$ can be written as that of the process at time $t$ with a penalizing inter-temporal probability; see also Michna, Palmowski, and Pistorius (\citeyear{MPP15}). Since $J$ by construction has nonnegative increments with $\E[J_{t}]=\lambda_{1}\alpha t/\beta>0$, applying a simple scaling factor we have that for any $(t,x)\in\mathbb{R}_{++}\times\mathbb{R}_{++}$,
\begin{equation}\label{A.6}
  P(t,x)=\PP\bigg[J_{t}>x+\frac{\lambda_{1}\alpha t}{\beta}\bigg]+\int^{t}_{0}\frac{\E\big[(J_{t-s}-\lambda_{1}\alpha(t-s)/\beta)^{-}\big]} {\lambda_{1}\alpha(t-s)/\beta}\PP\bigg[J_{s}\in x+\frac{\lambda_{1}\alpha s}{\beta}+\bigg[0,\frac{\lambda_{1}\alpha\dd s}{\beta}\bigg]\bigg].
\end{equation}
Clearly, based on (\ref{2.1.3}), the first term on the right-hand side of (\ref{A.6}) is computed from
\begin{equation*}
  \PP\bigg[J_{t}\leq x+\frac{\lambda_{1}\alpha t}{\beta}\bigg]=e^{-\lambda_{1}t}+\int^{x+\lambda_{1}\alpha t/\beta}_{0}f^{(\mathrm{c})}_{J_{t}}(y)\dd y.
\end{equation*}
For the second term, note that
\begin{align*}
  \E\Bigg[\bigg(J_{t-s}-\frac{\lambda_{1}\alpha(t-s)}{\beta}\bigg)^{-}\Bigg]&=-\int^{\lambda_{1}\alpha(t-s)/\beta}_{0}\bigg(y-\frac{\lambda_{1}\alpha (t-s)}{\beta}\bigg)f_{J_{t-s}}(y)\dd y \\
  &=\frac{\lambda_{1}\alpha (t-s)e^{-\lambda_{1}(t-s)}}{\beta}+\int^{\lambda_{1}\alpha(t-s)/\beta}_{0}\bigg(\frac{\lambda_{1}\alpha (t-s)}{\beta}-y\bigg)f^{(\mathrm{c})}_{J_{t-s}}(y)\dd y,
\end{align*}
and for $x>0$,
\begin{equation*}
  \PP\bigg[J_{s}\in x+\frac{\lambda_{1}\alpha s}{\beta}+\bigg[0,\frac{\lambda_{1}\alpha\dd s}{\beta}\bigg]\bigg]=\frac{\lambda_{1}\alpha}{\beta}f^{(\mathrm{c})}_{J_{s}}\bigg(x+\frac{\lambda_{1}\alpha s}{\beta}\bigg)\dd s,
\end{equation*}
since the distribution of $J_{s}$ is absolute continuous on $\mathbb{R}_{++}$. Therefore, the integral in (\ref{A.6}) becomes
\begin{align*}
  &\quad\int^{t}_{0}\frac{\lambda_{1}\alpha (t-s)e^{-\lambda_{1}(t-s)}/\beta+\int^{\lambda_{1}\alpha(t-s)/\beta}_{0}(\lambda_{1}\alpha (t-s)/\beta-y)f^{(\mathrm{c})}_{J_{t-s}}(y)\dd y} {\lambda_{1}\alpha(t-s)/\beta}\bigg(\frac{\lambda_{1}\alpha}{\beta}f^{(\mathrm{c})}_{J_{s}}\bigg(x+\frac{\lambda_{1}\alpha s}{\beta}\bigg)\bigg)\dd s \\
  &=\frac{\lambda_{1}\alpha}{\beta}\int^{t}_{0}\mathfrak{I}_{1}(t-s)f^{(\mathrm{c})}_{J_{s}}\bigg(x+\frac{\lambda_{1}\alpha s}{\beta}\bigg)\bigg)\dd s,
\end{align*}
leading to (\ref{3.2}) and (\ref{3.3}).
\end{proof}

\medskip

\noindent \textbf{Corollary \ref{cor:1}.}

\begin{proof}
By adapted-ness, the random variable $\bar{J}_{t_{0}}$ is $\mathscr{F}_{t_{0}}$-measurable. If $x<\bar{J}_{t_{0}}$, then $\bar{J}_{t}>x$ with probability 1. Otherwise, the independent and stationary increments property of the process $J-\lambda_{1}\alpha\mathrm{id}/\beta$ renders the required conditional probability the same as
\begin{equation*}
  \PP\bigg(\bar{J}'_{t-t_{0}}>x-J_{t_{0}}+\frac{\lambda_{1}\alpha t_{0}}{\beta}\bigg|J_{t_{0}}\bigg),\quad x\geq J_{t_{0}}-\frac{\lambda_{1}\alpha t_{0}}{\beta},
\end{equation*}
where $\bar{J}'$ is an independent copy of $\bar{J}$.
\end{proof}

\medskip

\noindent \textbf{Proposition \ref{pro:4}.}

\begin{proof}
Let $0\leq t_{0}<t$ be fixed throughout. By the independence between $W^{\ast}$ and $\sum^{N_{2}}_{k=1}V_{k}$ under $\PP^{\ast}$, (\ref{3.8}) is an immediate result from setting
\begin{equation*}
  E^{(1)}_{t_{0}}(t,u):=\E^{\ast}\Big[e^{-u\int^{t}_{t_{0}}\lambda^{(1)}_{3,v}\dd v}\Big|\mathscr{F}_{t_{0}}\Big],\quad E^{(2)}_{t_{0}}(t,u):=\E\Big[e^{-u\int^{t}_{t_{0}}\lambda^{(2)}_{3,v}\dd v}\Big|\mathscr{F}_{t_{0}}\Big],\quad u\geq0.
\end{equation*}

First, the distribution of the square integral of a Brownian motion with drift is well-known; see, e.g., Xia (\citeyear{X20}) for a detailed analysis. To obtain the $\mathscr{F}_{t_{0}}$-conditional expectation, observe from the independent stationary increments of $W^{\ast}$ that
\begin{equation*}
  \int^{t}_{t_{0}}(\kappa^{(1)}s+\varsigma^{(1)}W^{\ast}_{s})^{2}\dd s=\int^{t}_{t_{0}}(\kappa^{(1)}(s-t_{0})+\varsigma^{(1)}(W^{\ast}_{s}-W^{\ast}_{t_{0}})+\vartheta_{t_{0}})^{2}\dd s=\int^{t-t_{0}}_{0}(\kappa^{(1)}s+\varsigma^{(1)}W'_{s}+\vartheta_{t_{0}})^{2}\dd s,
\end{equation*}
where $\vartheta_{t_{0}}:=\kappa^{(1)}t_{0}+\varsigma^{(1)}W^{\ast}_{t_{0}}$ is $\mathscr{F}_{t_{0}}$-measurable, while $W'$ is an independent copy of $W^{\ast}$, thus independent of $\mathscr{F}_{t_{0}}$. Note that $\lambda^{(1)}_{3,t_{0}}=\vartheta^{2}_{t_{0}}$. Hence, $E^{(1)}_{t_{0}}(t,u)$ can be treated as an unconditional expectation of the square integral of $\kappa^{(1)}\mathrm{id}+\varsigma^{(1)}W'$ starting from $\vartheta_{t_{0}}$, and the formula (\ref{3.9}) can be obtained by following the proof of Xia (\citeyear{X20}) \text{Thm.} 1, employing the Karhunen--Lo\`{e}ve transform. The only difference lies in the $\vartheta_{t_{0}}$-terms: With the fractional representations of the hyperbolic tangent and secant functions (Gradshteyn and Ryzhik (\citeyear{GR07}) \text{Eqs.} 1.421 \& 1.422), one has
\begin{equation*}
  \sum^{\infty}_{k=1}\frac{(2\sqrt{2(t-t_{0})}/(\pi(2k-1)))^{2}u}{2\varsigma^{(1)2}(2(t-t_{0})/(\pi(2k-1)))^{2}u+1} =\frac{\sqrt{u}\tanh\sqrt{2\varsigma^{(1)2}(t-t_{0})^{2}u}} {\sqrt{2\varsigma^{(1)2}}}
\end{equation*}
and
\begin{equation*}
  \sum^{\infty}_{k=1}\frac{(32(-1)^{k+1}\kappa^{(1)}(t-t_{0})^{2}/(\pi(2k-1))^{3})u}{2\varsigma^{(1)2}(2(t-t_{0})/(\pi(2k-1)))^{2}u+1} =\frac{\kappa^{(1)}(1-\sech\sqrt{2\varsigma^{(1)2}(t-t_{0})^{2}u})}{\varsigma^{(1)2}}.
\end{equation*}

Second, the SDE (\ref{3.6}) has the solution conditional on $\mathscr{F}_{t_{0}}$
\begin{equation*}
  \lambda^{(2)}_{3,t}=\lambda^{(2)}_{3,t_{0}}e^{-\kappa^{(2)}(t-t_{0})} +\varsigma^{(2)}\int^{t}_{t_{0}}e^{-\kappa^{(2)}(t-s)}\dd\sum^{N_{2,s}}_{k=1}h(V^{-}_{k}),
\end{equation*}
which implies that
\begin{align}\label{A.7}
  \int^{t}_{t_{0}}\lambda^{(2)}_{3,s}\dd s&=\lambda^{(2)}_{3,t_{0}}\frac{1-e^{-\kappa^{(2)}(t-t_{0})}}{\kappa^{(2)}} +\varsigma^{(2)}\int^{t}_{t_{0}}\int^{s}_{t_{0}}e^{-\kappa^{(2)}(s-v)}\dd\sum^{N_{2,v}}_{k=1}h(V^{-}_{k})\dd s \nonumber\\
  &=\lambda^{(2)}_{3,t_{0}}\frac{1-e^{-\kappa^{(2)}(t-t_{0})}}{\kappa^{(2)}}+\varsigma^{(2)}\int^{t}_{t_{0}} \frac{1-e^{-\kappa^{(2)}(t-s)}}{\kappa^{(2)}}\dd\sum^{N_{2,s}}_{k=1}h(V^{-}_{k}),
\end{align}
where the second equality follows from integration-by-parts. Since $\sum^{N_{2}}_{k=1}h(V^{-}_{k})$ is a L\'{e}vy process, the infinite divisibility of its distribution leads to
\begin{align*}
  \E^{\ast}\Bigg[\exp\Bigg(-u\varsigma^{(2)}\int^{t}_{t_{0}}\frac{1-e^{-\kappa^{(2)}(t-s)}}{\kappa^{(2)}} \dd\sum^{N_{2,s}}_{k=1}h(V^{-}_{k})\Bigg)\Bigg|\mathscr{F}_{t_{0}}\Bigg]&=\prod^{t}_{t_{0}} \E^{\ast}\Big[e^{-u((1-e^{-\kappa^{(2)}(t-s)})/\kappa^{(2)})\sum^{N_{2,1}}_{k=1}h(V^{-}_{k})}\Big|\mathscr{F}_{t_{0}}\Big]^{\dd s}\\
  &=\exp\int^{t}_{t_{0}}\log\phi_{\sum^{N_{2,1}}_{k=1}h(V^{-}_{k})}\bigg(\frac{\ii u(1-e^{-\kappa^{(2)}s})}{\kappa^{(2)}}\bigg)\dd s,
\end{align*}
where $\prod^{\cdot}_{\cdot}$ denotes the geometric integral (Slav\'{i}k (\citeyear{S07})) and which combined with (\ref{A.7}) yields the expression (\ref{3.10}) for $E^{(2)}_{t_{0}}(t,u)$.
\end{proof}

\medskip

\noindent \textbf{Corollary \ref{cor:2}.}

\begin{proof}
If $h$ is given by (\ref{3.11}), then $h(V^{-}_{1})$ follows a categorical distribution (\text{a.k.a.} generalized Bernoulli distribution) with exactly four categories, separated by multiples of $-\sigma_{V}$. Specifically, the compound random variable $\sum^{N_{2,1}}_{k=1}h(V^{-}_{k})$ is easily seen to have the log-characteristic function
\begin{align}\label{A.8}
  \log\phi_{\sum^{N_{2,1}}_{k=1}h(V^{-}_{k})}(u)&=\log\E\Big[e^{\ii u\sum^{N_{2,1}}_{k=1}h(V^{-}_{k})}\Big] \nonumber\\
  &=\lambda_{2}\Bigg(e^{4\ii u}\PP[V^{-}_{1}\geq3\sigma_{V}]+\sum^{3}_{i=1}e^{\ii iu}\PP[i\sigma_{V}>V^{-}_{1}\geq(i-1)\sigma_{V}]-1\Bigg),\quad u\in\mathbb{R},
\end{align}
where the coefficients
\begin{align*}
  &a_{1}:=\PP[\sigma_{V}>V^{-}_{1}\geq0]=\PP[V_{1}>-\sigma_{V}],\; a_{2}:=\PP[2\sigma_{V}>V^{-}_{1}\geq\sigma_{V}]=\PP[-2\sigma_{V}<V_{1}\leq-\sigma_{V}],\\
  &a_{3}:=\PP[3\sigma_{V}>V^{-}_{1}\geq2\sigma_{V}]=\PP[-3\sigma_{V}<V_{1}\leq-2\sigma_{V}],\;
  a_{4}:=\PP[V^{-}_{1}\geq3\sigma_{V}]=\PP[V_{1}\leq-3\sigma_{V}]
\end{align*}
are directly computed using the normal cumulative distribution function, as shown in (\ref{3.13}).

With (\ref{A.8}), it then remains to compute
\begin{align*}
  \int^{t}_{t_{0}}\exp\bigg(\frac{-iu\varsigma^{(2)}(1-e^{-\kappa^{(2)}s})}{\kappa^{(2)}}\bigg)\dd s&=e^{-iu\varsigma^{(2)}/\kappa^{(2)}} \int^{-iu\varsigma^{(2)}e^{-\kappa^{(2)}t_{0}}/\kappa^{(2)}}_{-iu\varsigma^{(2)}e^{-\kappa^{(2)}t}/\kappa^{(2)}}\frac{e^{-s}}{\kappa^{(2)}s}\dd s\\
  &=:\frac{e^{-iu\varsigma^{(2)}/\kappa^{(2)}}}{\kappa^{(2)}}\bigg(\Ei\bigg(\frac{iu\varsigma^{(2)}e^{-\kappa^{(2)}t_{0}}}{\kappa^{(2)}}\bigg) -\Ei\bigg(\frac{iu\varsigma^{(2)}e^{-\kappa^{(2)}t}}{\kappa^{(2)}}\bigg)\bigg),
\end{align*}
with $i\in\{1,2,3,4\}$, where the substitution $-iu\varsigma^{(2)}e^{-\kappa^{(2)}s}\mapsto s$ has been used.
\end{proof}

\medskip

\noindent \textbf{Theorem \ref{thm:1}.}

\begin{proof}
For the non-default leg of the CoCo, using (\ref{4.1}), taking expectation on the discounted cash flows yields
\begin{align*}
  \E^{\ast}\Big[Ke^{-\int^{T}_{t_{0}}r_{v}\dd v}\mathbf{1}_{\{\tau>T\}}\big|\mathscr{F}_{t_{0}}\big]&=Ke^{-\int^{T}_{t_{0}}r_{v}\dd v} \PP^{\ast}\Big[\bar{J}_{T}\leq\overline{J}\big|\mathscr{F}_{t_{0}}\big]\PP^{\ast}\big[N_{3,\Lambda_{T}}=0\big|\mathscr{F}_{t_{0}}\big]\\
  &=Ke^{-\int^{T}_{t_{0}}r_{v}\dd v}(1-P_{t_{0}}(T,\overline{J}))E^{\ast}_{t_{0}}(T,1),
\end{align*}
and likewise, for each $i\in\mathbb{Z}\cap[1,M]$, since each coupon date is deterministic, then provided $t_{0}<t_{i}$,
\begin{equation*}
  \E^{\ast}\Big[c_{i}e^{-\int^{t_{i}}_{t_{0}}r_{v}\dd v}\mathbf{1}_{\{\tau>t_{i}\}}\Big|\mathscr{F}_{t_{0}}\Big]=c_{i}e^{-\int^{t_{i}}_{t_{0}}r_{v}\dd v} \big(1-P_{t_{0}}\big(t_{i},\overline{J}\big)\big)E^{\ast}_{t_{0}}(t_{i},1).
\end{equation*}
For the default leg, using (\ref{4.1.1}) and the given independence structure, we apply the law of iterated expectations to deduce that
\begin{align}\label{A.9}
  \E^{\ast}\Big[RKe^{-\int^{\tau}_{t_{0}}r_{v}\dd v}\mathbf{1}_{\{t_{0}<\tau\leq T\}}\Big|\mathscr{F}_{t_{0}}\Big]&=\E^{\ast}\Big[\varpi(1-w)Ke^{-\int^{\tau_{1}}_{t_{0}}r_{v}\dd v}\mathbf{1}_{\{t_{0}<\tau_{1}\leq T\wedge\tau_{3}\}}\Big|\mathscr{F}_{t_{0}}\Big] \nonumber\\
  &=\varpi(1-w)K\E^{\ast}\Big[e^{-\int^{\tau_{1}}_{t_{0}}r_{v}\dd v-\int^{\tau_{1}}_{t_{0}}\lambda_{3,v}\dd v}\mathbf{1}_{\{t_{0}<\tau_{1}\leq T\}}\Big|\mathscr{F}_{t_{0}}\Big] \nonumber\\
  &=\varpi(1-w)K\E^{\ast}\Big[e^{-\int^{\tau_{1}}_{t_{0}}r_{v}\dd v}\mathbf{1}_{\{t_{0}<\tau_{1}\leq T\}}E^{\ast}_{t_{0}}(\tau_{1},1)\Big|\mathscr{F}_{t_{0}}\Big] \nonumber\\
  &=\varpi(1-w)K\int_{\mathbb{R}_{+}}e^{-\int^{s}_{t_{0}}r_{v}\dd v}\mathbf{1}_{\{s\in(t_{0},T]\}}E^{\ast}_{t_{0}}(s,1) \dd_{s}\PP\big[\bar{J}_{s}>\overline{J}\big|\mathscr{F}_{t_{0}}\big] \nonumber\\
  &=\varpi(1-w)\int^{T}_{t_{0}}e^{-\int^{s}_{t_{0}}r_{v}\dd v}E^{\ast}_{t_{0}}(s,1)\dd_{s}P_{t_{0}}(s,\overline{J}),
\end{align}
where the differential is with respect to the time variable $s$.
\end{proof}

\medskip

\noindent \textbf{Theorem \ref{thm:2}.}

\begin{proof}
The non-default leg is no different from that of the write-down CoCo, as in Theorem \ref{thm:1}. Based on the stock price evolution under $\PP^{\ast}$ in (\ref{2.3.8}) and the definition (\ref{4.2.1}) of the conversion ratio, the default leg has time-$t_{0}$ value
\begin{align}\label{A.10}
  \E^{\ast}\big[\chi_{\tau}S_{\tau}e^{-\int^{\tau}_{t_{0}}r_{v}\dd v}\mathbf{1}_{\{\tau\leq T\}}\big|\mathscr{F}_{t_{0}}\big]&=\E^{\ast}\bigg[\frac{KR S^{p}_{\tau}}{S^{p}_{0}}e^{-\int^{\tau}_{t_{0}}r_{v}\dd v}\mathbf{1}_{\{\tau\leq T\}}\bigg|\mathscr{F}_{t_{0}}\bigg] \nonumber\\
  &=\frac{\varpi(1-w)K}{S^{p}_{0}}\E^{\ast}\Big[S^{p}_{\tau_{1}}e^{-\int^{\tau_{1}}_{t_{0}}r_{v}\dd v}\mathbf{1}_{\{\tau_{1}\leq T\wedge\tau_{3}\}}\Big|\mathscr{F}_{t_{0}}\Big].
\end{align}
The case $p=0$ is the same as (\ref{A.9}), so assume $p>0$ in the following. For $t\geq0$, the powered stock price process is given by
\begin{align*}
  S^{p}_{t}&=S^{p}_{0}\exp\Bigg(\int^{t}_{0}p\bigg(r_{s}-q-\frac{1}{2}\sigma^{2}-\lambda_{1}\psi_{1}(1)-\lambda_{2}\psi_{2}(1) +\gamma\lambda_{3,s}\bigg)\dd s \\
  &\qquad+p\sigma W^{\ast}_{t}+p\sum^{N_{2,t}}_{k=1}V_{k}-p\eta J_{t}+p\log(1-\gamma)N_{3,\Lambda_{t}}\Bigg).
\end{align*}
Clearly, by taking $\tilde{q}$ as the solution to
\begin{align*}
  &\quad p\bigg(r-q-\frac{1}{2}\sigma^{2}-\lambda_{1}\psi_{1}(1)-\lambda_{2}\psi_{2}(1)+\gamma\lambda_{3}\bigg)\\
  &=r-\tilde{q}-\frac{1}{2}(p\sigma)^{2} -\lambda_{1}\psi_{1}(p)-\lambda_{2}\psi_{2}(p)-((1-\gamma)^{p}-1)\lambda_{3},
\end{align*}
or
\begin{align*}
  \tilde{q}&=pq+(1-p)r+\frac{1}{2}p(1-p)\sigma^{2}+\lambda_{1}(p\psi_{1}(1)-\psi_{1}(p))+\lambda_{2}(p\psi_{2}(1)-\psi_{2}(p)) \\
  &\qquad-((1-\gamma)^{p}-1+p\gamma)\lambda_{3},
\end{align*}
we can rewrite $S^{p}$ as
\begin{align*}
  S^{p}_{t}&=S^{p}_{0}\exp\Bigg(\int^{t}_{0}\bigg(r_{s}-\tilde{q}_{s}-\frac{1}{2}(p\sigma)^{2}-\lambda_{1}\psi_{1}(p)-\lambda_{2}\psi_{2}(p) -((1-\gamma)^{p}-1)\lambda_{3,s}\bigg)\dd s \\
  &\qquad+p\sigma W^{\ast}_{t}+\sum^{N_{2,t}}_{k=1}(pV_{k})-p\eta J_{t}+p\log(1-\gamma)N_{3,\Lambda_{t}}\Bigg).
\end{align*}
The powered stock price can be seen as the price of a fictitious stock with initial price $S^{p}_{0}$, dividend yield process $\tilde{q}$, and scaled sources of randomness $p\sigma W^{\ast}+\sum^{N_{2,t}}_{k=1}(pV_{k})-p\eta J$, where $pV_{1}$ is normally distributed with mean $p\mu_{V}$ and variance $p^{2}\sigma^{2}_{V}$. By the definition of $\psi_{1}$ and $\psi_{2}$, it is also immediate that $\tilde{q}=r$ if $p=0$ and $\tilde{q}=q$ if $p=1$. Thus, we can define another equivalent probability measure $\tilde{\PP}\sim\PP^{\ast}$ using this fictitious stock as a num\'{e}raire, through the density process
\begin{align}\label{A.11}
  \frac{\dd\tilde{\PP}}{\dd\PP^{\ast}}\bigg|_{\mathscr{F}_{t}}&=\exp\Bigg(-\int^{t}_{0}\bigg(\frac{1}{2}(p\sigma)^{2}+\lambda_{1}\psi_{1}(p) +\lambda_{2}\psi_{2}(p)+((1-\gamma)^{p}-1)\lambda_{3,s}\bigg)\dd s \nonumber\\
  &\qquad+p\sigma W^{\ast}_{t}+\sum^{N_{2,t}}_{k=1}(pV_{k})-p\eta J_{t}+p\log(1-\gamma)N_{3,\Lambda_{t}}\Bigg).
\end{align}
According to the Girsanov--Meyer theorem for general semimartingales (see Jacod and Shiryaev (\citeyear{JS10}) \text{Chap.} III \text{Thm.} 7.18 \& 7.23]), it is straightforward to check that, under $\tilde{\PP}$, $W^{\ast}-p\sigma\mathrm{id}$ is a standard Brownian motion, $N_{2}$ has intensity $\lambda_{2}(\psi_{2}(p)+1)=\lambda_{2}e^{p\mu_{V}+(p\sigma_{V})^{2}/2}$, $V_{k}$'s are normally distributed with mean $\mu_{V}+p\sigma^{2}$ and variance $\sigma^{2}_{V}$, $N_{1}$ has intensity $\lambda_{1}(\psi_{1}(p)+1)=\lambda_{1}(1+p\eta/\beta)^{-\alpha}$, $A_{k}$'s are Erlang-distributed with parameters $\alpha$ (shape) and $\beta+p\eta$ (rate), and $N_{3}$ (without time inhomogeneity) has intensity $(1-\gamma)^{p}$,\footnote{The intuition behind the change in $N_{3}$ is as follows: The (forward) measure $\tilde{\PP}$ (partially) embeds risk associated with regulatory intervention into stock price reactions, and the updated intensity should complement the corresponding downward return magnitude as controlled by the parameter $\gamma$. In particular, the intensity reduces to 0 when $\gamma=1$ (bankruptcy), with the intervention risk fully absorbed into the stock price drop, and it stays at 1 when $\gamma=0$ (futility), in which case the stock price is completely unresponsive to intervention.} which results are also easily justifiable by analyzing the corresponding $\tilde{\PP}$-characteristic functions.

Then, by defining $\tilde{q}^{\circ}:=\tilde{q}+((1-\gamma)^{p}-1+p\gamma)\lambda_{3}$, which is deterministic as shown in (\ref{4.2.3}), and evaluating (\ref{A.10}) using the law of iterated expectations in a similar fashion as in (\ref{A.9}), we have that
\begin{align*}
  &\quad\E^{\ast}\big[\chi_{\tau}S_{\tau}e^{-\int^{\tau}_{t_{0}}r_{v}\dd v}\mathbf{1}_{\{t_{0}<\tau\leq T\}}\big|\mathscr{F}_{t_{0}}\big] \\
  &=\frac{\varpi(1-w)K}{S^{p}_{0}}\E^{\ast}\Big[S^{p}_{\tau}e^{-\int^{\tau}_{t_{0}}r_{v}\dd v}\mathbf{1}_{\{t_{0}<\tau_{1}\leq T\wedge\tau_{3}\}}\Big|\mathscr{F}_{t_{0}}\Big] \\
  &=\frac{\varpi(1-w)K}{S^{p}_{0}}\tilde{\E}\Big[S^{p}_{t_{0}}e^{-\int^{\tau}_{t_{0}}\tilde{q}_{v}\dd v}\mathbf{1}_{\{t_{0}<\tau_{1}\leq T\wedge\tau_{3}\}}\Big|\mathscr{F}_{t_{0}}\Big] \\
  &=\varpi(1-w)K\bigg(\frac{S_{t_{0}}}{S_{0}}\bigg)^{p}\tilde{\E}\Big[e^{-\int^{\tau_{1}}_{t_{0}}\tilde{q}^{\circ}_{v}\dd v} e^{\int^{\tau_{1}}_{t_{0}}((1-\gamma)^{p}-1+p\gamma-(1-\gamma)^{p})\lambda_{3,v}\dd v}\mathbf{1}_{\{t_{0}<\tau_{1}\leq T\}}\Big|\mathscr{F}_{t_{0}}\Big] \\
  &=\varpi(1-w)K\bigg(\frac{S_{t_{0}}}{S_{0}}\bigg)^{p}\tilde{\E}\Big[e^{-\int^{\tau_{1}}_{t_{0}}\tilde{q}^{\circ}_{v}\dd v}\tilde{E}_{t_{0}}(\tau_{1},1-p\gamma)\mathbf{1}_{\{t_{0}<\tau_{1}\leq T\}}\Big|\mathscr{F}_{t_{0}}\Big] \\
  &=\varpi(1-w)K\bigg(\frac{S_{t_{0}}}{S_{0}}\bigg)^{p}\int^{T}_{t_{0}}e^{-\int^{s}_{t_{0}}\tilde{q}^{\circ}_{v}\dd v}\tilde{E}_{t_{0}}(s,1-p\gamma)\dd_{s}\tilde{P}_{t_{0}}(s,\overline{J}),
\end{align*}
where the last line is well-defined because $1-p\gamma\in[0,1]$. This concludes the proof.
\end{proof}

\medskip

\noindent \textbf{Corollary \ref{cor:3}.}

\begin{proof}
This follows immediately from Theorem \ref{thm:2} by setting $p=1$, with $\tilde{q}\equiv\tilde{q}^{\circ}=q$.
\end{proof}

\bigskip

\section{Some details for power conversion mechanism}\label{B}

\renewcommand{\theequation}{B.\arabic{equation}}

Consider the equity-convertible CoCo in Section \ref{sec:4.2}, for which assume that the default time $\tau$ (given in (\ref{4.1})) satisfies that $\tau_{1}\leq T<\tau_{3}$, ensuring an accounting-triggered default. The following explanation is for a conversion price taking the highly nonstandard form
\begin{equation}\label{B.1}
  C_{\tau}:=\max\bigg\{\frac{1}{\Delta}\int^{\tau}_{\tau-\Delta}S_{s}\dd s,\underline{S}\bigg\},
\end{equation}
where $\underline{S}$ is a preset price floor, and $\Delta>0$ here is a fixed small observation window preceding default, typically 30 calendar days ($\Delta=1/12$).

As noted in Wilkens and Bethke (\citeyear{WB14}), the temporal averaging in (\ref{B.1}) mainly serves to smooth out spurious jumps in the stock price immediately before conversion, and if the average price happens to significantly deviate from $S_{\tau}$, then it is most likely above the price floor $\underline{S}$, and
\begin{equation*}
  C_{\tau}\approx\frac{1}{\Delta}\int^{\tau}_{\tau-\Delta}S_{s}\dd s\approx S^{p}_{\tau},
\end{equation*}
where the second approximation, with the power coefficient $p\in(0,1)$, is valid given that $C_{\tau}<S_{\tau}$ with high probability, by choosing the write-down fraction $w\in[0,1]$ such that $(1-w)K=S^{p}_{0}$ in (\ref{4.2.1}).

On the other hand, when the average price is close to $S_{\tau}$, then it follows that (with high probability)
\begin{equation*}
  C_{\tau}\approx\max\{S_{\tau},\underline{S}\}\approx S^{p}_{\tau},
\end{equation*}
which implies the choice $p\approx0$ if $S_{\tau}\leq\underline{S}$, by setting $(1-w)K=\underline{S}$; otherwise, $p\approx1$ and $w$ is chosen such that $(1-w)K=S_{0}$.



\end{appendices}

\end{document}